\documentclass[aps,pra,onecolumn,nofootinbib,superscriptaddress]{revtex4-2}

\usepackage{amsmath}
\usepackage{amssymb}
\usepackage{bm}
\usepackage{physics}
\usepackage{array}
\usepackage{booktabs}
\usepackage{graphicx}
\usepackage{hyperref}
\usepackage[capitalize]{cleveref}
\usepackage{xcolor,graphicx} 
\usepackage{amsmath, amssymb, bm, mathtools, array}
\usepackage{amsthm}
\usepackage{amsthm}
\usepackage{physics,geometry}
\usepackage{tikz}
\usepackage{quantikz}
\usepackage{comment}
\usetikzlibrary{backgrounds}
\usepackage{fullpage,physics}
\usepackage{scalerel}
\usepackage{stackengine,wasysym}
\usepackage{algorithmicx}
\usepackage{placeins}

\usepackage[linesnumbered,ruled,vlined]{algorithm2e}
\usepackage{amsmath,amsthm,amssymb,amsfonts,bm}
\newtheorem{theorem}{Theorem}[section]

\newtheorem{lemma}{Lemma}
{\bf}{\it}
\newtheorem{remark}[theorem]{Remark}

\newtheorem{Prop}[theorem]{Proposition}

\newcommand{\mat}[1]{\begin{pmatrix}#1\end{pmatrix}}
\newcommand{\F}{\mathcal{F}}
\newcommand{\Rz}{\mathrm{R}_z}
\newcommand{\Rx}{\mathrm{R}_x}
\newcommand{\Ry}{\mathrm{R}_y}
\newcommand{\Had}{\mathrm{H}}
\newcommand{\Kset}{\mathcal{K}}

\begin{document}

\title{Structure-Preserving Quantum Simulation of Wave Equations on a Trapped-Ion Processor}

\author{Abhishek Shringi}
\affiliation{Department of Computer Science and Engineering, Pennsylvania State University, University Park, PA 16802, USA}
\author{Hsuan-Cheng Wu}
\affiliation{Department of Mathematics, Pennsylvania State University, University Park, PA 16802, USA}
\author{Ahmed Shokry}
\affiliation{Department of Computer Science and Engineering, Pennsylvania State University, University Park, PA 16802, USA}
\author{Xiantao Li}
\affiliation{Department of Mathematics, Pennsylvania State University, University Park, PA 16802, USA}
\author{Mahmut Taylan Kandemir}
\affiliation{Department of Computer Science and Engineering, Pennsylvania State University, University Park, PA 16802, USA}
\date{\today}

\begin{abstract}

Wave equations provide a natural testbed for near-term quantum simulation of
partial differential equations, but hardware demonstrations have remained
limited in spatial dimension, equation class, system size, and physically
meaningful output. We develop and benchmark structure-preserving,
Fourier-based quantum circuits for the one- and two-dimensional acoustic wave
equations and Dirac dynamics with variable mass on the
Quantinuum H2-2 trapped-ion processor. A Hermitian first-order finite-difference
formulation and the quantum Fourier transform reduce the one-dimensional
acoustic evolution to mode-controlled rotations. In two dimensions, we
introduce an exact bright--dark decomposition that isolates the longitudinal
velocity coupled to flux variables and eliminates the product-formula error of
directional splitting. For the Dirac dynamics equation with a
spatially discontinuous mass, we combine the Fourier-space wave propagator with
position-space mass rotations through Strang splitting, extending the
hardware study to inhomogeneous dispersive dynamics. The experiments include
one-dimensional grids with up to \(1024\) points and \(32\times32\)
two-dimensional grids, corresponding to an encoded state-space dimension of up
to \(4096\). Rather than reconstructing the full fields, we estimate subdomain
kinetic energies directly from measurement samples. Across all tested acoustic
and Dirac dynamics problems, the H2-2 results track the classical kinetic-energy
dynamics with mean absolute errors between \(5.9\times10^{-3}\) and
\(2.4\times10^{-2}\). At fixed retained bandwidth, the compiled gate counts
grow approximately quadratically with the number of grid qubits; the acoustic
circuit sizes are essentially independent of evolution time, whereas the
cost also grows with the number of product-formula steps. These results provide
hardware-level evidence that accurate observable dynamics can remain
resolvable for structured wave problems with thousands of encoded degrees of
freedom on a present-day trapped-ion processor.

\end{abstract}

\maketitle



\section{Introduction}\label{sec:intro}
Partial differential equations (PDEs) lie at the heart of scientific
computing: they model acoustic and electromagnetic wave propagation, fluid
flow, heat transport, and quantum dynamics, and their discretizations generate
many of the largest computational problems solved today. This makes PDEs a
natural and demanding benchmark for quantum computation. In an
amplitude-encoded representation, a field sampled at
$N_h=2^{n_h}$ grid points can be stored using $n_h$ grid qubits. This
representational compression alone does not imply a computational advantage:
state preparation, time evolution, measurement, compilation overhead, and
hardware noise must all be taken into account.

The idea that quantum systems could efficiently simulate other quantum systems
was proposed by Feynman~\cite{feynman1982simulating} and put on an algorithmic
footing by Lloyd~\cite{lloyd1996universal}. Algorithms for linear differential
equations and sparse linear systems subsequently supplied important primitives
for quantum PDE solvers~\cite{berry2010high,harrow2009quantum}. Favorable
asymptotic bounds for such methods, however, typically rely on fault-tolerant
computation together with efficient access to the discretized operator, the
initial data, and the desired output. Explicit circuit constructions and
hardware experiments are therefore needed to determine which parts of this
potential survive on present-day devices.

On noisy intermediate-scale quantum (NISQ) processors~\cite{preskill2018quantum},
application performance is constrained by gate and readout errors as well as
by the overhead introduced during hardware-native compilation. Best-reported
two-qubit gate fidelities now exceed $99.99\%$ in trapped-ion systems and
$99.9\%$ in superconducting devices
~\cite{hughes2025trappedions9999,Marxer_2026}, but component-level fidelities
alone do not determine the accuracy of a compiled application circuit. This
motivates an empirical assessment based on explicit circuits, native-gate
resources, and physical-device measurements. Here we carry out such an
assessment for structured wave-propagation problems, varying spatial
dimension, initial-state complexity, grid size, and evolution time.

Our primary model is the acoustic wave equation, a canonical hyperbolic PDE
with an energy-conserving first-order formulation and applications including
seismic modeling and imaging~\cite{wen2025seismic}. Introducing pressure and
velocity variables converts the second-order equation into a
Schr\"odinger-type system. A compatible finite-difference discretization
preserves Hermiticity at the semi-discrete level. Under periodic boundary
conditions, the difference matrices are circulant, so the quantum Fourier
transform (QFT) decomposes the dynamics into independent Fourier-mode blocks.

This structure is well matched to trapped-ion hardware. An exact QFT on
$n_h$ qubits contains $O(n_h^2)$ controlled-phase gates, many acting between
non-neighboring logical qubits. Quantinuum's quantum charge-coupled-device
architecture provides effective all-to-all connectivity by transporting ions
between interaction zones~\cite{moses2023racetrack}, thereby avoiding the
explicit SWAP networks that may be required on architectures with fixed local
connectivity. Connectivity is only one contributor to performance, so we
report both logical and H2-2-native gate counts and depths and evaluate the
compiled circuits on the H2-2E emulator and the H2-2 processor.

For the one-dimensional acoustic wave equation, we adopt the low-mode QFT
strategy of Wright \emph{et al.}~\cite{wright2024noisy}: the finite-difference
dispersion relation is linearized on the Fourier modes occupied by a smooth
input, and the resulting mode dependence is synthesized by controlled
rotations. We extend the hardware benchmark to structured cosine and truncated
Gaussian inputs on grids as large as $N_h=1024$, and we independently scan the
evolution time and the number of grid qubits to compare logical and
hardware-native resources.

In two dimensions, pressure couples simultaneously to two velocity
components. Prior Fourier-space constructions implement the corresponding
smooth-field propagator through a product formula for noncommuting directional
terms~\cite{lubasch2025fourier}. We instead rotate each Fourier block into a
longitudinal ``bright'' velocity component, which couples to pressure, and a
transverse ``dark'' component, which decouples. This is a
Morris--Shore-type bright--dark reduction
~\cite{morris1983reduction}; our contribution is to identify and exploit it in
the two-dimensional acoustic Fourier block. Within the constant-coefficient,
periodic semi-discrete model, the reduction is exact and removes the
directional product-formula error. For the present hardware experiments, the
required mode-dependent angles are compiled for a fixed low-frequency set.

We also consider a one-dimensional Dirac equation with a stepwise mass
profile, a simple heterogeneous test related to effective variable-mass Dirac
models~\cite{unanyan2010spinor}. The kinetic generator is diagonal in Fourier
space, whereas the mass generator is diagonal in position space. We combine
them by Strang splitting, so that each step alternates position-controlled mass
rotations with the QFT-based kinetic propagator. This benchmark exposes the
additional resource growth with the number of product-formula steps.

The main contributions of this work are:
\begin{itemize}
  \item A common, structure-preserving finite-difference and QFT framework for
  one- and two-dimensional acoustic waves and one-dimensional variable-mass
  Dirac dynamics;
  \item An exact bright--dark reduction of each two-dimensional acoustic
  Fourier block, eliminating directional product-formula error for the
  semi-discrete model considered here;
  \item A resource benchmark under independent scans of evolution time and
  grid-qubit count, reporting logical and H2-2-native total and two-qubit gate
  counts and depths; and
  \item Physical H2-2 executions for all three model problems, comparing
  classical references, ideal finite-shot sampling, H2-2E where available,
  and H2-2 measurements of half-domain marginal observables---kinetic-energy
  fractions for the acoustic problems and a spinor-component population for
  the Dirac problem.
\end{itemize}

The experiments reveal a consistent, quantified picture. At fixed retained
bandwidth, empirical fits over the range studied show an approximately
quadratic compiled gate count and approximately linear compiled depth in
$n_h$ for the one-dimensional acoustic circuit, while both quantities are
approximately quadratic for the two-dimensional circuit. Across the H2-2
experiments, the mean absolute error in the chosen marginal observable ranges
from $5.9\times10^{-3}$ to $2.4\times10^{-2}$. The largest error occurs for
the deeper split-operator Dirac circuits. These results benchmark structured
inputs and observable-level outputs; they do not establish a quantum advantage
over classical PDE solvers.

The remainder of the paper is organized as follows.
\Cref{sec:related} positions the work among quantum wave-equation algorithms
and hardware demonstrations. \Cref{sec:experimental-setup} describes the
common classical-reference, simulation, compilation, and hardware pipeline,
and \cref{sec:disc} derives the structure-preserving discretizations. The
one-dimensional QFT construction and its low-frequency error bound are given
in \cref{sec:permode}. \Cref{sec:methods2d} develops the two-dimensional
bright--dark reduction, and \cref{sec:kg-methods} presents the Strang-split
Dirac circuit. Resource and hardware results are reported within the
corresponding equation-specific sections. Finally, \cref{sec:discussion}
discusses implications, limitations, and directions for future work; the
appendix gives detailed state-preparation and restricted-mode circuit
constructions.


\section{Related work}\label{sec:related}
Quantum algorithms relevant to wave propagation differ in how they encode the PDE, implement time evolution, and extract output. Qubit count, oracle complexity, logical gate count, hardware-native depth, product-formula steps, and measurement cost are therefore not interchangeable resource measures. We organize representative work below by circuit construction and hardware target and summarize the associated asymptotic results in~\cref{tab:scaling}.

A first family maps the discretized wave equation to Hamiltonian-simulation and quantum-linear-system primitives. Costa \emph{et
al.}~\cite{costa2019quantum} factorized finite-difference Laplacians to obtain a first-order Schr\"odinger representation under Dirichlet and Neumann boundary conditions. Their algorithm uses Hamiltonian simulation and a quantum linear-system subroutine, and they also discuss extensions to the Klein--Gordon and Maxwell equations. The resulting complexity statements are oracle based and assume efficient state preparation and output access; the work is an asymptotic algorithm rather than a NISQ circuit demonstration. Suau \emph{et al.}~\cite{suau2021practical} subsequently analyzed an explicit implementation of this approach. For their one-dimensional benchmark and a first-order product formula, they obtained a gate-count scaling.
\begin{equation}
  O\!\left(
    \frac{N_h^{3/2}\log^2N_h\,t^{3/2}}{\sqrt{\varepsilon}}
  \right),
\end{equation}
and concluded that error-corrected hardware would be required for direct 
implementation. More generally, Childs, Liu, and
Ostrander~\cite{childs2021highprecision} developed high-precision algorithms for linear PDEs using adaptive-order finite differences and spectral methods together with quantum linear-system algorithms. Those results provide important fault-tolerant context, but primarily address elliptic boundary-value problems and a different output model from the time-domain hardware experiments studied here.

A second family constructs explicit Hamiltonian-simulation circuits for
finite-difference discretizations. Sato \emph{et
al.}~\cite{sato2024hamiltonian} gave circuits for linear hyperbolic PDEs on an $N_h^d=2^{dn_h}$ lattice using $dn_h$ qubits, together with numerical tests and a small real-device wave-equation experiment. They report $O(dn_h^3T^2/\varepsilon)$ nonlocal gates, with an alternative construction using $O(dn_h^{2.5}T^{1.5}/\varepsilon^{0.5})$ nonlocal gates. Arseniev \emph{et al.}~\cite{arseniev2024high} decomposed banded
finite-difference matrices into mutually commuting groups of Pauli strings and studied higher-order stencils. Their one-dimensional wave experiments show that higher-order discretization can reduce the grid-qubit count required for a target spatial accuracy but does not, by itself, reduce the number of Trotter steps required to preserve the overall solution accuracy. In a closely related geophysical setting, Schade \emph{et al.}~\cite{schade2024elastic} used a sparsity-preserving transformation of the one-dimensional heterogeneous elastic wave equation to Schr\"odinger form. Their IBM Brisbane demonstration used a seven-point grid and achieved qualitative agreement with the noiseless calculation while also illustrating the strong effect of hardware noise at large circuit depth.

A third family works directly in Fourier space, as we do here. Wright
\emph{et al.}~\cite{wright2024noisy} constructed a low-mode QFT circuit for the one-dimensional acoustic wave equation and executed it on Quantinuum H1-1. Their evolution circuit has depth $O(n_h^2)$, and the small-angle linearization gives an infidelity
$O(2^{-4n_h}t^2)=O(t^2/N_h^4)$ for smooth, effectively low-bandwidth initial data. Their physical-device experiment used a variationally prepared Ricker wavelet with $n_h=6$ grid qubits ($N_h=64$); their numerical state-preparation and resource studies extended to $n_h=10$ grid qubits. Our one-dimensional benchmark retains the same low-mode mechanism but uses analytically specified cosine and truncated Gaussian Fourier inputs, executes $n_h=10$ ($N_h=1024$) grid-point instances on H2-2, and reports logical and H2-2-native resources under independent scans in $t$ and in $n_h$, extending the latter to $n_h=50$.

Lubasch \emph{et al.}~\cite{lubasch2025fourier} developed a broader QFT and
QSVT framework for the incompressible-advection, heat, isotropic-acoustic-wave,
and Poisson equations. For the acoustic wave equation, they distinguish a
general construction of depth
\begin{equation}
  \widetilde O\!\left(
    dn_h\left[tN_h+\log(1/\varepsilon)\right]
  \right)
\end{equation}
from a smooth-field construction of depth
\begin{equation}
  \widetilde O\!\left(
    \frac{d^3k_{\max}^2n_ht^2}{\varepsilon}
  \right),
\end{equation}
excluding the QFT depth. In more than one dimension, the latter construction
uses a first-order product formula for noncommuting directional Hamiltonians.
Our bright--dark construction is not a replacement for that general
framework; it exploits the special rank-one pressure--velocity coupling in the
constant-coefficient, periodic, two-dimensional acoustic system. For this
model, it diagonalizes each Fourier block exactly and therefore removes that
directional product-formula error before hardware compilation.

Other near-term differential-equation strategies use basis-encoded arithmetic, quantum annealing, or variational function learning rather than the coherent Hamiltonian evolution of an amplitude-encoded field. Zanger \emph{et al.}~\cite{zanger2021quantum} investigated fixed-point digital circuits for classical time-integration formulas applied to small ordinary differential equations, as well as a D-Wave implementation of a Runge--Kutta formulation. Schillo and Sturm~\cite{schillo2025variational} combined quantum circuit
learning with parameter-shift differentiation and demonstrated small
function-learning and differential-equation examples on IBM hardware. These methods address different encodings, outputs, and resource tradeoffs, so their counts are not directly comparable with the QFT Hamiltonian-simulation circuits considered here.

A methodologically distinct line of hardware work simulates relativistic wave equations through analog trapped-ion or small digital encodings. Gerritsma \emph{et al.} simulated the free Dirac equation and
Zitterbewegung~\cite{gerritsma2010dirac}, and later the Klein
paradox~\cite{gerritsma2011klein}, using internal and motional degrees of
freedom of trapped ions. Kapil \emph{et al.}~\cite{kapil2018quantum} proposed a two-qubit digital circuit for a two-component Klein--Gordon representation and reported local-simulator results using the IBM/Qiskit framework. Our relativistic benchmark instead uses a gate-based, finite-difference Dirac model with a spatially varying mass, combines its Fourier-space kinetic and position-space mass propagators by Strang splitting, and executes the compiled circuits on H2-2.

Taken together, prior work leaves a focused experimental and circuit-design gap addressed here: a common compilation and hardware benchmark spanning one-dimensional acoustic propagation, an exact two-dimensional acoustic block reduction, and split-operator variable-mass Dirac dynamics. We combine a larger one-dimensional hardware instance than the preceding QFT demonstration~\cite{wright2024noisy}, an exact alternative to directional Trotterization for the two-dimensional model considered here, and a position–Fourier split Dirac extension, with native-resource scaling and observable-level H2-2 measurements throughout. The scope remains deliberately limited to periodic problems, structured low-bandwidth inputs, and coarse marginal observables.

\begin{table*}[t]
\centering
\caption{Comparison of representative quantum algorithms for the wave
equation in terms of qubit requirements, circuit depth or total gate count,
and physical-hardware implementation. Here $N_h$ is the number of grid points
per spatial dimension, $n_h=\log_{2}N_h$, $d$ is the spatial dimension, $t$ is the evolution time, $\varepsilon$ is the target error, and $K$
is the largest retained wavenumber. An em dash in the hardware column
indicates that no execution on a physical quantum processor was reported.
Suau \emph{et al.}\cite{suau2021practical} used an idealized hardware model informed by realistic
gate data.}
\label{tab:scaling}
\begin{ruledtabular}
\begin{tabular}{l l c c l}
Reference & Approach & Hardware & Qubits & Depth / gate count \\
\colrule
Costa \emph{et al.}~\cite{costa2019quantum}
  & Direct quantum walk
  & --- & $O(d\log N_h)$ & $\mathrm{polylog}(N_h)$ \\
Suau \emph{et al.}~\cite{suau2021practical}
  & Resource benchmark
  & Idealized model & $O(\log N_h)$ & $O(N_h^{3/2}\log^{2}N_h)$ \\
Sato \emph{et al.}~\cite{sato2024hamiltonian}
  & Explicit FD Hamiltonian
  & IBM \texttt{ibm\_kawasaki} & $dn_h$
  & $O(dn_h^{3}T^{2}/\varepsilon)$ \\
Arseniev \emph{et al.}~\cite{arseniev2024high}
  & Trotter + high-order FD
  & --- & $O(dn_h)$ & $\tilde{O}(N_h^{2})$ steps \\
Wright \emph{et al.}~\cite{wright2024noisy}
  & Low-mode QFT
  & Quantinuum H1-1 & $dn_h+O(1)$
  & $O(\log^{2}N_h)$ \\
Lubasch \emph{et al.}~\cite{lubasch2025fourier}
  & QFT + QSVT
  & --- & $dn_h+O(1)$
  & $\tilde{O}(d^{3}K^{2}n_ht/\varepsilon)$ \\
\textbf{This work}
  & QFT + bright--dark/Strang
  & Quantinuum H2-2 & $dn_h+O(1)$
  & $O(\log^{2}N_h)$ per step \\
\end{tabular}
\end{ruledtabular}
\end{table*}

\section{Wave equations and structure-preserving discretizations}
\label{sec:disc}

Here we show the wave equations as well as the finite difference approximation in space. Although finite difference schemes for the wave equations are standard \cite{LarssonThomee2009}, we provide the explicit forms here to specify the experimental runs later.  

\subsection{One-dimensional acoustic wave equation}\label{subsec:wave1d}

We consider the one-dimensional acoustic wave equation on the unit interval
with periodic boundary conditions,
\begin{equation}\label{eq:wave1d}
  \partial_{t}^{2}u(x,t)=c^{2}\,\partial_{x}^{2}u(x,t),
  \qquad x\in[0,1],\quad t>0,
\end{equation}
together with the initial conditions 
\begin{equation}\label{eq:wave1d-ic}
  u(x,0)=f(x),\qquad \partial_{t}u(x,0)=g(x),
\end{equation}
and periodic boundary condition (PBC) $u(0,t)=u(1,t)$ and $\partial_x u(0,t)=\partial_x  u(1,t)$. Without loss of generality, we set the wave speed to
$c=1$ and chose the domain to be the unit interval as the spatial domain by properly rescaling time.

We first reformulate \cref{eq:wave1d} into a Schr\"odinger type of equation. To expose a Hermitian generator, we pass to a first-order velocity-pressure
system. Introducing an auxiliary field $v= \partial_t u$ alongside the flux $p= \partial_x u$,
\cref{eq:wave1d} is equivalent to
\begin{equation}\label{eq:wave1d-first}
  \partial_{t}v=\partial_{x}p,\qquad
  \partial_{t}p=\partial_{x}v.
\end{equation}

 Writing $\ket{\psi}=(v, p)^{\mathsf T}$,
\cref{eq:wave1d-first} takes the Schr\"odinger form
\begin{equation}\label{eq:H1d}
  \partial_{t}\ket{\psi}=-i\mathcal H\,\ket{\psi},
  \qquad
  \mathcal H=\mat{0 & i\partial_{x}\\ i\partial_{x} & 0}
   =\sigma_{x}\otimes(i\partial_{x}).
\end{equation}
Because $i\partial_{x}$ is Hermitian (self-adjoint) under periodic boundary conditions, $H$
is Hermitian and the evolution $\ket{\psi(t)}=e^{-iHt}\ket{\psi(0)}$ is unitary.  The initial conditions can be obtained from \cref{eq:wave1d-ic}: $v(x,0)= g(x)$ and $p(x,0)= f'(x).$  

\subsection{Finite-difference approximations}\label{subsec:fd}

We discretize $[0,1]$ on a uniform grid of $N_h$ points,
$x_{j}=jh$ with $j=0,\dots,N-1$ and spacing $h=1/N_h$. The nodal
values are collected into vectors $\bm v,\bm p\in\mathbb{C}^{N_h}$. For the first derivative in \cref{eq:H1d}, we can
use the standard forward difference
\begin{equation}\label{eq:D-fd}
  (D\bm q)_{j}=\frac{q_{j+1}-q_{j}}{h},
  \qquad
  D=\frac{1}{h}\,\operatorname{circ}(-1,1,0,\dots,0),
\end{equation}
with indices taken modulo $N$, so that $D$ is the circulant matrix whose first
row is $h^{-1}(-1,1,0,\dots,0)$. Its adjoint is
$D^{\dagger}=-D_{-}$, where $(D_{-}\bm q)_{j}=(q_{j}-q_{j-1})/h$ is the backward
difference.

 Let $q_j(t)$ denote the nodal value approximation $u(x_j,t)$ and following \cref{eq:wave1d-first} we define $p_j= D_{-}q_j $ and $v_j= \dot{q}_j$. 
Then we can pair $D$ with its adjoint in the two off-diagonal
blocks to approximate \cref{eq:H1d}
\begin{equation}\label{eq:semidiscrete1d}
  \partial_{t}\mat{\bm v\\ \bm p}=-iH\mat{\bm v\\ \bm p},
  \qquad
  H=\mat{0 & iD\\ -iD^{\dagger} & 0}.
\end{equation}
This discretization is structure-preserving in that 
 $H^{\dagger}=H$. Therefore, this system of equations can be regarded as a finite-dimensional Schr\"odinger equation $i\frac{d}{dt} \ket{\psi} = H \ket{\psi},$  where the wave function $\ket{\psi}$ encodes both the velocity and the flux coherently.

The discretization \Cref{eq:semidiscrete1d} can be converted back to the section order form, which gives,
\begin{equation}\label{eq:laplacian}
 \frac{d^2}{dt^2}\bm q_j =\frac{q_{j+1}-2q_{j}+q_{j-1}}{h^{2}},
\end{equation}
which reproduces the standard semi-discrete approximation of the wave equation \eqref{eq:wave1d} with $u(x_j,t) \approx q_j(t) $. In particular, the standard error analysis gives the error $O(Th^2)$ for the displacenent $u(x,t)$ \cite{Baker1976HyperbolicFEM}.

\subsection{One-dimensional Dirac dynamics with a spatially varying mass}
\label{subsec:dirac}

The acoustic-wave benchmarks considered above are translation invariant:
their finite-difference generators are diagonalized in the Fourier basis, and
the simulated time enters only through mode-controlled rotation angles. As a
final and more demanding benchmark, we consider a spatially varying-mass Dirac
equation. This example is chosen because it provides a minimal physically
meaningful extension beyond the Fourier-diagonal setting. Its kinetic term is
simple in Fourier space, whereas its mass term is simple in position space.
This is the classical analogue of quantum dynamics in the first quantization, and quantum algorithms that make use of such structures by operator splitting can be found in \cite{Babbush2018LowDepth,jin2022quantum}. In particular, 
when the mass varies spatially, the two terms do not commute, so their
simulation requires repeated changes of basis and product-formula steps. The
model therefore isolates the additional circuit cost associated with spatial
heterogeneity while retaining a two-component Hermitian evolution compatible
with the acoustic-wave circuit primitive.

In units in which $\hbar=c=1$, the one-dimensional Dirac equation is
\begin{equation}
  \label{eq:kg}
  i\partial_t\Psi(x,t)
  =
  \mathcal H_{\mathrm D}\Psi(x,t),
  \qquad
  \mathcal H_{\mathrm D}
  =
  -i\sigma_x\partial_x-m(x)\sigma_y,
  \qquad
  x\in[0,1],
\end{equation}
with periodic boundary conditions and a real-valued mass profile $m(x)$.

To see the impact of the mass on the wave propagation directly, one can write,
\begin{equation}
  \label{eq:dirac-spinor}
  \Psi(x,t)
  =
  \begin{pmatrix}
    \psi_A(x,t)\\
    \psi_B(x,t)
  \end{pmatrix},
\end{equation}
and look at
\begin{equation}
  \label{eq:dirac-counterpropagating-fields}
  \psi_{+}
  =
  \frac{\psi_A+\psi_B}{\sqrt{2}},
  \qquad
  \psi_{-}
  =
  \frac{\psi_A-\psi_B}{\sqrt{2}}.
\end{equation}
Equation~\eqref{eq:kg} is then equivalent to
\begin{equation}
  \label{eq:dirac-counterpropagating-system}
  \partial_t\psi_{+}
  =
  -\partial_x\psi_{+}-m(x)\psi_{-},
  \qquad
  \partial_t\psi_{-}
  =
  \partial_x\psi_{-}+m(x)\psi_{+}.
\end{equation}
When $m(x)=0$, the fields $\psi_{+}$ and $\psi_{-}$ propagate in opposite
directions. A nonzero mass term locally converts one component into the other,
producing dispersive spinor dynamics and, at a mass interface, reflection and
transmission will occur.

For the finite-dimensional implementation, we use the same uniform grid points, and let
\begin{equation}
  \label{eq:kg-first}
  M
  =
  \operatorname{diag}
  \left(
    m(x_0),\ldots,m(x_{N_h-1})
  \right),
\end{equation}
and let $D$ be the periodic forward-difference operator defined in
\cref{eq:D-fd}. Then the semi-discrete spinor
satisfies
\begin{equation}
  \label{eq:dirac-discrete-evolution}
  i\partial_t \ket{\psi(t)}
  =
  H_{\mathrm D}\ket{\psi(t)},
\end{equation}
where the implemented Hamiltonian is a direction extension of that in \cref{eq:H1d}
\begin{equation}
  \label{eq:dirac-first}
  H_{\mathrm D}
  =
  \underbrace{
  \begin{pmatrix}
    0 & iD\\
    -iD^\dagger & 0
  \end{pmatrix}
  }_{H_{\mathrm{wave}}}
  +
  \underbrace{
  \begin{pmatrix}
    0 & iM\\
    -iM & 0
  \end{pmatrix}
  }_{H_{\mathrm{mass}}}.
\end{equation}
Both contributions are Hermitian. The first term is the same
finite-difference propagation generator used for the one-dimensional acoustic
equation, but its two fields are now interpreted as spinor components rather
than flux and velocity. The second term can be written as
\begin{equation}
  \label{eq:dirac-mass-pauli}
  H_{\mathrm{mass}}
  =
  -M\otimes\sigma_y,
\end{equation}
and therefore implements a position-dependent rotation of the field qubit.

In the hardware benchmark below, we choose the step profile
\begin{equation}
  \label{eq:dirac-mass-profile-preview}
  m(x)
  =
  \begin{cases}
    0, & 0\leq x<\tfrac{1}{2},\\
    2, & \tfrac{1}{2}\leq x<1.
  \end{cases}
\end{equation}
This profile is useful for both physical and circuit-level reasons. Physically,
it creates an interface between a gapless region and a gapped region, causing
the two spinor components to undergo spatially inhomogeneous mixing. Because
the computational domain is periodic, there are two interfaces: one at
$x=\tfrac{1}{2}$ and one across the periodic seam $x=0\equiv1$. At the circuit
level, alignment of the step with the midpoint makes the mass value depend
only on the most-significant grid qubit. The full position-dependent mass
propagator can therefore be implemented by a single controlled rotation,
without a general coefficient-loading oracle.

The example is consequently simple enough to execute on present hardware but
nontrivial enough to test features absent from the acoustic benchmarks:
Fourier-mode mixing, noncommuting kinetic and mass terms, repeated QFT and
inverse-QFT layers, and growth of the circuit depth with the product-formula
step count. The cosine initial condition used below is a standing spinor
excitation rather than a localized incident wave packet. We therefore
interpret the measured dynamics as spatial redistribution between spinor
components, rather than as the direct observation of a single reflection
event.

\subsection{Two-dimensional acoustic wave equation}\label{subsec:wave2d}

In two dimensions, the acoustic wave equation on the unit square with periodic
boundary conditions is
\begin{equation}\label{eq:wave2d}
  \partial_t^2 u(x,y,t)=\partial_x^2u+\partial_y^2u,
  \qquad (x,y)\in[0,1]^2,
\end{equation}
with initial data $u(\cdot,\cdot,0)=f$ and $\partial_tu(\cdot,\cdot,0)=g$.
Introducing the velocity
\[
v=\partial_t u
\]
and the two components of
\[
p=(p_x,p_y)^{\mathsf T}:=(\partial_x u,\partial_y u)^{\mathsf T},
\]
we obtain the first-order system
\begin{equation}\label{eq:first-order-2d}
  \partial_t v=\partial_x p_x+\partial_y p_y,
  \qquad
  \partial_t p_x=\partial_x v,
  \qquad
  \partial_t p_y=\partial_y v.
\end{equation}
Equivalently, with
\[
\Psi=(v,p_x,p_y)^{\mathsf T},
\]
this system can be written as
\begin{equation}
  \partial_t \Psi=-i\mathcal H\Psi,
  \qquad
  \mathcal H=
  \begin{bmatrix}
    0 & i\partial_x & i\partial_y\\
    i\partial_x & 0 & 0\\
    i\partial_y & 0 & 0
  \end{bmatrix}.
\end{equation}
Under periodic boundary conditions, $i\partial_x$ and $i\partial_y$ are
self-adjoint, so $\mathcal H$ is Hermitian.

We now turn to the finite-difference discretization. On an
$N_h\times N_h$ uniform grid with periodic boundary conditions, we represent the
partial derivatives by the tensor-product operators
\begin{equation}\label{eq:DxDy}
  D_x=D\otimes I_{N_h},
  \qquad
  D_y=I_{N_h}\otimes D,
\end{equation}
where $D$ is the one-dimensional forward difference operator from
\cref{eq:D-fd}. Collecting the discrete unknowns at the grid points as
\[
\ket{\psi}=(\bm v,\bm p_x,\bm p_y)^{\mathsf T},
\]
the structure-preserving semi-discrete generator becomes
\begin{equation}\label{eq:semidiscrete2d}
  H_{2D}=
  \begin{bmatrix}
    0 & iD_x & iD_y\\
    -iD_x^\dagger & 0 & 0\\
    -iD_y^\dagger & 0 & 0
  \end{bmatrix},
  \qquad
  \frac{d}{dt}\ket{\psi}=-iH_{2D}\ket{\psi}.
\end{equation}
The off-diagonal blocks are paired with their adjoints, exactly as in the
one-dimensional case; hence $H_{2D}$ is Hermitian.

\subsection{Fast forwarding by QFT}\label{subsec:qft}

The forward difference $D$ in \cref{eq:D-fd} is circulant and is therefore
diagonalized by the discrete Fourier transform, which on $N_h=2^{n_h}$ grid points is
realized on $n_h$ qubits by the quantum Fourier transform. We use the standard
convention
\begin{equation}\label{eq:qft-def}
  \F\ket{j}=\frac{1}{\sqrt{N_h}}\sum_{k=0}^{N-1}e^{2\pi i\,jk/N_h}\ket{k},
\end{equation}
which acts on amplitude vectors as $\hat{\bm q}=\F\bm q$; the inverse transform
$\F^{\dagger}$ carries the conjugate phase. 

We express the velocity and pressure in a quantum state as $\ket{\psi}= \sum_j v_j \ket{j} \ket{0} + p_j \ket{j} \ket{1}$ with a field register. Then
applying QFT to $D$ then amounts to
\begin{equation}\label{eq:circulant-eig}
  \F\,D\,\F^{\dagger}=\operatorname{diag}(\ell_{0},\dots,\ell_{N_h-1}),
  \qquad
  \ell_{k}=\frac{1}{h}\bigl(e^{-2\pi i k/N_h}-1\bigr),
\end{equation}
the eigenvalues being obtained by applying $D$ to the Fourier mode
$(e^{-2\pi i jk/N_h})_{j}$.

Conjugating the semi-discrete generator \cref{eq:semidiscrete1d} by
$\F \otimes I$ and regrouping the two Fourier amplitudes $(\hat v_{k},\hat p_{k})$
mode by mode block-diagonalizes the dynamics into $2\times2$ blocks,
\begin{equation}\label{eq:Hhat1d}
  \widehat{H}=(\F \otimes I )\,H\,(\F \otimes I)^{\dagger}
  =\sum_{k=0}^{N_h-1}\ket{k}\!\bra{k}\otimes H_{k},
  \qquad
  H_{k}=\mat{0 & \lambda_{k}\\ \bar\lambda_{k} & 0},
\end{equation}
where the per-mode coupling is the eigenvalue of $-iD$,
\begin{equation}\label{eq:lambdak}
  \lambda_{k}=-i\ell_{k}
            =\frac{i}{h}\bigl(1-e^{2\pi i k/N_h}\bigr)
            =-\frac{2\sin\theta_{k}}{h}\,e^{-i\theta_{k}},
  \qquad
  \theta_{k}=\frac{\pi k}{N_h}=\pi k h.
\end{equation}

\section{Experimental setup}
\label{sec:experimental-setup}

We evaluate the circuits constructed in the following sections through a four-stage pipeline that proceeds from a classical reference solution to results measured on Quantinuum H2-2 device. Within this pipeline, we characterize the compiled circuits along two independent resource-scaling axes: as a function of the evolution time (t) at fixed register size, and as a function of the grid-qubit count($n_h$), at a fixed representative time (t=0.1).

First, we solve the relevant semi-discrete system directly in classical software(python). The finite-difference generator ($H$), or equivalently its Fourier-diagonalized form, such as \cref{eq:Hhat1d} for the one-dimensional wave equation, is assembled and exponentiated using NumPy/SciPy to produce a reference trajectory. This classical simulation is performed independently of the quantum circuit and fixes the exact target amplitudes used in the state preparation and evolution procedures. It therefore provides the ground truth against which all circuit-based results reported in this work are compared.

Second, each quantum circuit is assembled in Qiskit~\cite{qiskit2024} and evaluated using Qiskit's AerSimulator. Firstly, we evaluate the ideal circuit and extract its logical resource requirements before any hardware-specific compilation. These logical resources correspond to the circuit as constructed by Qiskit and include the total gate count, the two-qubit gate count, the total circuit depth, and the depth contributed by two-qubit gates alone. Secondly, we perform shot-based sampling to obtain the measurement-count distributions required to estimate the relevant kinetic-energy observables on an ideal, noise-free device. For this we sampled for a shot count of 8192. The resulting expectation values define the noiseless reference curves against which the emulator and hardware measurements are compared.

Third, the same logical circuit is retargeted to the native gate set and connectivity constraints of the Quantinuum H2-2 device using pytket-qiskit~\cite{sivarajah2020tket}. This procedure converts the Qiskit circuit into a \texttt{pytket} circuit for which compilation is performed by submitting the circuit to Quantinuum's \texttt{qnexus} cloud interface\cite{quantinuum_nexus}. The resources of the resulting native circuit correspond to compiled gate counts and depths and are identical for execution on the trapped-ion H2-2 processor and on its noise-modeled emulator H2-2E. As for the logical circuit, we separately report the total gate count and depth and the corresponding two-qubit-only contributions. This distinction is particularly relevant for the QCCD architecture, where two-qubit gates are implemented as long-range physical operations and constitute the dominant contribution to both execution time and error accumulation\cite{moses2023racetrack}.

Fourth, the compiled circuit is submitted through \texttt{qnexus}, which is used to target both the physical H2-2 trapped-ion processor and its noise-modeled emulator, H2-2E. For this we sampled for a shot count of 1024, it was kept less as it takes proportionately higher number of resources in terms of Hardware Quantum Credits(HQCs), which are required in order to execute a circuit on a Quantinuum hardware via qnexus. Furthermore, as it can be seen later, the chosen shot count was enough to obtain fairly good results in terms of evaluating the kinetic energy observable for various problems considered.

\section{Quantum circuit constructions}

\subsection{One-dimensional acoustic wave equation}\label{sec:permode}
\subsubsection{Fourier-mode propagator}
It is convenient to record the modulus of $\lambda_{k}$. With
$h=1/N_{h}$ and using $\sin\theta_{k}\ge 0$ for $0\le k<N_{h}$,
\begin{equation}\label{eq:modphase}
  \lvert\lambda_{k}\rvert=2N_{h}\sin\theta_{k},
  \qquad
  \theta_{k}=\frac{\pi k}{N_{h}}.
\end{equation}
Thus each block in \cref{eq:Hhat1d} simply consists of Pauli operators,
\begin{equation}
  H_{k}=-\lvert\lambda_{k}\rvert\bigl(\cos\theta_{k}\,\sigma_{x}+\sin\theta_{k}\,\sigma_{y}\bigr)
       =-\lvert\lambda_{k}\rvert\,\Rz(\theta_{k})\,\sigma_{x}\,\Rz(-\theta_{k}).
\end{equation}
Equivalently, using $\sigma_{x}=\Had\,\sigma_{z}\,\Had$ (with $\Had$ the Hadamard gate),
\begin{equation}\label{eq:Hkhad}
  H_{k}
    =\Rz(\theta_{k})\,\Had\bigl(-2N_{h}\sin\theta_{k}\,\sigma_{z}\bigr)\Had\,\Rz(-\theta_{k})
    =\Rz(\theta_{k})\,\Had\bigl(-\lvert\lambda_{k}\rvert\,\sigma_{z}\bigr)\Had\,\Rz(-\theta_{k}),
\end{equation}
which isolates a diagonal generator $-\lvert\lambda_{k}\rvert\,\sigma_{z}$ dressed by the
affine phase $\theta_{k}=\pi k/N_{h}$ and a Hadamard gate. Summing over modes, we have,
\begin{equation}\label{eq:Hhad}
  \widehat H=\sum_{k=0}^{N_{h}-1}\ket{k}\!\bra{k}\otimes
    \Bigl[\Rz(\theta_{k})\,\Had\bigl(-2N_{h}\sin\theta_{k}\,\sigma_{z}\bigr)\Had\,\Rz(-\theta_{k})\Bigr].
\end{equation}
Exponentiating either form gives a single $x$-rotation conjugated by $z$-rotations,
\begin{equation}\label{eq:Uk}
  e^{-iH_{k}t}=\Rz(\theta_{k})\,\Rx\!\bigl(-2t\lvert\lambda_{k}\rvert\bigr)\,\Rz(-\theta_{k}),
  \qquad \Rx(\alpha)=e^{-i\alpha\sigma_{x}/2},
\end{equation}
so that the full propagator for the wave dynamics \cref{eq:H1d} is given by
\begin{equation}\label{eq:U1d}
  U(t)=e^{-i\widehat{H}t}
      =\sum_{k=0}^{N_{h}-1}\ket{k}\!\bra{k}\otimes
       \Rz(\theta_{k})\,\Rx\!\bigl(-4tN_{h}\sin\theta_{k}\bigr)\,\Rz(-\theta_{k}).
\end{equation}

The only remaining obstacle to a shallow circuit is the rotation angle $-4tN_{h}\sin\theta_{k}$, which is
nonlinear in the mode index; the surrounding phase $\theta_{k}=\pi k/N_{h}$ is affine in $k$ and
hence already diagonal in the computational basis of the mode register.

\subsubsection{Low-mode linearization}
\label{subsec:1d-linearization}
For low Fourier modes $k=O(1)$, we can apply a linear approximation, 
\begin{equation}\label{eq:sinapprox}
  \sin\!\left(\frac{\pi k}{N_{h}}\right)\approx\frac{\pi}{N_{h}}\min(k,N_{h}-k)
  =\frac{\pi}{N_{h}}\,k_{s},
  \qquad k_{s}\equiv\min(k,N_{h}-k).
\end{equation}
Substituting this into the diagonal generator of~\eqref{eq:Hkhad} replaces
$-2N_{h}\sin\theta_{k}$ by $-2N_{h}\cdot\tfrac{\pi}{N_{h}}k_{s}=-2\pi k_{s}$, giving the approximate
diagonal generator
\begin{equation}\label{eq:HZtilde}
  \widetilde H_{Z}=\sum_{k=0}^{N_{h}-1}\ket{k}\!\bra{k}\otimes\bigl[-2\pi k_s \,\sigma_{z}\bigr].
\end{equation}
For even $N_{h} = 2^{n_{h}}$, splitting the fold at the midpoint gives,
\begin{align}
  \widetilde H_{Z}
    &=\sum_{k=0}^{N_{h}/2-1}\ket{k}\!\bra{k}\otimes(-2\pi k\,\sigma_{z})+\sum_{k=N_{h}/2}^{N_{h}-1}\ket{k}\!\bra{k}\otimes\bigl[-2\pi(N_{h}-k)\,\sigma_{z}\bigr].
\end{align}
Using $-2\pi(N_{h}-k)=-2\pi k-2\pi N_{h}+4\pi k$, this regroups as
\begin{align}\label{eq:HZtilde-msb}
  \widetilde H_{Z}
    &=\sum_{k=0}^{N_{h}-1}\ket{k}\!\bra{k}\otimes(-2\pi k\,\sigma_{z})
    +\sum_{k=N_{h}/2}^{N_{h}-1}\ket{k}\!\bra{k}\otimes(-2\pi N_{h}\,\sigma_{z})+\sum_{k=N_{h}/2}^{N_{h}-1}\ket{k}\!\bra{k}\otimes(4\pi k\,\sigma_{z}).
\end{align}

\begin{lemma}\label{lem:crz}
For any fixed $r\in\{0,\dots,n_{h}-1\}$ and any $\phi\in\mathbb{R}$,
\begin{equation}\label{eq:crz-def}
  CR_Z^{(r)}(\phi)=\sum_{k=0}^{N_{h}-1}\ket{k}\!\bra{k}\otimes\Rz(\phi k_{r}),
\end{equation}
where $k_{r}\in\{0,1\}$ is the $r$-th bit in the binary expansion
$k=\sum_{r=0}^{n_{h}-1}k_{r}2^{r}$, with $N_{h}=2^{n_{h}}$.
\end{lemma}

\begin{proof}
Since $\ket{0}\!\bra{0}_{r}+\ket{1}\!\bra{1}_{r}=I^{(r)}$ is the identity on the
$r$-th qubit, the controlled rotation
$CR_Z^{(r)}(\phi)=\ket{0}\!\bra{0}_{r}\otimes I+\ket{1}\!\bra{1}_{r}\otimes\Rz(\phi)$
can be written as
\begin{align}\label{eq:crz-expand}
  CR_Z^{(r)}(\phi)
    &=I\otimes I+\ket{1}\!\bra{1}_{r}\otimes\bigl(\Rz(\phi)-I\bigr)\nonumber\\
    &=\sum_{k=0}^{N_{h}-1}\ket{k}\!\bra{k}\otimes I
      +\sum_{k=0}^{N_{h}-1}k_{r}\,\ket{k}\!\bra{k}\otimes\bigl(\Rz(\phi)-I\bigr)\nonumber\\
    &=\sum_{k=0}^{N_{h}-1}\ket{k}\!\bra{k}\otimes\bigl[I+k_{r}\bigl(\Rz(\phi)-I\bigr)\bigr]
     =\sum_{k=0}^{N_{h}-1}\ket{k}\!\bra{k}\otimes\Rz(\phi k_{r}),
\end{align}
where the second line uses $I=\sum_{k}\ket{k}\!\bra{k}$ and
$\ket{1}\!\bra{1}_{r}=\sum_{k}k_{r}\ket{k}\!\bra{k}$, and the last equality
follows from $k_{r}\in\{0,1\}$, since then
$I+k_{r}(\Rz(\phi)-I)=\Rz(\phi k_{r})$.
\end{proof}

Any diagonal profile that is affine in the bits, $\theta_{k}=\sum_{r=0}^{n_{h}-1}\alpha_{r}k_{r}$,
therefore factorizes into a product of controlled $z$-rotations. By
Lemma~\ref{lem:crz} and the mutual commutativity of the $CR_Z^{(r)}$,
\begin{equation}\label{eq:M-synth}
  V_\theta=\sum_{k=0}^{N_{h}-1}\ket{k}\!\bra{k}\otimes\Rz(\theta_{k})
   =\sum_{k=0}^{N_{h}-1}\ket{k}\!\bra{k}\otimes\Rz\!\Bigl(\sum_{r=0}^{n_{h}-1}\alpha_{r}k_{r}\Bigr)
   =\prod_{r=0}^{n_{h}-1}CR_Z^{(r)}(\alpha_{r}).
\end{equation}
In particular, choosing $\alpha_{r}=\pi 2^{r}/N_{h}$ realizes the linear phase
$\theta_{k}=\pi k/N_{h}=\tfrac{\pi}{N_{h}}\sum_{r}2^{r}k_{r}$:
\begin{equation}\label{eq:M-linear}  V_\theta=\sum_{k=0}^{N_{h}-1}\ket{k}\!\bra{k}\otimes\Rz\!\Bigl(\tfrac{\pi k}{N_{h}}\Bigr)
   =\prod_{r=0}^{n_{h}-1}CR_Z^{(r)}\!\Bigl(\frac{\pi 2^{r}}{N_{h}}\Bigr).
\end{equation}

It thus follows that the approximate propagator is given by (note that here H is the Hadamard gate)
\begin{equation}
\label{eq:1d-fourier-app-ham}
\widehat U_{\mathrm{approx}}(t) = e^{-i\widehat H t} = V_\theta(I \otimes \mathrm H)e^{-i\widetilde H_Z t}(I \otimes \mathrm H ) V_\theta^\dagger
\end{equation}
This linear approximation has been implemented by the circuit diagram shown in \cref{fig:hz-tilde-circuit}

\medskip

Now, we show an estimate of the error due to \cref{eq:sinapprox}. From \cref{eq:U1d} and \cref{eq:sinapprox} it follows that
\begin{equation}\label{eq:Uapp}
  \widehat U_{\mathrm{approx}}(t)=\sum_{k=0}^{N_{h}-1}\ket{k}\!\bra{k}\otimes
       \Rz(\theta_{k})\,\Rx\!\bigl(-4\pi t\,k_{s}\bigr)\,\Rz(-\theta_{k}).
\end{equation}
Here, t denotes the global evolution time, not a small time step, so the resulting circuit depth is independent of t, which is precisely the fast-forwarding property. 
Note that both propagators are block-diagonal in the mode register and share the unitary operator
$\Rz(\pm\theta_{k})$, so the operator norm of their difference is the
largest single-block deviation, and conjugation by the dressing leaves each block invariant.

Let
\begin{equation}\label{eq:spectral-projector}
  \Pi_K
  :=
  \sum_{\substack{0\le k<N_h\\ k_s\le K}}
  \ket{k}\!\bra{k}\otimes I_2,
  \qquad
  k_s:=\min\{k,N_h-k\},
  \qquad
  0\le K\le \left\lfloor\frac{N_h}{2}\right\rfloor .
\end{equation}
Because both propagators are block diagonal in the Fourier basis and
preserve $\operatorname{Ran}(\Pi_K)$, their operator-norm difference on
the retained spectral subspace is
\begin{equation}\label{eq:blocknorm}
  \bigl\lVert
    \left(U(t)-U_{\mathrm{approx}}(t)\right)\Pi_K
  \bigr\rVert
  =
  \max_{k_s\le K}
  \bigl\lVert
    \Rx\!\bigl(-4tN_h\sin\theta_k\bigr)
    -
    \Rx\!\bigl(-4\pi t\,k_s\bigr)
  \bigr\rVert .
\end{equation}
Using
$\lVert\Rx(a)-\Rx(b)\rVert
 =\lVert\Rx(a-b)-I\rVert$,
let
\begin{equation}
  \delta_k
  :=
  \sin\theta_k-\frac{\pi k_s}{N_h},
  \qquad
  \theta_k=\frac{\pi k_s}{N_h}.
\end{equation}
Since $a-b=-4tN_h\delta_k$, the error in the $k$th Fourier block is
\begin{equation}\label{eq:rxdiff}
  \Rx\!\bigl(-4tN_h\delta_k\bigr)-I
  =
  e^{\,i(2tN_h\delta_k)\sigma_x}-I,
\end{equation}
and hence
\begin{equation}
  \bigl\lVert
    e^{\,i(2tN_h\delta_k)\sigma_x}-I
  \bigr\rVert
  =
  2\bigl|\sin(tN_h\delta_k)\bigr|
  \le
  2|t|N_h|\delta_k|.
\end{equation}
The linearization error is cubic:
\begin{equation}
  |\delta_k|
  =
  \left|
    \sin\!\left(\frac{\pi k_s}{N_h}\right)
    -
    \frac{\pi k_s}{N_h}
  \right|
  \le
  \frac{1}{6}
  \left(\frac{\pi k_s}{N_h}\right)^3.
\end{equation}
Therefore,
\begin{equation}\label{eq:errbound}
  {
  \bigl\lVert
    \left(U(t)-U_{\mathrm{approx}}(t)\right)\Pi_K
  \bigr\rVert
  \le
  \frac{\pi^3|t|K^3}{3N_h^2}.
  }
\end{equation}
Thus, for a family of initial states supported on a fixed spectral
band $k_s\le K=O(1)$, the state-vector error is
$O(|t|/N_h^2)$. More generally, for any normalized initial state
$\ket{\psi_0}$,
\begin{equation}\label{eq:state-error-with-tail}
  \bigl\|
    \left(U(t)-U_{\mathrm{approx}}(t)\right)\ket{\psi_0}
  \bigr\|
  \le
  \frac{\pi^3|t|K^3}{3N_h^2}
  +
  2\bigl\|(I-\Pi_K)\ket{\psi_0}\bigr\|.
\end{equation}
The second term explicitly accounts for the spectral tail outside the
retained low-frequency subspace.

Finally, since the infidelity between normalized pure states satisfies
\begin{equation}
  1-
  \left|
    \langle\psi_{\mathrm{exact}}
    \mid\psi_{\mathrm{approx}}\rangle
  \right|^2
  \le
  \left\|
    \ket{\psi_{\mathrm{exact}}}
    -
    \ket{\psi_{\mathrm{approx}}}
  \right\|^2,
\end{equation}
a fixed-bandwidth initial state obeys
\begin{equation}
  1-
  \left|
    \langle\psi_{\mathrm{exact}}
    \mid\psi_{\mathrm{approx}}\rangle
  \right|^2
  \le
  \frac{\pi^6t^2K^6}{9N_h^4}
  =
  O\!\left(\frac{t^2K^6}{N_h^4}\right).
\end{equation}
For fixed $K$, this reduces to $O(t^2/N_h^4)$, consistent with the
scaling reported for QFT-based wave-equation
circuits~\cite{wright2024noisy}.

\subsubsection{Complete evolution circuit}
The complete construction for the evolution of the wave equation is given by
\begin{equation}
    \ket{\Psi(t)}=e^{-iHt}\ket{\Psi(0)}
\end{equation}
\begin{equation}
    H=(\F\otimes I)^{\dagger}\widehat H(\F\otimes I)
    \quad\Longrightarrow\quad
    e^{-iHt}=(\F^{\dagger}\otimes I)\,\widehat U(t)\,(\F\otimes I),
    \qquad
    \widehat U(t)=e^{-i\widehat H t}
\end{equation}

For the initial conditions considered here, the velocity field vanishes,
\(v(x_j,0)=0\), while the flux field is initialized with the non-uniform profile \(p(x_j,0)=p_j\).

\begin{equation}
    \ket{\Psi(0)} = \ket{\psi_p}\ket{p}, \quad
    \ket{\psi_{p}} = U_{\mathrm{prep}}\ket{0}^{\otimes n_{h}}=\sum_{j}p_{j}\ket{j}, \quad
    \ket{\Psi(t)}
    =(\F^{\dagger}\otimes I)\,\widehat U(t)\,(\F\otimes I)\,
     U_{\mathrm{prep}}\ket{0}^{\otimes n_{h}}\ket{p}
\end{equation}
Here, $U_{prep}$ denotes the initial state preparation circuit in spatial domain where as $\widehat U_{prep}$ denotes the initial state preparation circuit in Fourier domain. The constructions of \cref{subsec:1d-cosine,subsec:1d-fivemode} instead give
$\widehat U_{\mathrm{prep}}$ with
\begin{equation}
    \widehat U_{\mathrm{prep}}\ket{0}^{\otimes n_{h}}=\ket{\psi_{p}}
    =(\F\otimes I)\, U_{\mathrm{prep}}\ket{0}^{\otimes n_{h}},
\end{equation} using construction of $\widehat U_{\mathrm{prep}}$ in
 \cref{eq:cos-target,eq:psiG} corresponding to cosine and gaussian profile respectively, we have
\begin{equation}\label{eq:1d-overall-circuit}
    \ket{\Psi(t)}
    =(\F^{\dagger}\otimes I)\,\widehat U(t)\,
     \widehat U_{\mathrm{prep}}\ket{0}^{\otimes n_{h}}\ket{p}
\end{equation}

The corresponding circuit schematic of \cref{eq:1d-overall-circuit} is shown in \cref{fig:evolution-circuit}.

\subsection{Two-dimensional acoustic wave equation}
\label{sec:methods2d}

The one-dimensional construction relies on the fact that the discrete derivative is diagonalized by the QFT and that each Fourier mode reduces to a two-level system.  In two dimensions, the Fourier diagonalization remains exact, but each mode contains two velocity components coupled to one flux component.  A naive directional splitting recovers two one-dimensional subproblems, but it introduces a sizeable product-formula error. The main observation of this section is that the unsplit two-dimensional block has rank two: only the
longitudinal velocity component couples to pressure, while the transverse
component is dark.  This gives an exact mode-wise reduction to a single
bright-flux rotation.

\subsubsection{Fourier block structure}
\label{subsec:2d-Fourier -blocks} 
To reduce the Hamiltonian in \cref{eq:semidiscrete2d} to the Fourier domain, we apply the two-dimensional QFT acting on the two spatial registers.
\begin{equation}
    \mathcal F_2=\mathcal F\otimes \mathcal F
\end{equation}
 The transform
is applied directly and unconditionally to both registers; no ancillary qubit is
needed to choose an $x$- or $y$-directional QFT.  Specifically, by using
\cref{eq:DxDy,eq:circulant-eig}, we have the following reduction
\begin{equation}
  \mathcal F_{2}D_{x}\mathcal F_{2}^{\dagger}
  =\operatorname{diag}(\ell_k) \otimes I,
  \qquad
  \mathcal F_{2}D_{y}\mathcal F_{2}^{\dagger}
  =I \otimes \operatorname{diag}(\ell_k).
\end{equation}
Therefore, conjugating \cref{eq:semidiscrete2d} by \(\mathcal F_2\) decouples
the dynamics into independent Fourier blocks indexed by
\(\bm k=(k_x,k_y)\):
\begin{equation}\label{eq:2d-Fourier -H}
  \widehat H_{2\mathrm D}
  =\sum_{k_x,k_y=0}^{N-1}
  \ket{k_x,k_y}\!\bra{k_x,k_y}\otimes H_{\bm k},
\end{equation}
where, in the field ordering \(\{\ket{v_x},\ket{v_y},\ket p\}\), each diagonal block exhibits the following matrix structure,
\begin{equation}\label{eq:2d-Hk-3by3}
  H_{\bm k}
  =
  \mat{0 & 0 & \lambda_{x,\bm k}\\
       0 & 0 & \lambda_{y,\bm k}\\
       \bar\lambda_{x,\bm k} & \bar\lambda_{y,\bm k} & 0 } .
\end{equation}
Here the nonzero entries are given by,
\begin{equation}\label{eq:2d-lambda-def}
    \lambda_{x,\bm k}
    =\frac{i}{h}\left(1-e^{-2\pi i k_x/N_h}\right),
    \qquad
    \lambda_{y,\bm k}
    =\frac{i}{h}\left(1-e^{-2\pi i k_y/N_h}\right).
\end{equation}
For circuit implementation, we pad this block by one unused field state and use
four field labels
\begin{equation}\label{eq:2d-field-labels}
    \ket{v_x}=\ket{00},
    \qquad
    \ket{v_y}=\ket{01},
    \qquad
    \ket p=\ket{10},
    \qquad
    \ket a=\ket{11}.
\end{equation}
The padded mode Hamiltonian is a matrix in $\mathbb{C}^{4\times 4}$ 
\begin{equation}\label{eq:2d-Hk-padded}
    \overline H_{\bm k}
    =
    \begin{pmatrix}
        0 & 0 & \lambda_x & 0\\
        0 & 0 & \lambda_y & 0\\
        \bar\lambda_x & \bar\lambda_y & 0 & 0\\
        0 & 0 & 0 & 0
    \end{pmatrix},
\end{equation}
where, to reduce notation, we write
\(\lambda_\mu=\lambda_{\mu,\bm k}\), \(\mu\in\{x,y\}\), inside a fixed block.

\subsubsection{Directional operator splitting }
\label{subsec:2d-trotterization}

A natural approach is to decompose the block into directional couplings,
\begin{equation}\label{eq:2d-Hsplit}
    \overline H_{\bm k}=H_x(k_x)+H_y(k_y),
\end{equation}
where
\begin{equation}\label{eq:2d-Hx-Hy}
    H_x=\lambda_x\ket{v_x}\!\bra p+\lambda_x^*\ket p\!\bra{v_x},
    \qquad
    H_y=\lambda_y\ket{v_y}\!\bra p+\lambda_y^*\ket p\!\bra{v_y},
\end{equation}
and implement them one at a time. Each of the Hamiltonians is acting in one direction, for which the circuit implemented has been outlined in \cref{sec:permode}.  
This is the two-dimensional analogue of applying the one-dimensional
Fourier-mode circuit separately in the two coordinate directions, and is close
in spirit to the directional Fourier approach used in \cite{lubasch2025fourier}.
Let
\begin{equation}
    h_\mu(k_\mu)=
    \begin{pmatrix}0&\lambda_\mu\\ \lambda_\mu^*&0\end{pmatrix}.
\end{equation}
Then, with the field register written as \(f_1\) and \(f_0\),
\begin{equation}\label{eq:2d-directional-blocks}
    H_x(k_x)=h_x(k_x)_{f_1}\otimes\ket0\!\bra0_{f_0},
\end{equation}
whereas the \(y\)-coupling between \(\ket{01}\) and \(\ket{10}\) can be mapped
to the same field-qubit axis by a CNOT,
\begin{equation}\label{eq:2d-y-directional-block}
    H_y(k_y)
    =C\left[h_y(k_y)_{f_1}\otimes\ket1\!\bra1_{f_0}\right]C,
    \qquad
    C=\operatorname{CNOT}_{f_1\to f_0}.
\end{equation}
Promoting these blocks to the full Fourier register gives
\begin{equation}\label{eq:2d-Hxy-full}
    \widehat H_x
    =\sum_{k_x=0}^{N-1}\ket{k_x}\!\bra{k_x}\otimes I_y\otimes H_x(k_x),
    \qquad
    \widehat H_y
    =I_x\otimes\sum_{k_y=0}^{N-1}\ket{k_y}\!\bra{k_y}\otimes H_y(k_y),
\end{equation}
with \(\widehat H_{2\mathrm D}=\widehat H_x+\widehat H_y\). 

The corresponding directional evolutions reduce to the one-dimensional
mode-evolution circuit.
\begin{equation}\label{eq:2d-Ux}
    e^{-i\widehat H_x\tau}
    =\sum_{k_x=0}^{N-1}\ket{k_x}\!\bra{k_x}\otimes I_y\otimes
      \left[
      e^{-ih_x(k_x)\tau}\otimes\ket0\!\bra0_{f_0}
      +I_{f_1}\otimes\ket1\!\bra1_{f_0}
      \right].
\end{equation}
Similarly,
\begin{equation}\label{eq:2d-Uy}
    e^{-i\widehat H_y\tau}
    =(I_{x}\otimes I_y\otimes C)
    \left[
    I_x\otimes\sum_{k_y=0}^{N-1}\ket{k_y}\!\bra{k_y}\otimes
      \left(
      e^{-ih_y(k_y)\tau}\otimes\ket1\!\bra1_{f_0}
      +I_{f_1}\otimes\ket0\!\bra0_{f_0}
      \right)
    \right]
    (I_{x}\otimes I_y \otimes C).
\end{equation}

Although each Hamiltonian induces an evolution that can be simulated separately,  the two directional Hamiltonians do not commute, i.e., for a fixed mode pair,
\begin{equation}\label{eq:2d-commutator}
    [H_x,H_y]
    =\lambda_x\lambda_y^*\ket{v_x}\!\bra{v_y}
    -\lambda_x^*\lambda_y\ket{v_y}\!\bra{v_x},
    \qquad
    \|[H_x,H_y]\|=|\lambda_x|\,|\lambda_y|.
\end{equation}
Thus,
\begin{equation}
    \norm{[\widehat H_x,\widehat H_y]}
    =\max_{k_x,k_y}|\lambda_{x,\bm k}|\,|\lambda_{y,\bm k}|.
\end{equation}

The first-order Lie--Trotter approximation,
\begin{equation}
    U_{\rm LT}^{(r)}(t)
    \coloneqq\left(e^{-i\widehat H_xt/r}e^{-i\widehat H_yt/r}\right)^r,
\end{equation}
therefore obeys
\begin{equation}\label{eq:2d-trotter-bound}
    \left\|e^{-i(\widehat H_x+\widehat H_y)t}-U_{\rm LT}^{(r)}(t)\right\|
    \le
    \frac{t^2}{2r}\|[\widehat H_x,\widehat H_y]\|.
\end{equation}
Since
\begin{equation}
    |\lambda_{\mu,\bm k}|
    =\frac{2}{h}\sin\left(\frac{\pi m_\mu}{N}\right),
    \qquad
    m_\mu=\min(k_\mu,N-k_\mu),
\end{equation}
and \(h=1/N\), the full estimate in operator norm is given by, 
\begin{equation}\label{eq:2d-trotter-full-spectrum}
    \left\|e^{-i(\widehat H_x+\widehat H_y)t}-U_{\rm LT}^{(r)}(t)\right\|
    \le
    \frac{2N^2t^2}{r}.
\end{equation}
Even if the state is restricted to low modes \(m_x\le K_x\), \(m_y\le K_y\),
we only obtain
\begin{equation}\label{eq:2d-trotter-lowmode}
    \left\|e^{-i(\widehat H_x+\widehat H_y)t}-U_{\rm LT}^{(r)}(t)\right\|
    \le
    \frac{2\pi^2t^2K_xK_y}{r}.
\end{equation}
For example, \(K_x=K_y=1\), \(t=1\), and target error \(10^{-1}\) already give
\(r\gtrsim 2\times 10^2\) from this bound.  Thus the directional approach is
simple but produces circuits whose depth is dominated by the required number of
product-formula steps.  We next describe an exact mode-wise reduction that
removes this product-formula error.

\subsubsection{Bright--dark reduction}
\label{subsec:bright-dark-reduction}

The directional decomposition treats the two pressure--velocity couplings
separately, even though they belong to the same three-state Fourier block. A more efficient construction follows by first identifying the velocity superposition that actually couples to pressure. This is the same algebraic structure that appears in the Morris--Shore transformation of multilevel optical systems, where a change of basis separates coupled bright states from uncoupled dark states~\cite{morris1983reduction,fleischhauer2005eit}.

For a fixed Fourier mode \(\bm k=(k_x,k_y)\), define
\begin{equation}\label{eq:2d-amplitudes-Omega}
    A_\mu=|\lambda_\mu|,
    \qquad
    \lambda_\mu=A_\mu e^{i\phi_\mu},
    \qquad
    \Omega_{\bm k}=\sqrt{A_x^2+A_y^2}.
\end{equation}
When \(\Omega_{\bm k}\neq0\), introduce the mixing angle
\begin{equation}\label{eq:eta-def}
    \cos\eta_{\bm k}=\frac{A_x}{\Omega_{\bm k}},
    \qquad
    \sin\eta_{\bm k}=\frac{A_y}{\Omega_{\bm k}}.
\end{equation}
The normalized velocity superposition parallel to the discrete Fourier
gradient is
\begin{equation}\label{eq:bright-state}
    \ket{b_{\bm k}}
    =\cos\eta_{\bm k}\,e^{i\phi_x}\ket{v_x}
     +\sin\eta_{\bm k}\,e^{i\phi_y}\ket{v_y},
\end{equation}
whereas the orthogonal velocity superposition is
\begin{equation}\label{eq:dark-state}
    \ket{d_{\bm k}}
    =-\sin\eta_{\bm k}\,e^{i\phi_x}\ket{v_x}
     +\cos\eta_{\bm k}\,e^{i\phi_y}\ket{v_y}.
\end{equation}
A direct substitution into \cref{eq:2d-Hk-3by3} gives
\begin{equation}\label{eq:bright-dark-action}
    H_{\bm k}\ket{b_{\bm k}}=\Omega_{\bm k}\ket p,
    \qquad
    H_{\bm k}\ket p=\Omega_{\bm k}\ket{b_{\bm k}},
    \qquad
    H_{\bm k}\ket{d_{\bm k}}=0.
\end{equation}
Thus flux couples only to the longitudinal, or bright, velocity
component. The transverse velocity component is dark because it is orthogonal
to the discrete wave vector and therefore has zero discrete divergence.

In the ordered basis
\(\{\ket{b_{\bm k}},\ket{d_{\bm k}},\ket p,\ket a\}\), the padded
Hamiltonian becomes
\begin{equation}\label{eq:bd-Htilde}
    \widetilde H_{\bm k}
    =\Omega_{\bm k}
      \left(\ket{b_{\bm k}}\!\bra p+\ket p\!\bra{b_{\bm k}}\right)
    =
    \begin{pmatrix}
        0&0&\Omega_{\bm k}&0\\
        0&0&0&0\\
        \Omega_{\bm k}&0&0&0\\
        0&0&0&0
    \end{pmatrix}.
\end{equation}
For the zero mode, \(\Omega_{\bm k}=0\), the entire block vanishes and its
evolution is the identity.

Let \(T_{\bm k}\) be the change-of-basis matrix whose columns are
\(\ket{b_{\bm k}},\ket{d_{\bm k}},\ket p,\ket a\). Then
\begin{equation}\label{eq:T-H-Tdag}
    \overline H_{\bm k}
    =T_{\bm k}\widetilde H_{\bm k}T_{\bm k}^\dagger.
\end{equation}
This transformation factorizes into a diagonal phase transformation and a real
Givens rotation,
\begin{equation}\label{eq:T-PG}
    T_{\bm k}=P_{\bm k}G_{\bm k},
\end{equation}
where
\begin{equation}\label{eq:P-def}
    P_{\bm k}
    =\operatorname{diag}\left(e^{i\phi_x},e^{i\phi_y},1,1\right)
\end{equation}
and
\begin{equation}\label{eq:G-def}
    G_{\bm k}
    =
    \begin{pmatrix}
        \cos\eta_{\bm k} & -\sin\eta_{\bm k} & 0 & 0\\
        \sin\eta_{\bm k} & \phantom{-}\cos\eta_{\bm k} & 0 & 0\\
        0&0&1&0\\
        0&0&0&1
    \end{pmatrix}.
\end{equation}
Consequently, the exact mode evolution is
\begin{equation}\label{eq:exact-bd-block-evolution}
    e^{-i\overline H_{\bm k}t}
    =
    P_{\bm k}G_{\bm k}
    e^{-i\widetilde H_{\bm k}t}
    G_{\bm k}^\dagger P_{\bm k}^\dagger.
\end{equation}

The remaining evolution is only a two-level rotation. Assigning the transformed
field states to
\begin{equation}
    \ket{b_{\bm k}}=\ket{00},
    \qquad
    \ket{d_{\bm k}}=\ket{01},
    \qquad
    \ket p=\ket{10},
    \qquad
    \ket a=\ket{11},
\end{equation}
gives
\begin{equation}\label{eq:bd-H-as-X}
    \widetilde H_{\bm k}
    =\Omega_{\bm k}X_{f_1}\otimes\ket0\!\bra0_{f_0}.
\end{equation}
With \(R_X(\theta)=e^{-i\theta X/2}\), its propagator is therefore
\begin{equation}\label{eq:Rx-bright-pressure}
    e^{-i\widetilde H_{\bm k}t}
    =
    R_{X_{f_1}}(2\Omega_{\bm k}t)\otimes\ket0\!\bra0_{f_0}
    +I_{f_1}\otimes\ket1\!\bra1_{f_0}.
\end{equation}
Similarly, the Givens transformation is an \(R_Y\) rotation on \(f_0\),
controlled on \(f_1=0\):
\begin{equation}\label{eq:G-as-controlled-Ry}
    G_{\bm k}
    =
    \ket0\!\bra0_{f_1}\otimes R_Y^{(f_0)}(2\eta_{\bm k})
    +\ket1\!\bra1_{f_1}\otimes I_{f_0}.
\end{equation}
Thus each nonzero Fourier block is compiled into one mode-dependent
bright--flux \(R_X\) rotation, one mode-dependent Givens \(R_Y\) rotation,
and the two phase functions contained in \(P_{\bm k}\).

The simplification is even greater for the pressure-only initial states used in
our hardware experiments,
\begin{equation}\label{eq:pressure-only-input}
    \ket{\Psi_p}
    =\sum_{\bm k}c_{\bm k}\ket{\bm k}\ket p.
\end{equation}
Because
\(P_{\bm k}^\dagger\ket p=G_{\bm k}^\dagger\ket p=\ket p\), the initial
change to the bright--dark basis acts trivially. Hence
\begin{equation}\label{eq:pressure-only-reduction}
    e^{-i\widehat H_{2\mathrm D}t}\ket{\Psi_p}
    =\mathcal P\mathcal G\mathcal R_x\ket{\Psi_p},
\end{equation}
where
\begin{equation}\label{eq:P-G-Rx-full}
    \mathcal P
    =\sum_{\bm k}\proj{\bm k}\otimes P_{\bm k},
    \qquad
    \mathcal G
    =\sum_{\bm k}\proj{\bm k}\otimes G_{\bm k},
\end{equation}
and
\begin{equation}
    \mathcal R_x
    =\sum_{\bm k}\proj{\bm k}\otimes
    \left[
        R_X(2\Omega_{\bm k}t)_{f_1}\otimes\ket0\!\bra0_{f_0}
        +I_{f_1}\otimes\ket1\!\bra1_{f_0}
    \right].
\end{equation}
For a general initial state containing velocity components, the full sequence
\(\mathcal P\mathcal G\mathcal R_x
\mathcal G^\dagger\mathcal P^\dagger\) is required.

The bright--dark basis itself is well established in optical control and has
also been used in the design of Raman quantum gates~\cite{torosov2020high}.
Our contribution is to recognize that every two-dimensional acoustic Fourier
block has precisely this structure and to convert its mode-dependent
Morris--Shore transformation into an explicit digital circuit. To our
knowledge, this reduction has not been used in previous quantum wave-equation
algorithms~\cite{wright2024noisy,lubasch2025fourier}. Unlike directional
splitting, it evolves the coupled \(x\)- and \(y\)-dynamics exactly within each
retained Fourier block. It therefore removes both the directional
product-formula error and the \(r\)-fold repetition of the directional
subcircuits. The price is a joint dependence on \((k_x,k_y)\), which we compile
explicitly below.

\subsubsection{Compilation of the nonlinear mode dependence}
\label{subsec:nonlinear-compilation}

The bright--dark reduction leaves two nonlinear, mode-dependent rotation
angles,
\begin{equation}\label{eq:nonlinear-angles}
    \alpha_{\bm k}
    =2\eta_{\bm k}
    =2\operatorname{atan2}(A_y,A_x),
    \qquad
    \beta_{\bm k}
    =2t\Omega_{\bm k}
    =2t\sqrt{A_x^2+A_y^2}.
\end{equation}
The phases \(\phi_x\) and \(\phi_y\) in \(P_{\bm k}\) are simpler because each
depends on only one Fourier index.

In the hardware implementation reported here, we do not evaluate the square
root or \(\operatorname{atan2}\) coherently. Instead, we exploit the fixed
low-frequency support of the initial states. We retain
\begin{equation}
    \mathcal K_1=\{0,1,2,N_h-2,N_h-1\}
\end{equation}
in each spatial direction, giving
\(\mathcal K=\mathcal K_1\times\mathcal K_1\). Thus
\(|\mathcal K|=25\) remains fixed as the spatial grid is refined.

The required values of
\(\alpha_{\bm k}\), \(\beta_{\bm k}\), \(\phi_x\), and \(\phi_y\) are
precomputed classically from the exact finite-difference symbols. The circuit
then proceeds in three stages. First, the reversible circuit \(U_{\rm mag}\)
extracts the sign and magnitude bits of \(k_x\) and \(k_y\). Second, Boolean
expansions of those bits implement \(\mathcal R_x\), \(\mathcal G\), and
\(\mathcal P\) as products of singly and doubly controlled rotations. Finally,
\(U_{\rm mag}^\dagger\) uncomputes the work registers. For the pressure-only
inputs considered here, the resulting circuit is
\(U_{\rm mag}^\dagger\mathcal P\mathcal G\mathcal R_xU_{\rm mag}\), as detailed
in \cref{subsec:restricted-implementation,eq:pressure-bd-circuit}.

This is the compilation strategy used in all two-dimensional hardware
experiments below. In particular, it is not a lookup over all \(N_h^2\) Fourier
modes: only the fixed retained set is synthesized. Moreover, because the exact
finite-difference values are precomputed on that set, the implemented
bright--dark evolution has no directional-splitting error and no
symbol-linearization error. Changing the evolution time modifies the
coefficients of the controlled \(R_X\) rotations but does not increase the
number of circuit layers.

For broader Fourier support, the same reduction could instead be combined with
reversible fixed-point arithmetic. Such a circuit would compute
\begin{equation}
    (k_x,k_y)\longmapsto A_x,A_y,\Omega_{\bm k},\eta_{\bm k},
\end{equation}
apply \(R_Y(2\eta_{\bm k})\) and \(R_X(2t\Omega_{\bm k})\), and then uncompute
the arithmetic registers. This provides a scalable extension of the
bright--dark construction, but it is not the strategy used in the present
hardware experiments.

For completeness, one may also simplify the arithmetic by replacing the exact
finite-difference magnitudes with their low-mode approximations,
\begin{equation}\label{eq:linearized-symbol-2d}
    A_\mu\approx2\pi|\kappa_\mu|,
    \qquad
    \Omega_{\bm k}\approx
    2\pi\sqrt{\kappa_x^2+\kappa_y^2},
    \qquad
    \eta_{\bm k}\approx
    \operatorname{atan2}(|\kappa_y|,|\kappa_x|),
\end{equation}
where \(\kappa_\mu\) is the signed wave number. Let
\(H_{\bm k}^{\rm lin}\) denote the block obtained from this linearized symbol.
For modes satisfying \(|\bm\kappa|\le K\),
\begin{equation}\label{eq:linearized-error-2d}
    \left\|
        e^{-iH_{\bm k}t}
        -e^{-iH_{\bm k}^{\rm lin}t}
    \right\|
    \le
    \frac{\pi^3t}{3N_h^2}|\bm\kappa|^3
    \le
    \frac{\pi^3t}{3N_h^2}K^3.
\end{equation}
Consequently, for an initial state \(\ket{\Psi_0}\) and the low-mode projector
\(\Pi_K\),
\begin{equation}\label{eq:lowmode-state-error}
\begin{aligned}
    \left\|
        U(t)\ket{\Psi_0}
        -U_{\rm lin}(t)\ket{\Psi_0}
    \right\|
    &\le
    2\left\|(I-\Pi_K)\ket{\Psi_0}\right\| \\
    &\quad
    +\frac{\pi^3t}{3N_h^2}K^3
    +\epsilon_{\rm synth},
\end{aligned}
\end{equation}
where \(\epsilon_{\rm synth}\) accounts for finite-precision arithmetic and
rotation-synthesis errors. This estimate applies to the optional linearized
construction; the retained-mode circuits used in our experiments employ the
exact finite-difference values described above.

\subsection{Dirac dynamics equation via Trotterization}\label{sec:kg-methods}

We now turn to the linearized Dirac dynamics equation \cref{eq:kg-first}. Since the two operators do not commute, a single QFT can not diagonalize the Hamiltonian.  Hence we implement the propagator by a symmetric (Strang) splitting,
\begin{equation}\label{eq:strang}
  e^{-iH_{\mathrm{D}}t}
  \approx\Bigl[\,e^{-iH_{\mathrm{mass}}\tau/2}\,
               e^{-iH_{\mathrm{wave}}\tau}\,
               e^{-iH_{\mathrm{mass}}\tau/2}\Bigr]^{t/\tau},
\end{equation}
with time step $\tau$ and a local error of order $\tau^{3}$ per step. For our specific operator, we show the commutator error as follows.

\begin{Prop}[Explicit commutators for the Dirac Strang splitting]
Let
\begin{equation}
  A=H_{\mathrm{mass}}
  =
  \begin{pmatrix}
    0&iM\\
    -iM&0
  \end{pmatrix},
  \qquad
  B=H_{\mathrm{wave}}
  =
  \begin{pmatrix}
    0&iD\\
    -iD^\dagger&0
  \end{pmatrix},
\end{equation}
where $M=M^\dagger$ is the diagonal mass matrix. Define
\begin{align}
  X_A
  &=
  M^2D+DM^2-2MD^\dagger M,
  \label{eq:dirac-XA}\\
  X_B
  &=
  2DMD-DD^\dagger M-MD^\dagger D.
  \label{eq:dirac-XB}
\end{align}
Then the nested commutators appearing in the Strang error have the
explicit block forms
\begin{align}
  [A,[A,B]]
  &=
  \begin{pmatrix}
    0&iX_A\\
    -iX_A^\dagger&0
  \end{pmatrix},
  \label{eq:dirac-double-commutator-A}\\
  [B,[A,B]]
  &=
  \begin{pmatrix}
    0&iX_B\\
    -iX_B^\dagger&0
  \end{pmatrix}.
  \label{eq:dirac-double-commutator-B}
\end{align}
Consequently, the Strang commutator constant can be written exactly as
\begin{equation}
  \mathcal C_{\mathrm S}
  =
  \norm{X_A}
  +
  2\norm{X_B}.
  \label{eq:dirac-strang-constant-explicit}
\end{equation}
The one-step error therefore satisfies
\begin{equation}
  \left\|
    e^{-i(A+B)\tau}
    -
    e^{-iA\tau/2}e^{-iB\tau}e^{-iA\tau/2}
  \right\|
  \leq
  \frac{\tau^3}{24}
  \left(
    \norm{X_A}
    +
    2\norm{X_B}
  \right).
  \label{eq:dirac-strang-local-error}
\end{equation}
After $r=t/\tau$ steps, unitarity and a telescoping argument give
\begin{equation}
  \left\|
    e^{-i(A+B)t}
    -
    S_2(\tau)^r
  \right\|
  \leq
  \frac{t^3}{24r^2}
  \left(
    \norm{X_A}
    +
    2\norm{X_B}
  \right).
  \label{eq:dirac-strang-global-error}
\end{equation}
\end{Prop}

\begin{proof}
Direct block multiplication first gives
\begin{equation}
  [A,B]
  =
  \begin{pmatrix}
    MD^\dagger-DM&0\\
    0&MD-D^\dagger M
  \end{pmatrix}.
\end{equation}
Commuting this block-diagonal matrix once more with $A$ and $B$
produces \cref{eq:dirac-XA,eq:dirac-XB}. Moreover, for any matrix $X$,
\begin{equation}
  \left\|
  \begin{pmatrix}
    0&iX\\
    -iX^\dagger&0
  \end{pmatrix}
  \right\|
  =
  \norm{X},
\end{equation}
which follows by squaring the block matrix. The stated estimates then
follow from the standard unitary Strang bound and telescoping over the
$r$ time steps.
\end{proof}

\begin{remark}[Bulk and interface contributions]
The explicit commutators distinguish the bulk splitting error from
the additional error generated by spatial variations of the mass. For
the periodic forward-difference operator, $D$ is normal, and hence
\begin{equation}
  L_h
  :=
  DD^\dagger
  =
  D^\dagger D.
\end{equation}
Equations~\eqref{eq:dirac-XA} and \eqref{eq:dirac-XB} can then be
rewritten as
\begin{align}
  X_A
  &=
  2M^2(D-D^\dagger)
  +
  [D,M^2]
  -
  2M[D^\dagger,M],
  \label{eq:dirac-XA-bulk-interface}\\
  X_B
  &=
  2M(D^2-L_h)
  +
  2[D,M]D
  -
  [L_h,M].
  \label{eq:dirac-XB-bulk-interface}
\end{align}
The first term in each expression is a bulk contribution. The remaining
commutators contain discrete derivatives of the mass. Indeed,
\begin{equation}
  [D,M]_{jk}
  =
  D_{jk}(m_k-m_j),
\end{equation}
so $[D,M]$ is nonzero only on difference stencils joining grid points
with different mass values. 
\end{remark}

As a specific test, we specialize the general space-dependent mass $m(x)\ge0$ of \cref{eq:kg} to a step function on the periodic domain $x\in(0,1)$, with a  discontinuity at the midpoint (and one at the periodic boundary), \begin{equation}\label{eq:kg-mass-profile} m(x)= \begin{cases} m_{-}, & 0<x<\tfrac12,\\ m_{+}, & \tfrac12\le x<1, \end{cases} \end{equation}

\section{Resource scaling and hardware results}\label{sec:results}

\subsection{One-dimensional acoustic wave equation}

\subsubsection{Resource scaling}
\label{subsec:1d-resources}

\begin{figure}[htbp]
    \centering
    \includegraphics[width=0.95\linewidth]{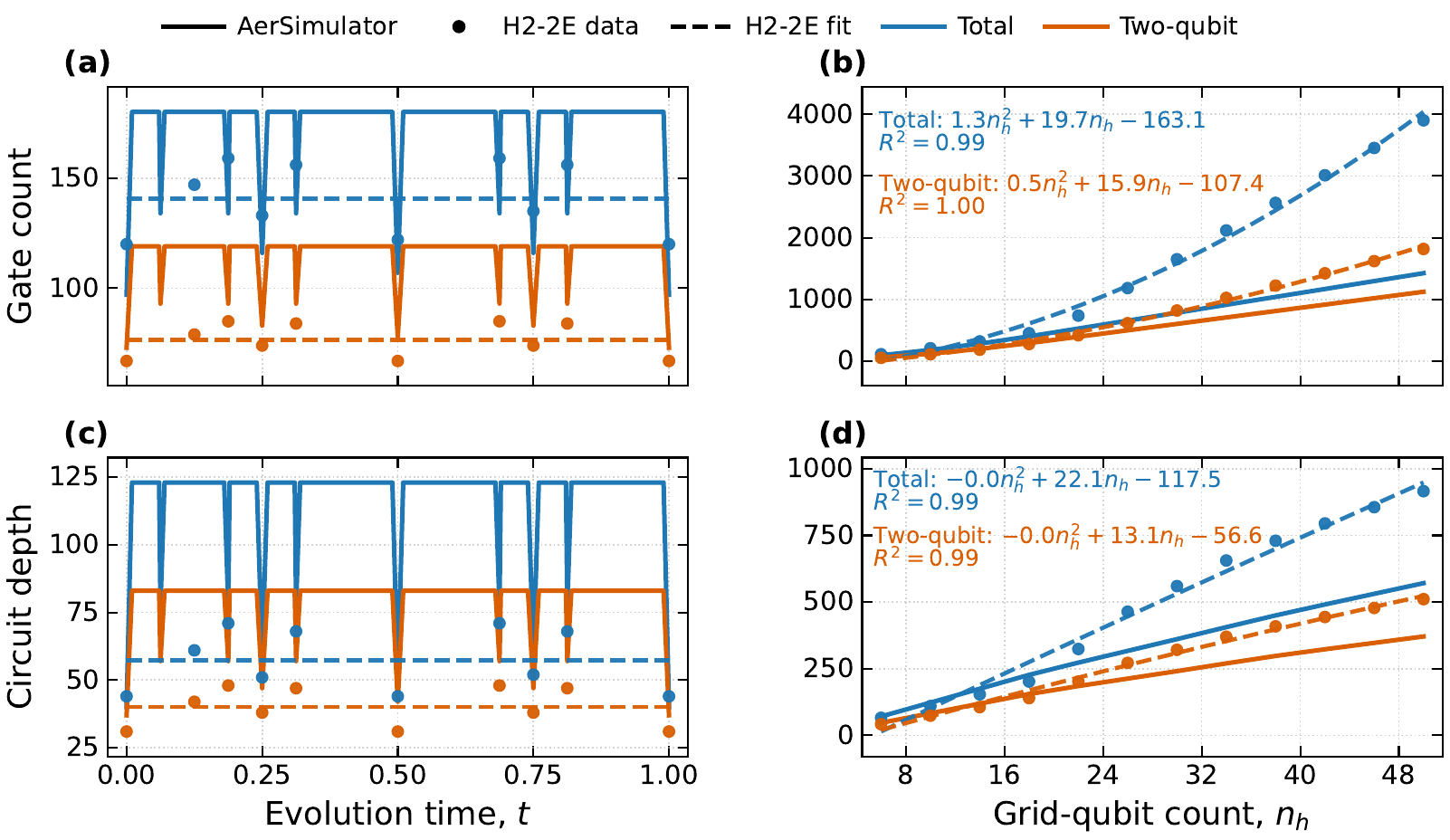}
    \caption{Circuit resources for the one-dimensional cosine pressure initial
state with $k_{0}=1$. Solid curves show the logical circuit before
hardware-native compilation; markers show the circuit compiled for the
Quantinuum H2-2 native gate set using H2-2E, with dashed curves denoting
least-squares fits. Blue denotes total gate count or depth, and orange denotes
the corresponding two-qubit contribution. (a),(c) Gate count and circuit
depth, respectively, as functions of $t\in[0,1]$ at fixed $n_{h}=10$.
(b),(d) Gate count and circuit depth as functions of the number $n_h$ of grid
qubits at fixed $t=0.1$.}
    \label{fig:1d_cosine_resources}
\end{figure}

For the cosine pressure initial state with a single mode $k_{0}=1$, the explicit circuit schematics are given in \cref{fig:hz-tilde-circuit,fig:cosine-circuit}. Firstly we study scaling with evolution time, shown in  \cref{fig:1d_cosine_resources} ([panels (a) and (c)),
after transpilation to the native Quantinuum H2-2 gate set corresponding to grid size $N_h$ = 1024 ($n_h$ = 10), the compiled
circuit uses between 120 and 159 gates, of which 67-85 are two-qubit gates,
and reaches a total depth of between 44 and 71 (two-qubit depth 31-48). Neither quantity grows with $t$; instead both are
periodic with period $\Delta t=0.5$, attaining minima at $t=0,0.5,1$. This periodicity is a direct
consequence of the mode structure derived in \cref{subsec:1d-linearization}:
for the dominant mode $k_{0}=1$ populated by the cosine initial state, the
linearized rotation angle of \cref{eq:Uapp} is $-4\pi t$,
which returns to a multiple of $2\pi$ every $\Delta t=0.5$; at these times
several of the compiled single-qubit rotations become trivial and are removed
during compilation, reducing both gate count and depth. The same
periodicity is visible, in the logical circuit
reported by AerSimulator, whose total gate count collapses from its generic
value of $180$ (two-qubit: $119$) to as low as $97$ (two-qubit: $73$) at these times. These results confirm that QFT allows the wave dynamics can be simulated at almost constant depth, due to the fast-forwarding by QFT.

Next, we examine the scaling with grid size with results shown in \cref{fig:1d_cosine_resources} ([panels (b) and (d)).
At the fixed time $t=0.1$ we scan the register size over
$n_{h}\in\{6,10,14,\dots,50\}$, in order to capture asymptotic behavior with respect to grid-qcount. For the
logical circuit (AerSimulator), a least-squares quadratic fit gives
total gate count $\approx0.08\,n_{h}^{2}+26.3\,n_{h}-86.3$
($R^{2}=0.99$) which is consistent with the $O(n_{h}^{2})$ count of controlled-phase
gates contributed by QFT (\cref{subsec:qft}) that dominate the
circuit and total depth $\approx-0.04\,n_{h}^{2}+13.6\,n_{h}-7.9$
($R^{2}=0.99$), which is almost a linear scaling with $n_h$.

For the compiled circuit on quantinuum's H2-2E, the fits annotated in \cref{fig:1d_cosine_resources} (b) and (d)
are total gate count $\approx1.3\,n_{h}^{2}+19.7\,n_{h}-163.1$
($R^{2}=0.99$) and total depth $\approx-0.0\,n_{h}^{2}+22.1\,n_{h}-117.5$
($R^{2}=0.99$), again a clearly quadratic gate count against an
essentially linear depth, in the same qualitative split as for the logical
circuit. 
Unlike the logical circuit, however, the raw H2-2E data is not
smooth: between $n_{h}=18$ and $n_{h}=22$, total gate count jumps by $285$
(versus increments of $96$-$132$ over each preceding step of the same size)
and total depth jumps by $122$ (versus $44$--$48$ before), after which both
quantities settle into a new, steeper near-linear trend for $n_{h}\ge22$. A plausible explanation is that
the QCCD architecture's effective gate-level parallelism
\cite{moses2023racetrack} is limited by the number of physical zones in which a
two-qubit gate can execute concurrently; once the register requires more
simultaneous operations than the device can service in parallel, additional
gates must serialize, steepening both the gate-count and depth. 

\begin{figure}[htbp]
    \centering
    \includegraphics[width=0.95\linewidth]{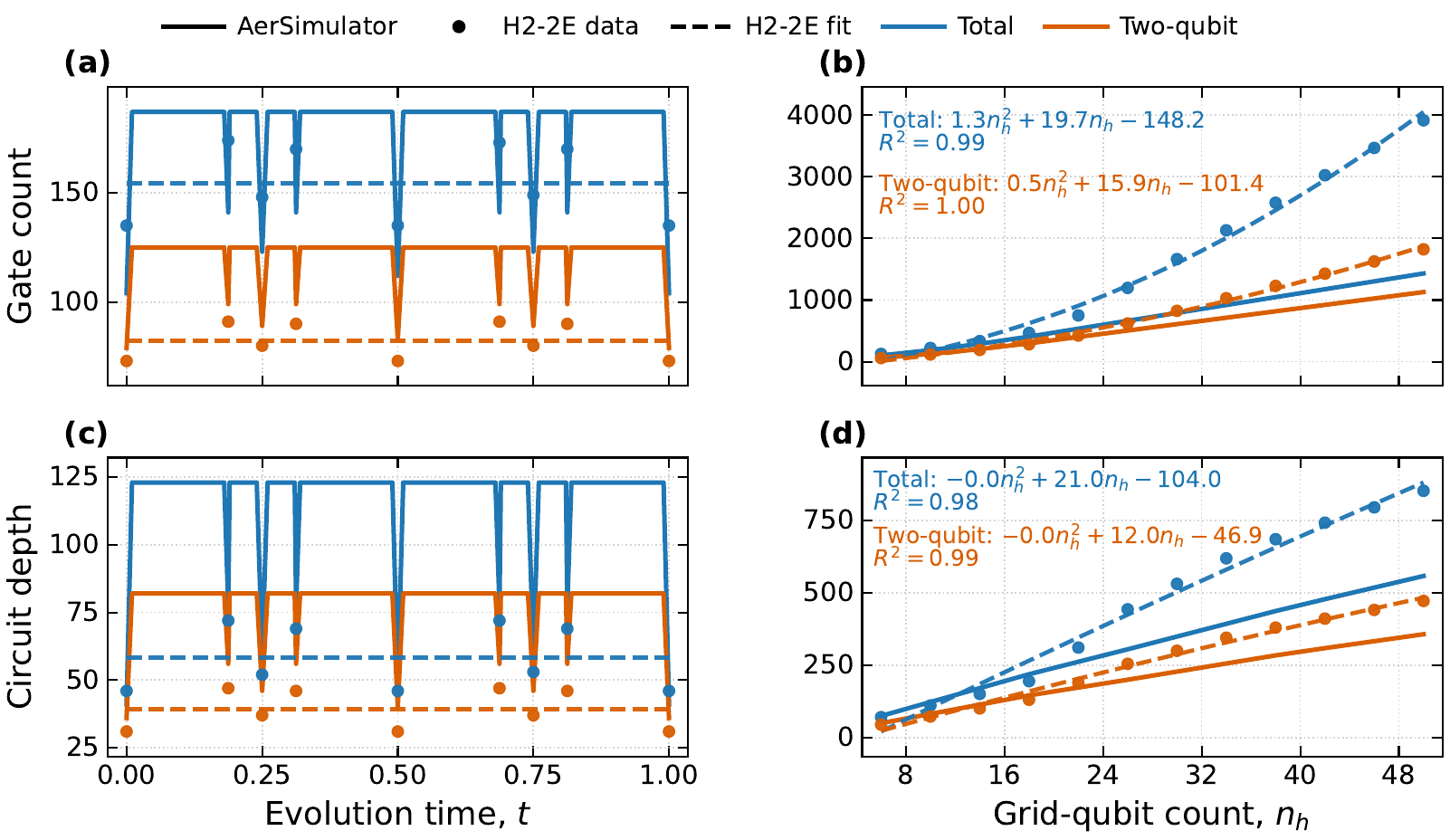}
\caption{Circuit resources for the one-dimensional Gaussian pressure initial
state with $\sigma=0.2$ and five retained Fourier modes. Solid curves show
the logical circuit before hardware-native compilation; markers show the
circuit compiled for the Quantinuum H2-2 native gate set using H2-2E, with
dashed curves denoting least-squares fits. Blue denotes total gate count or
depth, and orange denotes the corresponding two-qubit contribution.
(a),(c) Gate count and circuit depth, respectively, as functions of
$t\in[0,1]$ at fixed $n_h=10$. (b),(d) Gate count and circuit depth as
functions of $n_h$ at fixed $t=0.1$.}
    \label{fig:1d_gaussian_resources}
\end{figure}

We next tested a Gaussian pressure initial state with $\sigma = 0.2$, which has about  5 significant modes in the Fourier basis. The explicit circuit schematic can be seen in \cref{fig:gaussian-circuit}. The transpiled gate count and depth are modestly
higher throughout at $n_{h}=10$: about 135 to 176 total gates (73-91 two-qubit)
and depth 46 to 72 (two-qubit depth 31-47), but
retain the same period-$0.5$ structure, since the extra state-preparation
gates are $t$-independent. The grid-qubit count scan shows the same quadratic gate-count,
near-linear depth split. The Gaussian
circuit uses $14$--$15$ more transpiled gates than the cosine circuit at every
sampled $n_{h}$ from $6$ to $50$,
confirming that the extra state-preparation cost does not itself grow with
the grid size beyond the five fixed rotation gates identified in
\cref{subsec:1d-fivemode}. The \emph{depth} offset, by contrast, is not
constant, which could be attributed to compiler specific optimization.

\subsubsection{Hardware results}
\label{subsec:1d-hardware}

Rather than reconstructing the velocity and flux fields $v(x),p(x)$ at every
grid point $x\in(0,1)$, we evaluate a single physically meaningful observable,
the kinetic energy on a subdomain $X\subset(0,1)$. Because the full measurement
count distribution is available, the kinetic energy can be computed on any
subdomain in post-processing; we report the half-domain $X=(0,\tfrac12)$, i.e.\
the grid points $j=0,\dots,N_{h}/2-1$.

The kinetic-energy observable is the projector onto these velocity components,
\begin{equation}\label{eq:KE-def}
    \mathrm{KE}=\sum_{j=0}^{N_{h}/2-1}\ket{v_{j}}\!\bra{v_{j}}.
\end{equation}
In the register layout, the velocity components sit in the field qubit state
$\ket{0}$, so $\ket{v_{j}}=\ket{j}\otimes\ket{0}$ and
\begin{equation}\label{eq:KE-block}
    \langle\mathrm{KE}\rangle=\sum_{j=0}^{N_{h}/2-1}\ket{j}\!\bra{j}\otimes\ket{0}\!\bra{0}.
\end{equation}
Restricting $j$ to the first half of the grid fixes the most-significant qubit(MSB) to $\ket{0}$; the projector therefore factorizes as
\begin{equation}\label{eq:KE-factored}
    \langle\mathrm{KE}\rangle
    =\ket{0}\!\bra{0}_{\mathrm{MSB}}
     \otimes\Bigl(\sum_{r=0}^{N_{h}/2-1}\ket{r}\!\bra{r}\Bigr)
     \otimes\ket{0}\!\bra{0}_{\mathrm{field}}
    =\ket{0}\!\bra{0}_{\mathrm{MSB}}\otimes I_{n_{h}-1}\otimes\ket{0}\!\bra{0}_{\mathrm{field}}.
\end{equation}
Evaluating $\langle\mathrm{KE}\rangle$ then reduces to summing the measured
marginal probabilities over all outcomes with the MSB and the field qubit
both equal to $0$, which is immediate from the count distribution.
\Cref{fig:1d_cosine_ke} shows the resulting kinetic energy as a function of
evolution time, for four sources: the classical reference
trajectory of \cref{sec:experimental-setup}; the noiseless AerSimulator
circuit, which differs from the classical reference only by finite-shot
sampling statistics; the H2-2E noise-modeled emulator; and the trapped-ion H2-2
processor.

\begin{figure}[htbp]
    \centering
    \includegraphics[width=0.95\linewidth]{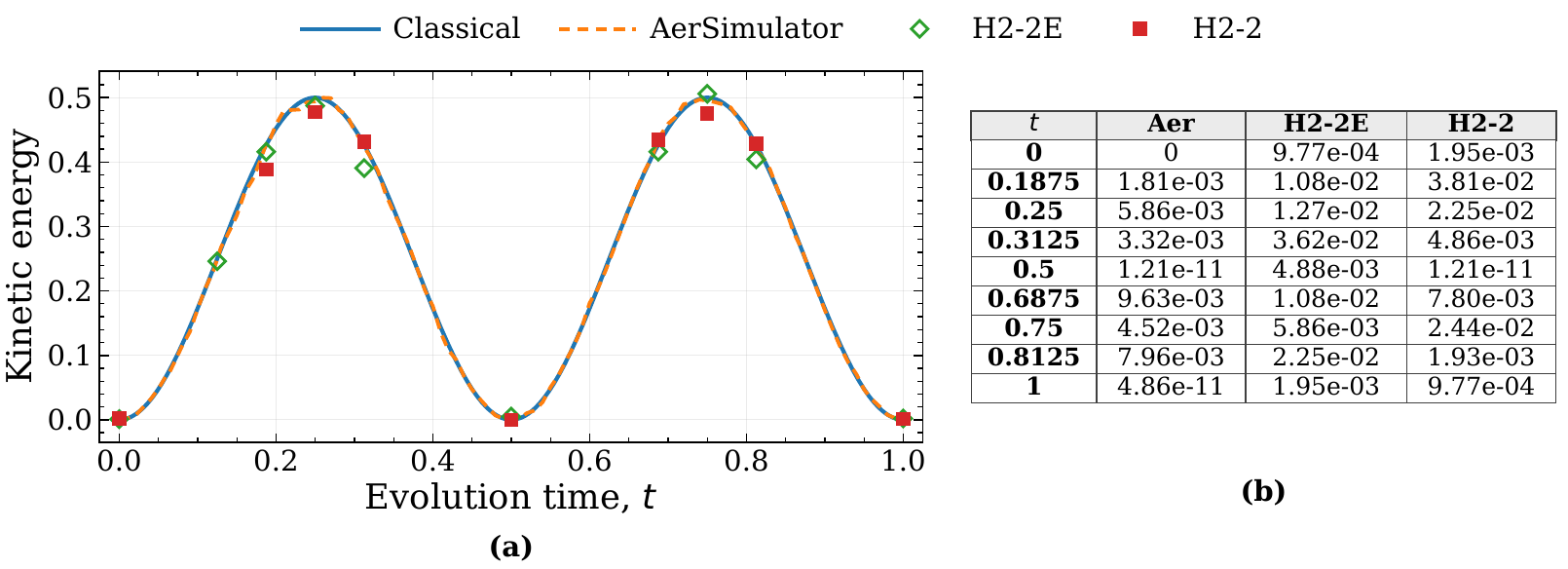}
\caption{Half-domain kinetic energy for the one-dimensional cosine pressure
initial state with $k_{0}=1$ on a grid of $N_h=1024$ points
($n_h=10$ grid qubits and 11 qubits in total). (a) The classical reference
trajectory is compared with noiseless finite-shot AerSimulator results, the
H2-2E noise-modeled emulator, and measurements from the Quantinuum H2-2
processor. (b) Absolute error
$\lvert\langle\mathrm{KE}\rangle_{\mathrm{method}}
-\langle\mathrm{KE}\rangle_{\mathrm{classical}}\rvert$
at each sampled time. A dash denotes a time not sampled on the corresponding
backend. The mean absolute errors are $3.7\times10^{-3}$, $1.2\times10^{-2}$,
and $1.1\times10^{-2}$ for AerSimulator, H2-2E, and H2-2, respectively.}
    \label{fig:1d_cosine_ke}
\end{figure}

The simulation was performed over grid size of$N_h = 1024$. Averaged over the nine common sample points in the time interval $[0,1]$. The mean absolute error is
$3.7\times10^{-3}$ for AerSimulator (attributed to shot-noise), $1.2\times10^{-2}$
for H2-2E, and $1.1\times10^{-2}$ for H2-2 - i.e.\ the calibrated emulator
and the physical device agree with the classical reference to essentially the
same degree, and hardware noise contributes roughly three times
the AerSimulator reference. Against the kinetic energy's maximum of $0.5$ at t = 0.25 and t = .75, this amounts to a mean relative error of order $2$-$3\%$
on H2-2, with a worst case of $3.81\times10^{-2}$ (relative error $\approx9\%$)
at $t=0.1875$.

\begin{figure}[htbp]
    \centering
    \includegraphics[width=0.95\linewidth]{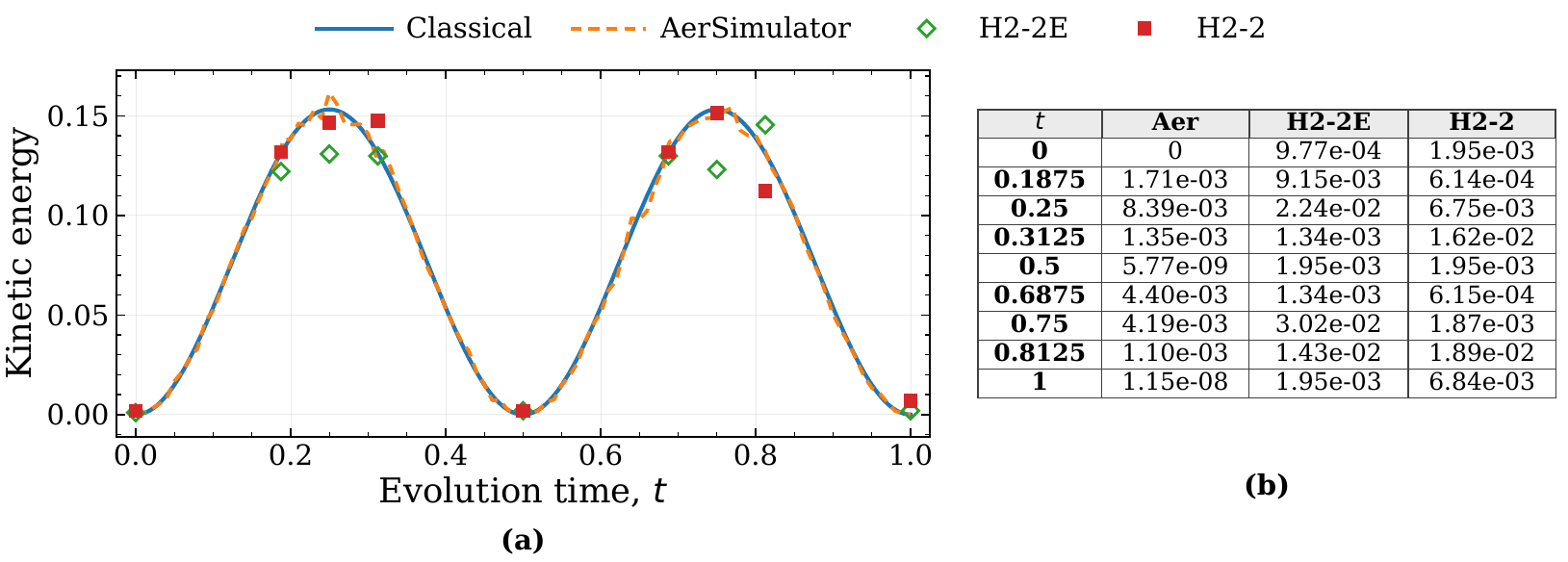}
\caption{Half-domain kinetic energy for the one-dimensional Gaussian pressure
initial state with $\sigma=0.2$ on a grid of $N_h=1024$ points
($n_h=10$ grid qubits and 11 qubits in total). (a) Classical reference,
noiseless finite-shot AerSimulator, H2-2E noise-modeled emulator, and
Quantinuum H2-2 results. (b) Absolute error relative to the classical
reference at each sampled time. The mean absolute errors are
$2.3\times10^{-3}$, $9.3\times10^{-3}$, and $5.9\times10^{-3}$ for
AerSimulator, H2-2E, and H2-2, respectively.}
    \label{fig:1d_gaussian_ke}
\end{figure}

For the Gaussian state, the mean absolute error over its set of
common sample points is $2.3\times10^{-3}$ for AerSimulator, $9.3\times10^{-3}$
for H2-2E, and $5.9\times10^{-3}$ for H2-2, it is to be noted that these absolute
errors are in \emph{smaller} than its cosine counterpart, despite the
Gaussian circuit using a nearly identical number of gates depth( see \cref{subsec:1d-resources} ). Relative
to the compressed signal, however, a mean absolute error of $5.9\times10^{-3}$
against a peak of $0.153$ is a mean relative error of order
$4\%$, roughly double the cosine case, with a worst case of
$1.89\times10^{-2}$ (relative error $\approx14\%$) at $t=0.8125$. We also note, that for
the Gaussian state the noise-modeled emulator's mean error exceeds that of the physical device itself, which is the reverse of the naive expectation that hardware should be noisier than its own calibrated emulator, and a reminder that the emulator's noise model need not bound the physical device's behavior specifically.

Taken together, these results indicate that circuit depth and gate count
alone are exhaustive predictors of observed noise resilience, the dynamic range of the observable
being measured should also be taken into account. Resource estimates reported purely in terms of gate count and
depth, as in \cref{tab:scaling}, should therefore be read alongside the dynamic range of the specific observable and initial condition under study. Further, in comparison to earlier hardware demonstrations of the one-dimensional wave equation, the present results access substantially larger system sizes: Sato \textit{et al.}~\cite{sato2024hamiltonian} implemented a hardware run using $n_h=2$ grid qubits ($N_h=4$ grid points) on \texttt{ibm\_kawasaki}, a scale the authors attribute to noise-resilience limitations of current hardware, while Wright \textit{et al.}~\cite{wright2024noisy} demonstrated a smooth Ricker-wavelet initial condition on $n_h = 6$ grid qubits using the Quantinuum H1-1 processor. The experiments reported here extend this to $n_h = 10$ grid qubits on the Quantinuum H2-2 processor - a $16 \times 16$ increase in grid resolution over Wright \textit{et al.}~\cite{wright2024noisy}.

\subsection{Two-dimensional acoustic wave equation}

\subsubsection{Resource scaling}
\label{subsec:2d_resource}

\begin{figure}[htbp]
    \centering
    \includegraphics[width=0.95\linewidth]{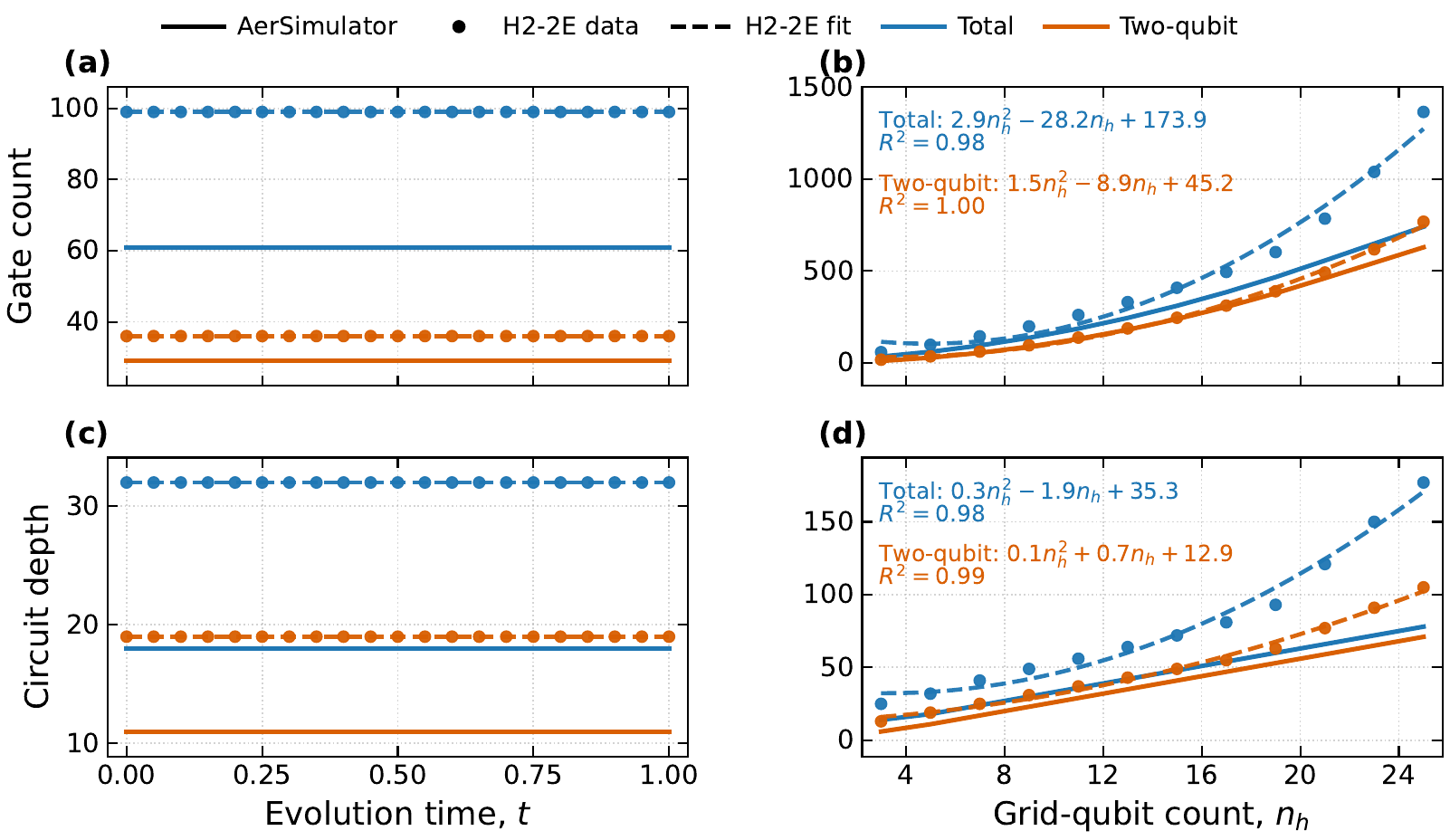}
\caption{Circuit resources for the two-dimensional separable cosine pressure
initial state with $k_x=k_y=1$. Solid curves show the logical circuit before
hardware-native compilation; markers show the circuit compiled for the
Quantinuum H2-2 native gate set using H2-2E, with dashed curves denoting
least-squares fits. Blue denotes total gate count or depth, and orange denotes
the corresponding two-qubit contribution. (a),(c) Gate count and circuit
depth, respectively, as functions of $t\in[0,1]$ at fixed $n_h=5$ grid
qubits per spatial direction. (b),(d) Gate count and circuit depth as
functions of the per-direction register size $n_h$ at fixed $t=0.1$.}
    \label{fig:2d_cosine_resources}
\end{figure}

Similar to the one-dimensional case, we first consider the cosine initial condition for the flux with  $k_{x}=k_{y}=1$. The detailed circuit construction is described in \cref{subsec:2d-initial-state}. Further, note that there are only 2 Fourier  modes retained in each direction, the Hamiltonian circuit is simplified accordingly. The resource estimation here follows
the same protocol as \cref{subsec:1d-resources}. 

We begin with the report of the resource scaling with evolution time (\cref{fig:2d_cosine_resources} panels (a) and (c)). 
At $n_{h}=5 $  per dimension $(N=2048)$, the compiled circuit uses exactly $99$ gates
($36$ two-qubit) and reaches a depth of exactly $32$ ($19$ two-qubit gates); the
logical circuit is likewise constant, at $61$ gates ($29$ two-qubit)
and depth $18$ ($11$ two-qubit). This can be attributed to fact that the circuit is specific to 4 Fourier  modes, corresponding to the cosine initial condition.

Next, we study resource scaling with grid size in \cref{fig:2d_cosine_resources} (panels b and d) for the logical circuit.  A quadratic fit gives total gate count
$\approx0.92\,n_{h}^{2}+6.7\,n_{h}+4.0$ ($R^{2}=0.99$) and total depth
$\approx0.007\,n_{h}^{2}+2.8\,n_{h}+4.7$ ($R^{2}=0.99$, linear fit
$R^{2}=0.99$). Thus, as in the one-dimensional case, gate count is clearly
quadratic while depth is nearly linear with respect to $n_h$. For the compiled circuit (H2-2E), however, the results are different
from \cref{subsec:1d-resources}. In particular, the fits annotated in \cref{fig:2d_cosine_resources} (b) and (d) are
total gate count $\approx2.9\,n_{h}^{2}-28.2\,n_{h}+173.9$ ($R^{2}=0.98$) and
total depth $\approx0.29\,n_{h}^{2}-1.9\,n_{h}+35.3$ ($R^{2}=0.98$). Here
the quadratic scaling for depth is clearly evident as opposed to linear scaling in one-dimensional case.
 A plausible reason is that the two-dimensional circuit couples
two separate $n_{h}$-qubit spatial registers through the bright--dark
reduction (\cref{subsec:bright-dark-reduction}), so a larger fraction of its
two-qubit gates span both grid registers and cannot be scheduled independently in
the way that a single grid register's internal QFT gates can. On the other hand, similar to the
one-dimensional case, we also observe a further steepening in the last few
sampled points ($n_{h}\gtrsim19$), consistent with the same
compiler threshold effect discussed in
\cref{subsec:1d-resources}.

\begin{figure}[htbp]
    \centering
    \includegraphics[width=0.95\linewidth]{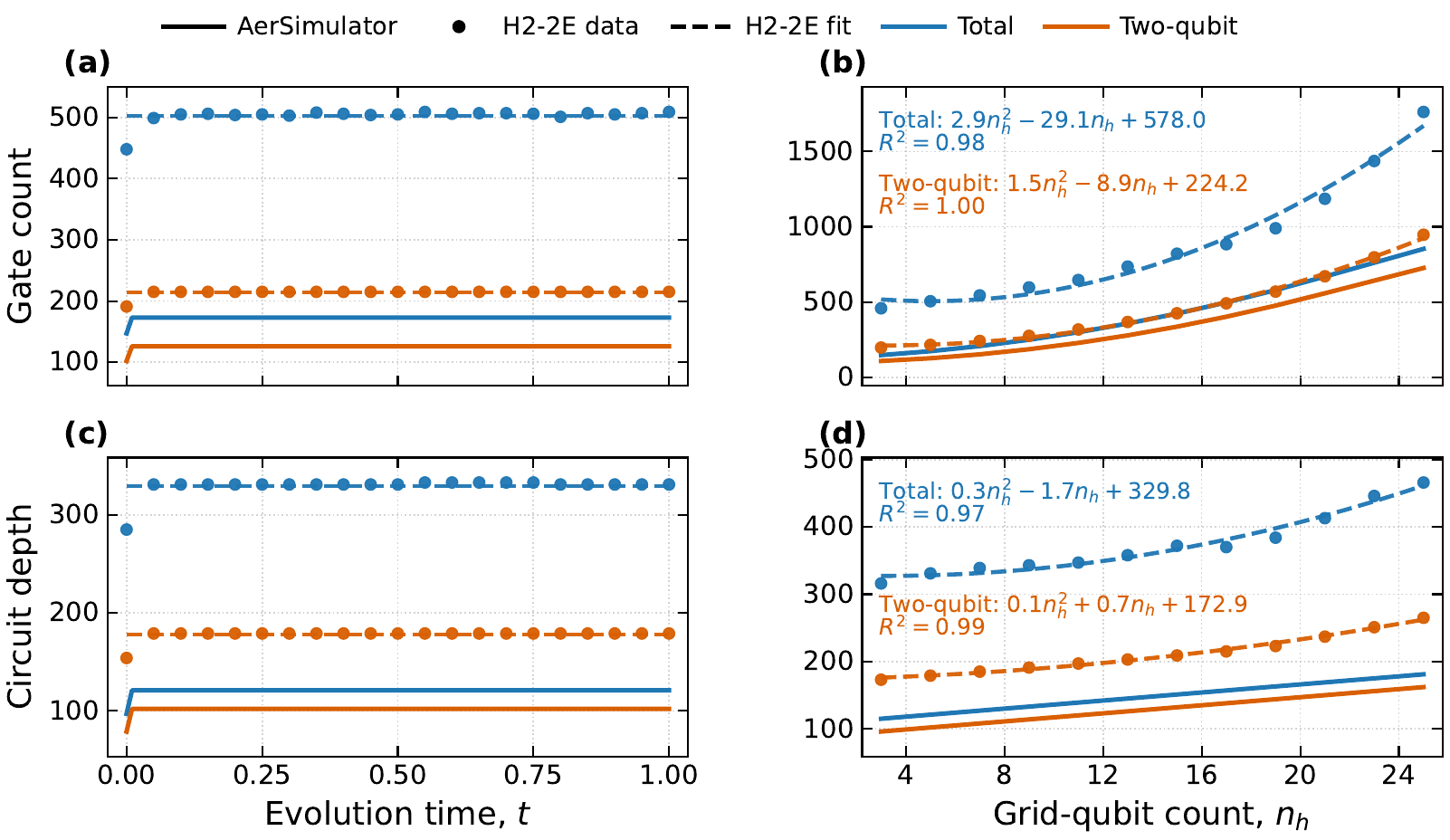}
\caption{Circuit resources for the two-dimensional separable Gaussian
pressure initial state with $\sigma_x=\sigma_y=0.2$ and five retained Fourier
modes in each spatial direction. Solid curves show the logical circuit before
hardware-native compilation; markers show the circuit compiled for the
Quantinuum H2-2 native gate set using H2-2E, with dashed curves denoting
least-squares fits. Blue denotes total gate count or depth, and orange denotes
the corresponding two-qubit contribution. (a),(c) Gate count and circuit
depth, respectively, as functions of $t\in[0,1]$ at fixed $n_h=5$ per
direction. (b),(d) Gate count and circuit depth as functions of $n_h$ at
fixed $t=0.1$.}
    \label{fig:2d_Gaussian_resources}
\end{figure}

For Gaussian initial profile, we take $\sigma = 0.2$, more details on circuit construction could be found in \cref{subsec:2d-initial-state}. Note that here the 5 modes are retained in each direction leading to a total of 25 pairs, so we use Hamiltonian circuit construction described in \cref{subsec:restricted-implementation}. 
As, it can be observed in \cref{fig:2d_Gaussian_resources} (panels a and c), the compiled gate count and depth at
$n_{h}=5$ remains constant for almost all the sampled times ( $509$ total gates, $215$ two-qubit gates,
depth $333$ and two-qubit depth $179$) except at t = 0, where both gate count and depth are significantly lower, which can be attributed to the fact that all the time-dependent rotation angles become 0, so compiler specific optimization leads to simplified circuit similar to one-dimensional case.

The resources are substantially higher than the cosine case, which is attributed to relatively large number of retained Fourier  modes retained. The
grid-size scan shows the same qualitative quadratic gate-count,
quadratic depth split as the cosine state - fits
$\approx2.9\,n_{h}^{2}-29.1\,n_{h}+578.0$ ($R^{2}=0.98$) for total gate count
and $\approx0.28\,n_{h}^{2}-1.7\,n_{h}+329.8$ ($R^{2}=0.97$) for total depth -
offset upward by an amount consistent with the additional separable
state-preparation circuit of \cref{subsec:2d-initial-state} and more modes in Hamiltonian circuit.

\medskip

Finally, we consider a nonseparable initial
flux profile, correlated between the two spatial directions,
\begin{equation}\label{eq:entangled-profile}
    p(x,y)\propto\exp\!\Bigl[\kappa\bigl(\cos 2\pi(x-\tfrac12)+\cos 2\pi(y-\tfrac12)-2\bigr)
    +\gamma\bigl(\cos 2\pi(x-y)-1\bigr)\Bigr],
\end{equation}
with $\kappa=1.0$ and $\gamma=0.4$. The parameter $\gamma$ couples $x$ and $y$ so that the flux can not be factorized. This profile is treated by the low-rank state preparation ~\cite{Araujo_2024} of
\cref{subsec:2d-initial-state} (the rank-two preparation circuit
$U_{\mathrm{init}}^{(2)}$ of \cref{eq:Uinit2}) described in \cref{subsec:2d-initial-state}. 

\begin{figure}[htbp]
    \centering
    \includegraphics[width=0.95\linewidth]{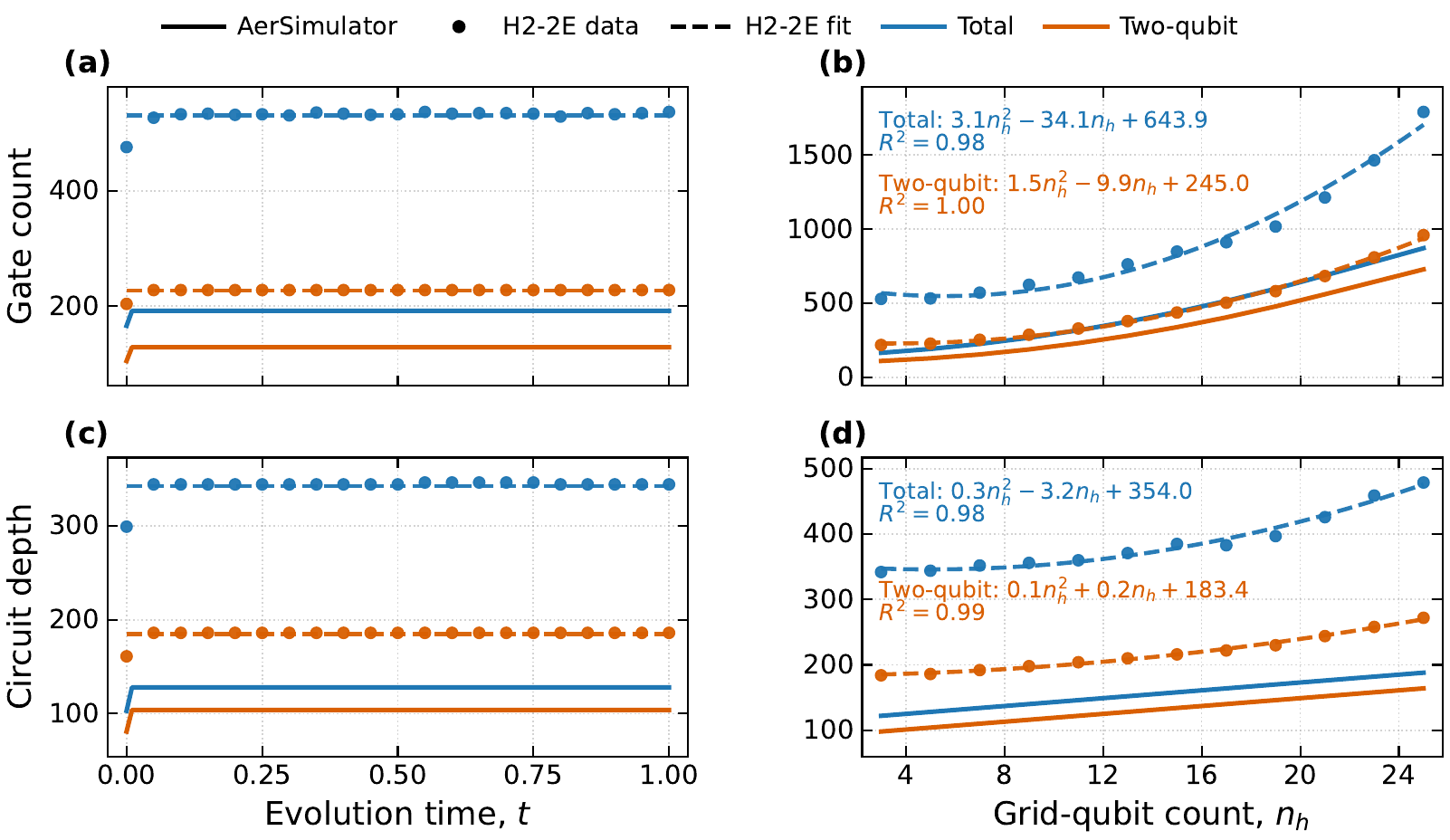}
\caption{Circuit resources for the two-dimensional nonseparable pressure
initial state in \cref{eq:entangled-profile}, with $\kappa=1.0$,
$\gamma=0.4$, and a rank-two state-preparation circuit. Solid curves show the
logical circuit before hardware-native compilation; markers show the circuit
compiled for the Quantinuum H2-2 native gate set using H2-2E, with dashed
curves denoting least-squares fits. Blue denotes total gate count or depth,
and orange denotes the corresponding two-qubit contribution. (a),(c) Gate
count and circuit depth, respectively, as functions of $t\in[0,1]$ at fixed
$n_h=5$ grid qubits per spatial direction. (b),(d) Gate count and circuit
depth as functions of $n_h$ at fixed $t=0.1$.}
    \label{fig:2d_entangled_resources}
\end{figure}

As it can be observed in \cref{fig:2d_entangled_resources} (panels a and c), with $n_{h}=5$, the resources of compiled circuit remains almost constant around $537$ total gates
($228$ two-qubit) and depth $346$ (two-qubit depth
$186$) across the sampled times except at t = 0, notably higher than both the cosine
and the separable Gaussian state, consistent with the rank-two construction. 

The grid-size scan again shows a
quadratic scaling in gate count and depth, with fits
$\approx3.1\,n_{h}^{2}-34.1\,n_{h}+643.9$ ($R^{2}=0.98$) for total gate count
and $\approx0.33\,n_{h}^{2}-3.2\,n_{h}+354.0$ ($R^{2}=0.98$) for total depth. Again, the offset across the intial states considered is almost constant as the circuit construction has been done on the retained modes only.

\subsubsection{Hardware results}
\label{subsec:2d_hardware}

As in \cref{subsec:1d-hardware}, we evaluate a single physically meaningful
observable rather than reconstructing the full field: the kinetic energy
carried by the $v_{x}$ component on a spatial subdomain $x \in [0, 1/2)$. In the field
encoding of \cref{eq:2d-field-labels}, $v_{x}$ corresponds to both field qubit state being $\ket{00}$, so
\begin{equation}\label{eq:2d-KE-def}
    \langle \mathrm{KE}_{v_{x}} \rangle
    =I_{y}\otimes \ket{0}\bra{0}_{MSB_x} \otimes I_{n_h - 1}\otimes \ket{00}\!\bra{00}_{f},
\end{equation}
for a subdomain $X$ of the $x$-register, evaluated identically to the
one-dimensional case by summing measured marginal probabilities over
outcomes with both field qubits and most-significant qubit of $x$-register equal to $0$. The simulations have been done for grid-size of $N_h = 32(n_h = 5)$ corresponding to total system of $n = 12$ qubits.

\begin{figure}[htbp]
    \centering
    \includegraphics[width=0.95\linewidth]{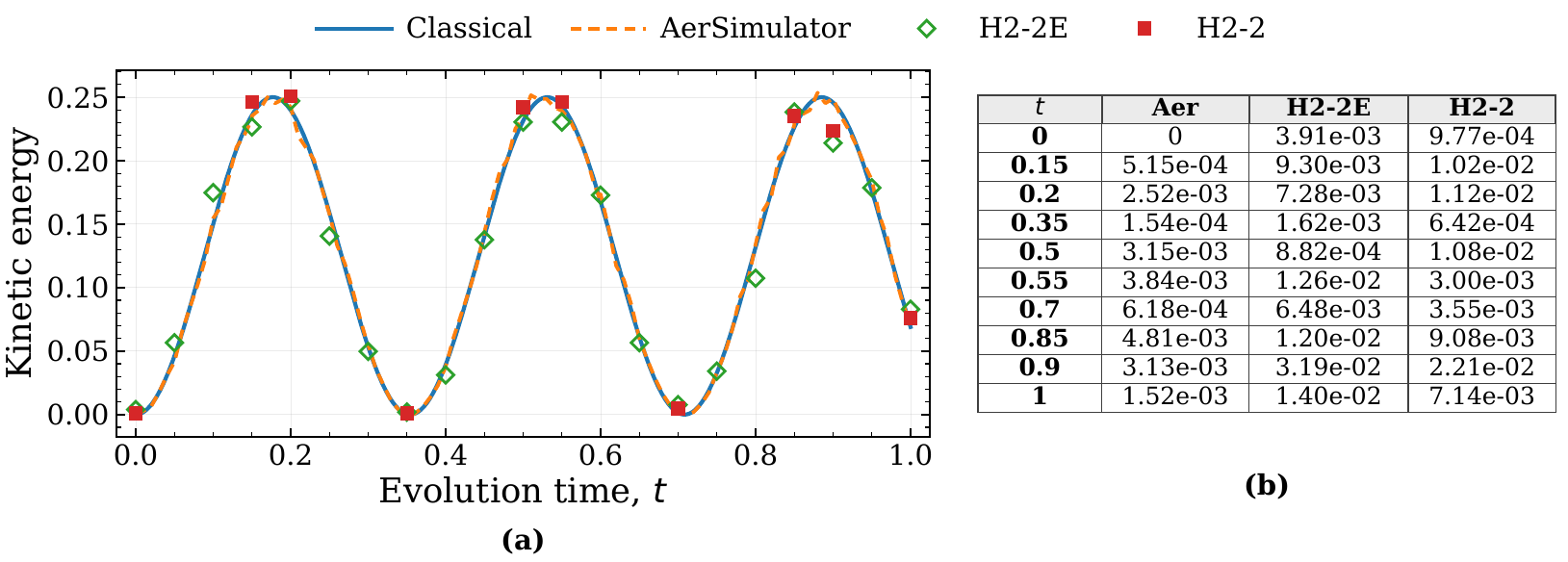}
\caption{Kinetic energy of the $v_x$ component on the half-domain
$x\in[0,\tfrac12)$ for the two-dimensional cosine pressure initial state
with $k_x=k_y=1$. The calculation uses a $32\times32$ spatial grid
($n_h=5$ qubits per direction and 12 qubits in total). (a) Classical
reference, noiseless finite-shot AerSimulator, H2-2E noise-modeled emulator,
and Quantinuum H2-2 results. (b) Absolute error relative to the classical
reference at each sampled time. The mean absolute errors are
$2.8\times10^{-3}$, $9.8\times10^{-3}$, and $7.9\times10^{-3}$ for
AerSimulator, H2-2E, and H2-2, respectively.}
    \label{fig:2d_cosine_ke}
\end{figure}

For the cosine intial pressure,  the mean absolute error averaged over its sampled points is
$2.8\times10^{-3}$ for AerSimulator, $9.8\times10^{-3}$ for H2-2E, and
$7.9\times10^{-3}$ for H2-2 - the same ordering, and similar $\sim3\times$ increase from the shot-noise reference to hardware noise,
as observed for the one-dimensional cosine state
(\cref{subsec:1d-hardware}). Against the kinetic energy's peak of
$\approx0.25$, this is a mean relative error of order $3\%$ on H2-2, with a
worst case of $2.21\times10^{-2}$ (relative error $\approx9\%$) at $t=0.9$.

\begin{figure}[htbp]
    \centering
    \includegraphics[width=0.95\linewidth]{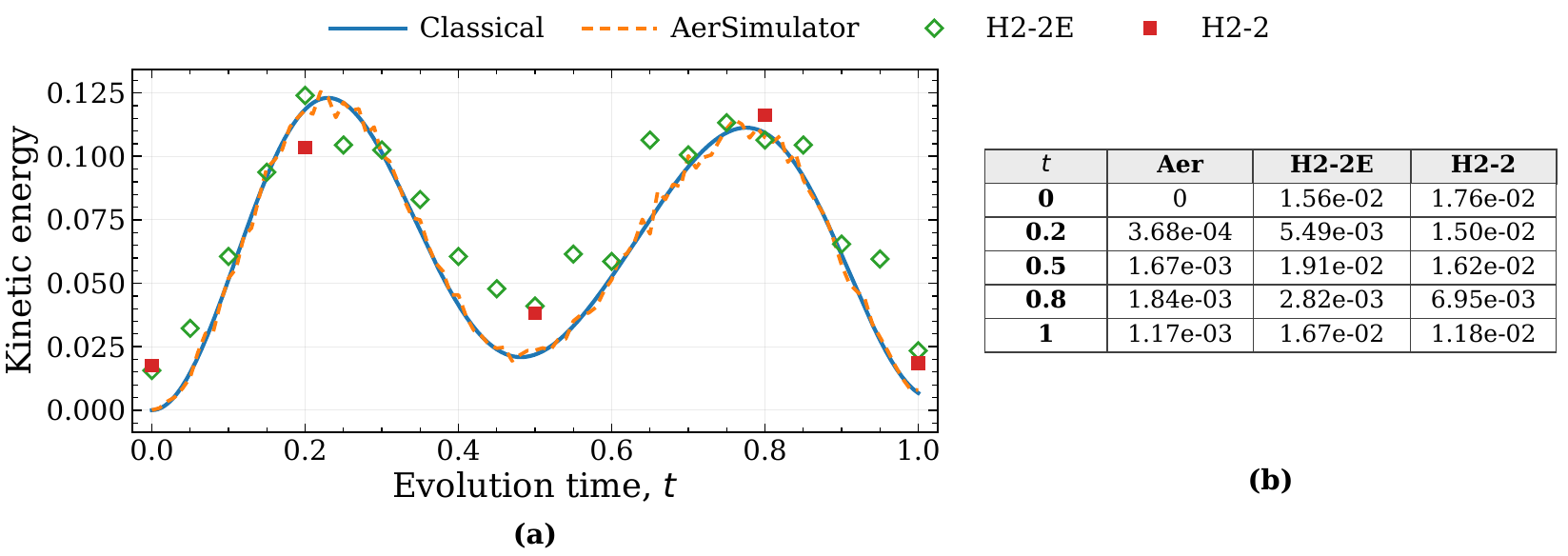}
\caption{Kinetic energy of the $v_x$ component on the half-domain
$x\in[0,\tfrac12)$ for the two-dimensional separable Gaussian pressure
initial state with $\sigma_x=\sigma_y=0.2$. The calculation uses a
$32\times32$ spatial grid ($n_h=5$ qubits per direction and 12 qubits in
total). (a) Classical reference, noiseless finite-shot AerSimulator, H2-2E
noise-modeled emulator, and Quantinuum H2-2 results. (b) Absolute error
relative to the classical reference at each sampled time. The mean absolute
errors are $2.3\times10^{-3}$, $2.5\times10^{-2}$, and $1.4\times10^{-2}$
for AerSimulator, H2-2E, and H2-2, respectively.}
    \label{fig:2d_Gaussian_ke}
\end{figure}

For the Gaussian initial flux (\cref{fig:2d_Gaussian_ke}), the mean absolute error is $2.3\times10^{-3}$ for
AerSimulator, $2.5\times10^{-2}$ for H2-2E, and $1.4\times10^{-2}$ for H2-2, respectively. As in the one-dimensional case, the
noise-modeled emulator's mean error ($2.5\times10^{-2}$) exceeds that of the physical device itself ($1.4\times10^{-2}$), again cautioning against reading the emulator's noise model as a upper bound on hardware performance.

\begin{figure}[htbp]
    \centering
    \includegraphics[width=0.95\linewidth]{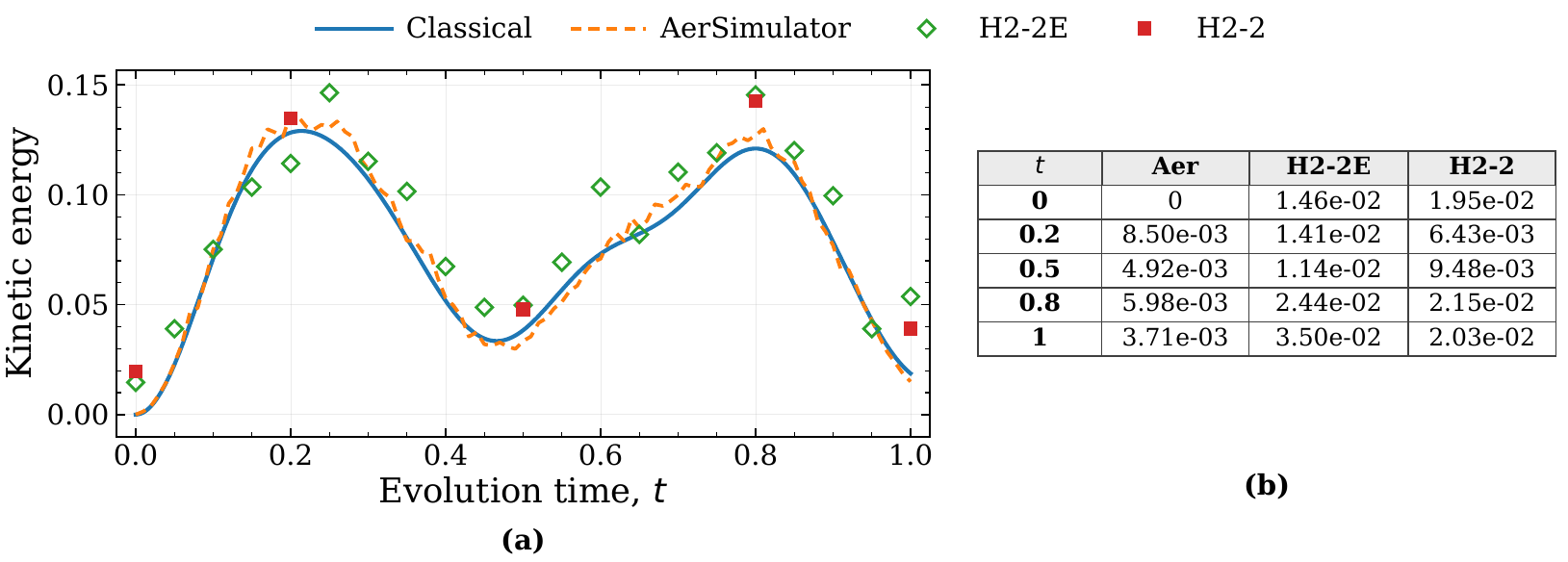}
\caption{Kinetic energy of the $v_x$ component on the half-domain
$x\in[0,\tfrac12)$ for the nonseparable pressure initial state in
\cref{eq:entangled-profile}, with $\kappa=1.0$ and $\gamma=0.4$. The
calculation uses a $32\times32$ spatial grid ($n_h=5$ qubits per direction
and 12 qubits in total). (a) Classical reference, noiseless finite-shot
AerSimulator, H2-2E noise-modeled emulator, and Quantinuum H2-2 results.
(b) Absolute error relative to the classical reference at each sampled time.
The mean absolute errors are $3.9\times10^{-3}$, $1.5\times10^{-2}$, and
$1.5\times10^{-2}$ for AerSimulator, H2-2E, and H2-2, respectively.}
    \label{fig:2d_entangled_ke}
\end{figure}

Finally, \cref{fig:2d_entangled_ke} reports the kinetic energy for the
nonseparable initial state of \cref{eq:entangled-profile}.
The mean absolute error is $3.9\times10^{-3}$ for AerSimulator,
$1.5\times10^{-2}$ for H2-2E, and $1.5\times10^{-2}$ for H2-2.
Against the kinetic energy's peak of $\approx0.13$, this is a mean relative
error of order $12\%$ on H2-2, comparable to the Gaussian state's and
noticeably worse than the cosine state's $3\%$, with a worst case of
$2.15\times10^{-2}$ (relative error $\approx17\%$) at $t=0.8$. Taken
together with the resource-estimation results above, the entangled state is
the most expensive of the three two-dimensional cases considered by every
measure.

\subsection{Dirac dynamics}

\subsubsection{Resource scaling}
\label{subsec:kg_resource}

For the cosine pressure initial state, we take $k_{0}=1$ and the step-function mass profile of
\cref{eq:kg-mass-profile} with $m_{-}=0$, $m_{+}=2$, and, following the same
protocol as \cref{subsec:1d-resources}. Because the Trotter step count $r=t/\tau$ in \cref{eq:strang} must grow
with $t$ to control the local $\tau^{3}$ error, our implementation uses a
coarse step schedule or Trotter number given by 
\begin{equation}
\label{eq:kg_step_schedule}
r(t)=
\begin{cases}
1, & 0\le t\le 0.1,\\
2, & 0.1<t\le 0.3,\\
3, & 0.3<t\le 0.5,\\
6, & 0.5<t\le 0.7,\\
7, & 0.7<t\le 1.
\end{cases}
\end{equation}
We came up with this coarse-step schedule by classically simulating Dirac dynamics for various $r$, in order to get better fidelity. The time-domain resource scan therefore probes both
the underlying circuit's dependence on $t$ and the discrete cost of this step
schedule simultaneously, while the grid-size scan is taken at $t=0.1$
($r=1$) so that it isolates the grid-size dependence alone. Results are
summarized in \cref{fig:kg_cosine_resources}.

\begin{figure}[htbp]
    \centering
    \includegraphics[width=0.95\linewidth]{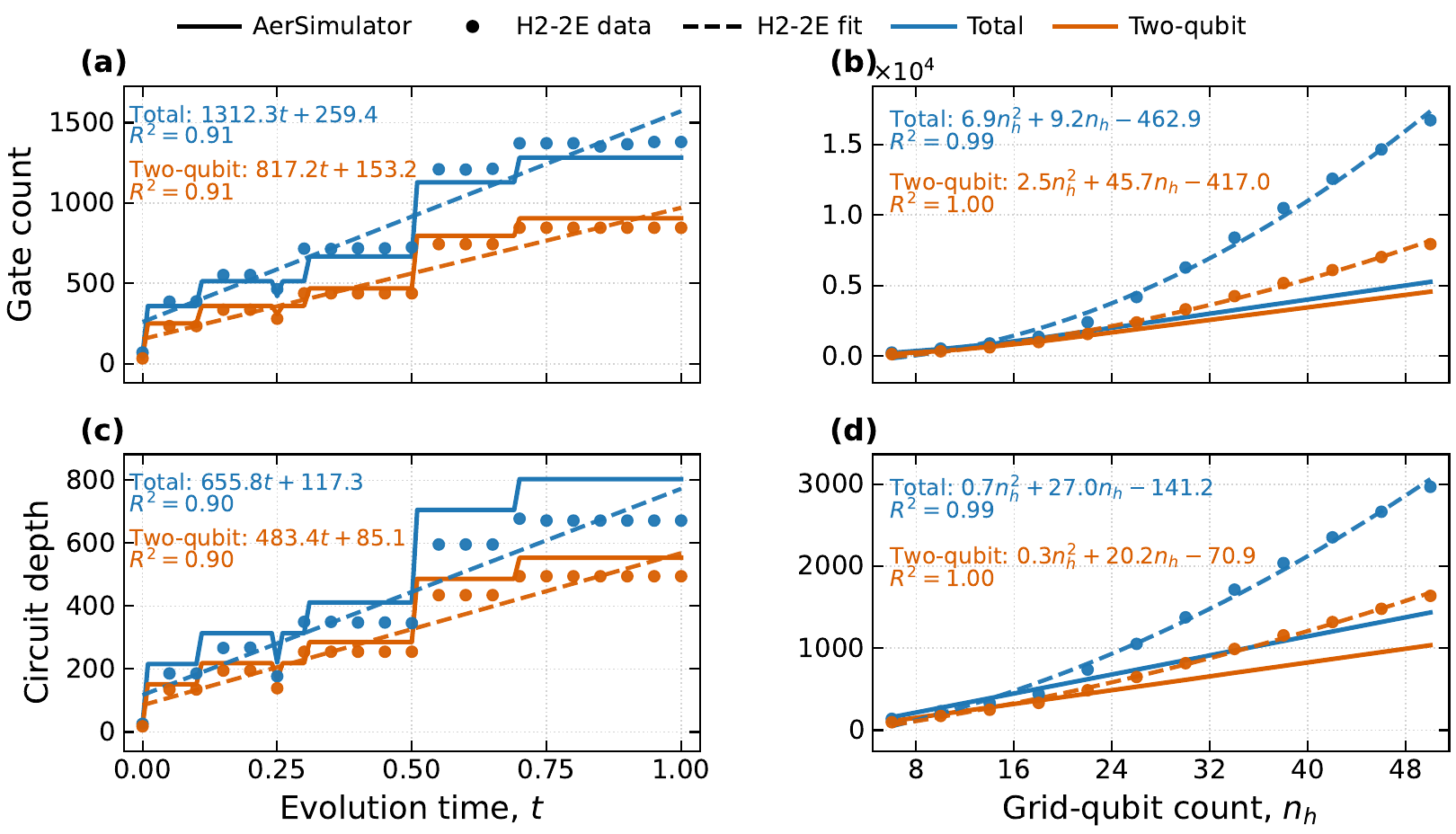}
\caption{Circuit resources for the one-dimensional linear Dirac dynamics 
 with cosine pressure initial data ($k_0=1$) and the stepwise mass
profile $m_-=0$, $m_+=2$. Solid curves show the logical circuit before
hardware-native compilation; markers show the circuit compiled for the
Quantinuum H2-2 native gate set using H2-2E, with dashed curves denoting
least-squares fits. Blue denotes total gate count or depth, and orange denotes
the corresponding two-qubit contribution. (a),(c) Gate count and circuit
depth, respectively, as functions of $t\in[0,1]$ at fixed $n_h=8$, using the
Strang-step schedule in \cref{eq:kg_step_schedule}. (b),(d) Gate count and
circuit depth as functions of $n_h$ at fixed $t=0.1$, for which one Strang
step is used.}
    \label{fig:kg_cosine_resources}
\end{figure}

As it can be observed in \cref{fig:kg_cosine_resources} (panels a and c), the compiled circuit's gate count and depth grow, on average, linearly with
$t$, and they fit $\approx1312\,t+259$ ($R^{2}=0.91$) for total gate count and
$\approx656\,t+117$ ($R^{2}=0.90$) for total depth,  However the underlying data follows a staircase structure due the coarse-step schedule \cref{eq:kg_step_schedule} for trotterization. This is in contrast with the 1D/2D case, whose resource cost mostly does not depend on $t$, the Dirac-dynamics circuit's
size is intrinsically coupled to the fidelity target through $r$

\smallskip 
For the grid-size scan, we summarize the results  in \cref{fig:kg_cosine_resources} (panels b and d), at fixed $t = 0.1 (r = 1)$ the compiled circuit (H2-2E), we found the fits for the total gate count $\approx6.9\,n_{h}^{2}+9.2\,n_{h}-462.9$
($R^{2}=0.99$) and total depth $\approx0.74\,n_{h}^{2}+27.0\,n_{h}-141.2$
($R^{2}=0.99$) - a strong quadratic scaling in comparison to 1D/2D case, even at the single-trotter-step
$r=1$ used here. This can be attributed to the fact that the mass sub-step's
position-controlled rotations must be interleaved with the wave sub-step's
QFT/rotation/inverse-QFT sequence rather than absorbed into it, and inverse-QFT corresponding to cosine state preparation, adding a
further layer of structure that the compiler cannot optimize in the
same manner as a single inverse-QFT.

\subsubsection{Hardware results}
\label{subsec:kg_hardware}

We evaluate the same kinetic-energy observable as in
\cref{subsec:1d-hardware}, \cref{eq:KE-def}, since the Dirac-dynamics circuit
retains the one-dimensional construction's velocity--flux field encoding
with the mass term added only through the Trotterized $H_{\mathrm{mass}}$
sub-step. \Cref{fig:kg_cosine_ke} shows the resulting kinetic energy as a
function of evolution time.

\begin{figure}[htbp]
    \centering
    \includegraphics[width=0.95\linewidth]{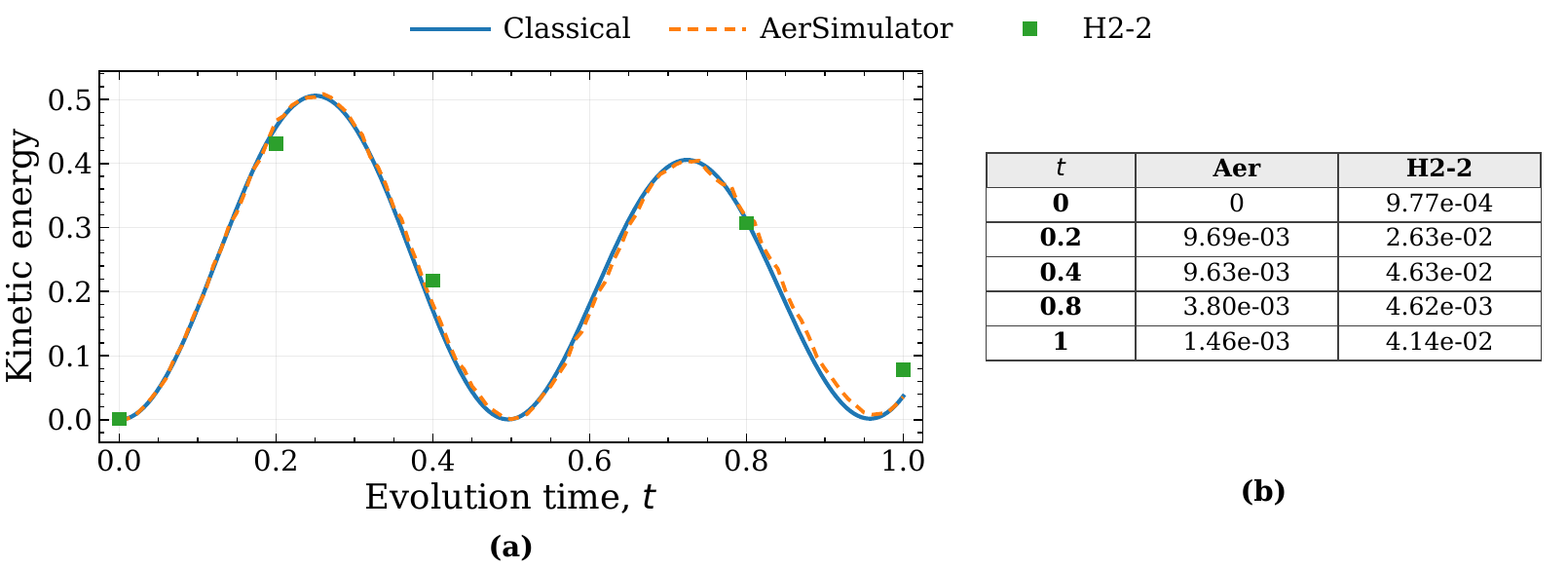}
\caption{Half-domain kinetic energy for the one-dimensional linear
Dirac dynamics equation with cosine pressure initial state ($k_0=1$), the
stepwise mass profile $m_-=0$, $m_+=2$, and $N_h=256$ grid points
($n_h=8$ grid qubits and 9 qubits in total). (a) Classical reference,
noiseless finite-shot AerSimulator, and Quantinuum H2-2 results.
(b) Absolute error relative to the classical reference at each sampled time.
The mean absolute errors are $7.3\times10^{-3}$ for AerSimulator and
$2.4\times10^{-2}$ for H2-2.}
    \label{fig:kg_cosine_ke}
\end{figure}

The mean absolute error is $7.3\times10^{-3}$ for AerSimulator and
$2.4\times10^{-2}$ for H2-2, against a kinetic-energy peak of
$\approx0.51$ - a mean relative error of order $5\%$ on H2-2, with a worst
case of $4.63\times10^{-2}$ (relative error $\approx9\%$) at $t=0.4$. This is relatively larger than 1D case(\cref{subsec:1d-hardware})
is consistent with the Dirac-dynamics circuit's substantially larger
compiled gate count and depth even at lower grid size
(\cref{subsec:kg_resource}).



\section{Discussion and outlook}
\label{sec:discussion}

We presented QFT-based circuits for one- and two-dimensional acoustic-wave propagation and for a one-dimensional split wave--mass model, using a common structure-preserving discretization and observable-estimation pipeline. In one dimension, Fourier diagonalization reduces each mode to a two-level rotation. In two dimensions, the bright--dark transformation identifies the longitudinal velocity component coupled to flux and reduces every retained Fourier block to an exact two-level pressure--bright rotation, eliminating the directional product-formula error that would otherwise require hundreds of repeated steps at modest accuracy. The wave and position-dependent mass terms are instead simple in different representations and are combined through Strang splitting. At fixed spectral bandwidth, the resulting acoustic circuits have resource costs polynomial in \(n_h=\log_2N_h\), rather than in the number of grid points, and changing the simulated acoustic time modifies rotation angles without increasing circuit depth.

The physical-device results are particularly encouraging. Across all tested initial conditions and spatial dimensions, the Quantinuum H2-2 device resolved the time dependence of the half-domain kinetic energy without full-state reconstruction. The mean absolute errors range from \(5.9\times10^{-3}\) to \(2.4\times10^{-2}\), despite the QFT layers, mode-controlled rotations, nonseparable state preparation, and, for the split wave--mass experiment, repeated changes between Fourier and position space. These results demonstrate the strong performance of Quantinuum hardware for the structured circuits considered here: after compilation to the H2-2 native gate set, the physical observables remain clearly distinguishable over the tested circuit depths. The H2-2E emulator provides a useful calibrated comparison, although its error does not consistently upper-bound that of the physical device.

Scaling real-hardware execution to larger $n_h$ will likely require noise-management techniques beyond compilation alone. Promising directions include quantum error-mitigation methods such as zero-noise extrapolation~\cite{temme2017error}, as well as error-suppression techniques such as dynamical decoupling~\cite{ezzell2023dynamical,watkins2026improvingddtrappedion}. Their effectiveness, however, depends strongly on the hardware noise characteristics, circuit structure, and observable of interest~\cite{cai2023quantum}. In particular, generic mitigation or suppression strategies can increase circuit duration, introduce additional control errors, or amplify statistical uncertainty, and may therefore fail to improve( or even degrade) the accuracy of the measured observable. It will consequently be important to tailor these techniques to the QFT, controlled-rotation, and measurement structures appearing in the present algorithms and to assess their benefits against their additional circuit and sampling costs. While fully fault-tolerant execution remains beyond near-term resource budgets at the scales considered here, lightweight error-detection schemes combined with syndrome-based postselection or symmetry verification~\cite{self2024protecting, perlin2026faulttolerantexecutionerrorcorrectedquantum}, together with classical correction of measurement errors~\cite{maciejewski2020mitigation}, provide another worthwhile direction toward progressively more reliable wave-equation simulations.

The scope of this conclusion is nevertheless deliberately limited. Our inputs have compact Fourier representations, and the measured quantity is a marginal probability that can be obtained directly from samples. Loading an arbitrary field, resolving a broadband spectrum, or reconstructing all \(N_h^d\) amplitudes could dominate the end-to-end cost. Consequently, the present experiments do not try to establish a quantum advantage over classical wave solvers. A demonstration we would have liked to include, but have not yet found a sufficiently shallow implementation for on current hardware, is a fully coherent two-dimensional circuit acting on a broad set of Fourier modes. Such a circuit would reversibly evaluate the finite-difference dispersion relation, the bright-mode frequency, and the Givens angle, rather than precomputing rotations for a fixed retained set. Extending the experiments to general pressure--velocity inputs and scalable state-preparation procedures presents a related challenge. These missing components mark the main gap between the present low-bandwidth benchmark and a general-purpose wave simulator.

A natural next step is wave propagation with nonperiodic boundary conditions, spatially varying coefficients, and irregular geometries. In these settings, the QFT no longer exactly diagonalizes the spatial operator, and structure-preserving finite-difference or finite-element discretizations must be considered. The resulting mass and stiffness matrices would require boundary-aware transforms, sparse block encodings, operator splitting, or other Hamiltonian-simulation constructions. An important objective is to retain the central advantages demonstrated here---Hermitian semi-discrete dynamics, efficient evolution of structured states, and direct estimation of coarse physical observables---while accommodating Dirichlet, Neumann, absorbing, and mixed boundary conditions.

Even when Fourier diagonalization remains applicable, the same framework suggests extensions to elastic wave \cite{xu2026hamiltonian} and Maxwell's equation \cite{jin2024quantummaxwell}, and beyond conservative wave propagation. For periodic constant-coefficient diffusion and advection--diffusion equations, the QFT separates the advective phase evolution from the mode-dependent diffusive attenuation. The latter is contractive and therefore cannot be implemented deterministically by a closed-system unitary on the solution register alone. Possible hardware realizations include ancilla dilations and postselection, or genuinely open-system circuits in which refreshed ancillas, qubit reset, or mid-circuit weak measurements generate the required dissipation. Determining how the circuit depth, sampling cost, postselection probability, and measurement-induced bias scale with the physical diffusion time will be central to such extensions. The present results therefore provide both a hardware benchmark for Fourier-structured hyperbolic equations and a starting point for quantum simulations of more general conservative and dissipative PDEs.

\begin{acknowledgments}
This work was supported by NSF Grant DMS-2411120, and Oak Ridge Leadership Computing Facility (OLCF) projects APM006, CSC539. 
\end{acknowledgments}

\appendix
\section{One-dimensional acoustic wave equation}
\subsection{Initial-state preparation}
\label{subsec:1d-initial-state}

We restrict attention to initial conditions with vanishing velocity,
$v(x,0)=0$, so that the complete initial state factorizes as
\begin{equation}
    \ket{\Psi(0)}=\ket{\psi_{p}}\ket{p},
\end{equation}
where $\ket{\psi_{p}}$ carries the Fourier coefficients of the initial pressure
field and $\ket{p}$ is the field register. It therefore suffices to
prepare $\ket{\psi_{p}}$ in the spatial register. We give two constructions: a cosine profile, and a gaussian profile

\subsubsection{Cosine initial condition}
\label{subsec:1d-cosine}

Consider the sampled cosine
\begin{equation}
    f_{j}=\cos\!\left(\frac{2\pi k_{0}j}{N_{h}}\right),
    \qquad j=0,\dots,N_{h}-1.
\end{equation}
Its discrete Fourier transform is supported on the two conjugate modes $k_{0}$
and $N_{h}-k_{0}$; for $k_{0}\notin\{0,N_{h}/2\}$ these are distinct and the normalized
Fourier-space state is the two-mode cat
\begin{equation}\label{eq:cos-target}
    \ket{\widehat f_{k_{0}}}
    =\frac{\ket{k_{0}}+\ket{N_{h}-k_{0}}}{\sqrt{2}}.
\end{equation}

Write the two target indices in binary as $a=(k_{0})_{2}$ and
$b=\bigl((N_{h}-k_{0})\bmod N_{h}\bigr)_{2}$, and let $d=a\oplus b$ record the bits in
which they differ; all three strings are computed classically. Pick any index
$r$ with $d_{r}=1$. A Hadamard on qubit $r$, followed by a CNOT fanout from $r$
to every other qubit $j$ with $d_{j}=1$, prepares the cat state
\begin{equation}
    \frac{\ket{0^{n_{h}}}+\ket{d}}{\sqrt{2}}.
\end{equation}
Applying the bitwise Pauli $X(a)=\bigotimes_{j=0}^{n_{h}-1}X_{j}^{a_{j}}$ then
translates both branches by $a$,
\begin{equation}
    X(a)\,\frac{\ket{0^{n_{h}}}+\ket{d}}{\sqrt{2}}
    =\frac{\ket{a}+\ket{a\oplus d}}{\sqrt{2}}
    =\frac{\ket{k_{0}}+\ket{N_{h}-k_{0}}}{\sqrt{2}},
\end{equation}
using $a\oplus d=a\oplus a\oplus b=b$. This reproduces~\eqref{eq:cos-target},
so the full initial state is
\begin{equation}
    \ket{\Psi_{\cos}(0)}
    =\frac{\ket{k_{0}}+\ket{N_{h}-k_{0}}}{\sqrt{2}}\,\ket{p}.
\end{equation}

The explicit circuit construction for \cref{eq:cos-target} for $k_0=1$ is shown in \cref{fig:cosine-circuit}

\subsubsection{Gaussian initial condition}
\label{subsec:1d-fivemode}

A spatially localized flux profile is naturally modeled by a Gaussian
centered at the midpoint of the periodic domain $[0,1)$,
\begin{equation}\label{eq:gauss-profile}
    f(x)=\exp\!\left(-\frac{\bigl(x-\tfrac12\bigr)^{2}}{2\sigma^{2}}\right),
    \qquad
    f_{j}=f(x_{j}),\quad x_{j}=\frac{j}{N_{h}}.
\end{equation}
Since $f$ is real and symmetric about $x=\tfrac12$, i.e.\ $f(x)=f(1-x)$, its
discrete Fourier coefficients (with the same convention as above,
$\widehat f_{k}=\tfrac1{N_{h}}\sum_{j}f_{j}\,e^{-2\pi ikj/N_{h}}$) are real and invariant
under $k\mapsto N_{h}-k$. When $\sigma$ is chosen so that $f$ is smooth across
$(0,1)$ and negligible at the boundaries, replacing the sum by its Poisson-summation integral gives the
closed form
\begin{equation}\label{eq:gauss-spectrum}
    \widehat f_{k}\approx\sqrt{2\pi}\,\sigma\,(-1)^{k}\,e^{-2\pi^{2}\sigma^{2}k^{2}},
    \qquad
    \widehat f_{N_{h}-k}=\widehat f_{-k}=\widehat f_{k},
    \qquad |k|\ll N_{h}.
\end{equation}
The spectral weight is therefore itself Gaussian in the mode index, with
$k$-space width $\sim 1/(2\pi\sigma)$: the smoother (broader) the profile, the
faster the decay and the fewer the dominant modes. The relative magnitude of
successive modes is
\begin{equation}\label{eq:gauss-ratio}
    \frac{\lvert\widehat f_{k}\rvert}{\lvert\widehat f_{0}\rvert}
    \approx e^{-2\pi^{2}\sigma^{2}k^{2}}.
\end{equation}
For a moderate width (e.g.\ $\sigma\gtrsim 0.15$, giving
$\lvert\widehat f_{3}\rvert/\lvert\widehat f_{0}\rvert\approx e^{-18\pi^{2}\sigma^{2}}\lesssim 2\times10^{-2}$),
essentially all of the weight is concentrated on $|k|\le 2$. Truncating to these
dominant modes together with their conjugates leaves
\begin{equation}
    \Kset=\{0,\,1,\,2,\,N_{h}-2,\,N_{h}-1\}.
\end{equation}

Imposing parity symmetry $k\mapsto N_{h}-k$ leaves three free real amplitudes,
\begin{equation}\label{eq:psiG}
    \ket{\psi_{G}}
    =\alpha_{0}\ket{0}
    +\alpha_{1}\bigl(\ket{1}+\ket{N_{h}-1}\bigr)
    +\alpha_{2}\bigl(\ket{2}+\ket{N_{h}-2}\bigr),
    \qquad
    |\alpha_{0}|^{2}+2|\alpha_{1}|^{2}+2|\alpha_{2}|^{2}=1.
\end{equation}

We first load the amplitudes onto a \emph{compact} set of low-index basis states
and then map the negative-frequency half into place with a sign-extension step.
Let $s=q_{n_{h}-1}$ be the sign qubit and $q_{1}q_{0}$ the two
least-significant magnitude qubits. Group the amplitudes by sign
branch,
\begin{equation}
    r_{+}=\sqrt{|\alpha_{0}|^{2}+|\alpha_{1}|^{2}+|\alpha_{2}|^{2}},
    \qquad
    r_{-}=\sqrt{|\alpha_{1}|^{2}+|\alpha_{2}|^{2}},
\end{equation}
so that $r_{+}$ and $r_{-}$ are the total weights on the $s=0$ and $s=1$
branches and $r_{+}^{2}+r_{-}^{2}=1$ by~\eqref{eq:psiG}. For real amplitudes the
required rotation angles are
\begin{equation}
    \theta_{1}=\operatorname{atan2}(r_{-},r_{+}),
    \qquad
    \theta_{2}=\operatorname{atan2}\!\bigl(\alpha_{2},\sqrt{\alpha_{0}^{2}+\alpha_{1}^{2}}\bigr),
    \qquad
    \theta_{3}=\operatorname{atan2}(\alpha_{1},\alpha_{0}),
    \qquad
    \theta_{4}=\operatorname{atan2}(\alpha_{1},\alpha_{2}).
\end{equation}

The state is then built as a rotation tree, each gate acting on
$\ket{0}$ through $\Ry(2\theta)\ket{0}=\cos\theta\ket{0}+\sin\theta\ket{1}$.
First, $\Ry^{(s)}(2\theta_{1})$ splits the sign branches with weights
$(r_{+},r_{-})$. On the $s=0$ branch, an open-controlled $\Ry^{(q_{1})}(2\theta_{2})$
separates $\ket{2}$ (weight $\alpha_{2}$) from the $\{\ket{0},\ket{1}\}$ block,
and a further $\Ry^{(q_{0})}(2\theta_{3})$ conditioned on $s=0,q_{1}=0$ resolves
that block into $\alpha_{0}\ket{0}+\alpha_{1}\ket{1}$. On the $s=1$ branch, a
single $\Ry^{(q_{0})}(2\theta_{4})$ produces $\alpha_{2}\ket{00}+\alpha_{1}\ket{01}$.
The result is the compact state
\begin{equation}\label{eq:psi-comp}
    \ket{\psi_{\mathrm{comp}}}
    =\alpha_{0}\ket{0_{s}00}
    +\alpha_{1}\ket{0_{s}01}
    +\alpha_{2}\ket{0_{s}10}
    +\alpha_{2}\ket{1_{s}00}
    +\alpha_{1}\ket{1_{s}01}.
\end{equation}

Sign extension propagates the sign qubit across the remaining qubits,
\begin{equation}
    S_{\mathrm{ext}}=\prod_{j=1}^{n_{h}-2}\operatorname{CX}_{s\rightarrow q_{j}},
\end{equation}
leaving the $s=0$ branch untouched while mapping the $s=1$ branch to the
top of the register,
\begin{equation}
    \ket{1_{s}00}\longmapsto\ket{1\cdots1\,10}=\ket{N_{h}-2},
    \qquad
    \ket{1_{s}01}\longmapsto\ket{1\cdots1\,11}=\ket{N_{h}-1}.
\end{equation}
This turns~\eqref{eq:psi-comp} into $\ket{\psi_{G}}$ of~\eqref{eq:psiG}, so
\begin{equation}
    \label{eq:gaussian-full}
    \ket{\Psi_{G}(0)}=\ket{\psi_{G}}\ket{p}.
\end{equation}

The explicit circuit construction for \cref{eq:gaussian-full} for $k_0=1$ is shown in \cref{fig:gaussian-circuit}

\section{Two-dimensional acoustic wave equation}
\subsection{Explicit synthesis for the retained low-mode set}
\label{subsec:restricted-implementation}

We now specialize to the retained one-dimensional mode set
\begin{equation}\label{eq:retained-Kset}
    \mathcal K_1=\{0,1,2,N_h-2,N_h-1\}.
\end{equation}
The two-dimensional retained set is \(\mathcal K=\mathcal K_1\times\mathcal K_1\).
This case is small enough that the nonlinear angles can be synthesized by
Boolean expansions in the retained-mode bits, avoiding a general arithmetic
circuit.

\subsubsection{Sign-magnitude representation}
\label{subsubsec:sign-mag}

For \(k\in\mathcal K_1\), write \(m=\min(k,N_h-k)\in\{0,1,2\}\).  Let \(s\) be the
sign bit, with \(s=0\) for \(k=0,1,2\) and \(s=1\) for \(k=N_h-2,N_h-1\), and let
\(k_1k_0\) be the two low-order bits of the Fourier index.  The required
magnitude bits \(m_1m_0\) are
\begin{equation}\label{eq:signmag-table}
\begin{array}{c|c|cc|cc}
    k & s & k_1 & k_0 & m_1 & m_0 \\ \hline
    0   & 0 & 0 & 0 & 0 & 0\\
    1   & 0 & 0 & 1 & 0 & 1\\
    2   & 0 & 1 & 0 & 1 & 0\\
    N_h-2 & 1 & 1 & 0 & 1 & 0\\
    N_h-1 & 1 & 1 & 1 & 0 & 1
\end{array}
\qquad
\begin{aligned}
    m_0&=k_0,\\
    m_1&=k_1\oplus(s\,k_0).
\end{aligned}
\end{equation}
The computation is reversible: copy \(k_0\) into \(m_0\), copy \(k_1\) into
\(m_1\), and then apply a Toffoli update
\(m_1\mathrel{\oplus=}(s\land k_0)\).  We write the magnitude bits for the two
spatial directions as
\begin{equation}
    x_i=(m_x)_i,
    \qquad
    y_i=(m_y)_i,
    \qquad i=0,1.
\end{equation}

\subsubsection{Bright-flux rotation angle}
\label{subsubsec:beta-synthesis}

Let
\begin{equation}\label{eq:A0A1A2}
    A_m=2N_h\sin\left(\frac{\pi m}{N}\right),
    \qquad
    A_0=0,
    \quad
    A_1=2N_h\sin\left(\frac{\pi}{N_h}\right),
    \quad
    A_2=2N_h\sin\left(\frac{2\pi}{N_h}\right).
\end{equation}
The pressure-bright rotation angle is
\begin{equation}
    \beta_{\bm k}=2t\Omega_{\bm k}=2t\sqrt{A_{m_x}^2+A_{m_y}^2}.
\end{equation}
Define
\begin{equation}\label{eq:beta-values}
    \beta_1=2A_1t,
    \qquad
    \beta_2=2A_2t,
    \qquad
    \beta_{11}=2t\sqrt{2A_1^2},
    \qquad
    \beta_{12}=2t\sqrt{A_1^2+A_2^2},
    \qquad
    \beta_{22}=2t\sqrt{2A_2^2}.
\end{equation}
On the valid magnitude subspace the angle has the Boolean expansion
\begin{equation}\label{eq:beta-boolean}
\begin{aligned}
    \beta_{\bm k}
    &=\beta_1(x_0+y_0)+\beta_2(x_1+y_1)  \\
    &\quad
    +(\beta_{11}-2\beta_1)x_0y_0
    +(\beta_{12}-\beta_1-\beta_2)(x_0y_1+x_1y_0)
    +(\beta_{22}-2\beta_2)x_1y_1 .
\end{aligned}
\end{equation}
Since all terms correspond to rotations about the same \(x\)-axis, they commute.
Thus \(R_X(\beta_{\bm k})\) can be implemented as a product of singly and doubly
controlled \(R_X\) gates carrying the coefficients in \cref{eq:beta-boolean}.

\subsubsection{Givens angle}
\label{subsubsec:alpha-synthesis}

The Givens rotation angle is
\begin{equation}
    \alpha_{\bm k}=2\eta_{\bm k}=2\operatorname{atan2}(A_{m_y},A_{m_x}).
\end{equation}
Let
\begin{equation}
    \gamma=\tan^{-1}\left(\frac{A_1}{A_2}\right).
\end{equation}
With \(\alpha_{00}=0\), the following Boolean expansion reproduces all retained
mode values:
\begin{equation}\label{eq:alpha-boolean}
    \alpha_{\bm k}
    =\pi(y_0+y_1)
    -\frac{\pi}{2}(x_0y_0+x_1y_1)
    -2\gamma x_0y_1
    -(\pi-2\gamma)x_1y_0 .
\end{equation}
Indeed,
\begin{equation}
    \alpha_{10}=\alpha_{20}=0,
    \quad
    \alpha_{01}=\alpha_{02}=\pi,
    \quad
    \alpha_{11}=\alpha_{22}=\frac{\pi}{2},
    \quad
    \alpha_{12}=\pi-2\gamma,
    \quad
    \alpha_{21}=2\gamma,
\end{equation}
where the subscripts denote \((m_x,m_y)\).  The gate \(G_{\bm k}\) is then
implemented as an open-controlled \(R_Y(\alpha_{\bm k})\) on \(f_0\), controlled
on \(f_1=0\):
\begin{equation}\label{eq:G-gate-boolean}
    G_{\bm k}
    =\ket0\!\bra0_{f_1}\otimes R_Y^{(f_0)}(\alpha_{\bm k})
    +\ket1\!\bra1_{f_1}\otimes I_{f_0}.
\end{equation}

\subsubsection{Phase transformation}
\label{subsubsec:phase-synthesis}

For \(k\neq0\), write \(\lambda_k=A_m e^{i\phi_k}\).  Let
\(\delta=\pi/N_h\).  The retained values are
\begin{equation}\label{eq:phase-values}
    \phi_1=\pi-\delta,
    \qquad
    \phi_2=\pi-2\delta,
    \qquad
    \phi_{N-2}=2\delta,
    \qquad
    \phi_{N-1}=\delta,
    \qquad
    \phi_0=0.
\end{equation}
In magnitude bits,
\begin{equation}\label{eq:phi-boolean}
    \phi_k
    =(\pi-\delta)m_0
    +(\pi-2\delta)m_1
    +(-\pi+2\delta)s\,m_0
    +(-\pi+4\delta)s\,m_1 .
\end{equation}
The two-dimensional phase factor separates as
\begin{equation}\label{eq:P-factorization}
    P_{\bm k}=P_{k_x}^{(x)}P_{k_y}^{(y)},
\end{equation}
where
\begin{equation}
    P_{k_x}^{(x)}=\operatorname{diag}(e^{i\phi_x},1,1,1),
    \qquad
    P_{k_y}^{(y)}=\operatorname{diag}(1,e^{i\phi_y},1,1).
\end{equation}
Using \(\operatorname{CP}(\phi)=\operatorname{diag}(1,1,1,e^{i\phi})\), these are
realized on the field register by
\begin{equation}\label{eq:P-as-CP}
    P_{k_x}^{(x)}=(X\otimes X)\operatorname{CP}(\phi_x)(X\otimes X),
    \qquad
    P_{k_y}^{(y)}=(X\otimes I)\operatorname{CP}(\phi_y)(X\otimes I).
\end{equation}
Each \(\operatorname{CP}(\phi_k)\) is synthesized from controlled phases
conditioned on \(m_0\), \(m_1\), \((s,m_0)\), and \((s,m_1)\), with coefficients
specified by \cref{eq:phi-boolean}.

\subsubsection{Complete restricted-mode block circuit}
\label{subsubsec:complete-restricted-block}

Let \(U_{\rm mag}\) compute the magnitude bits for both spatial registers into
work qubits, and define \(\mathcal P\), \(\mathcal G\), and \(\mathcal R_x\) by
the controlled gates described above.  The general-mode block implementation is
\begin{equation}\label{eq:general-bd-circuit}
    \boxed{
    U(t)
    =U_{\rm mag}^{\dagger}
      \mathcal P\mathcal G\mathcal R_x\mathcal G^{\dagger}\mathcal P^{\dagger}
      U_{\rm mag} .
    }
\end{equation}
For pressure-only input, the initial inverse transformation is unnecessary, so
\begin{equation}\label{eq:pressure-bd-circuit}
    \boxed{
    U(t)
    =U_{\rm mag}^{\dagger}
      \mathcal P\mathcal G\mathcal R_x
      U_{\rm mag} .
    }
\end{equation}

\subsection{Preparation of two-dimensional flux states}
\label{subsec:2d-initial-state}

We finally describe the preparation of the Fourier-space flux state used as
input to the circuit.  For a pressure-only initial condition
\(v_x(x,y,0)=v_y(x,y,0)=0\), the initial Fourier state has the form
\begin{equation}\label{eq:2d-init}
    \ket{\Psi_f(0)}
    =\ket{\psi_f}_{xy}\ket p,
    \qquad
    \ket{\psi_f}_{xy}
    =
      \sum_{k_x,k_y=0}^{N_h-1}\widehat f_{k_x,k_y}
      \ket{k_x}_x\ket{k_y}_y.
\end{equation}
If the initial data are prepared in physical space, one may apply the two
uncontrolled QFTs to obtain this state.  In the low-mode implementation, it is
often more efficient to prepare the retained Fourier state directly.

\paragraph{Separable profiles.}
If \(f(x,y)=g(x)h(y)\), then
\(\widehat f_{k_x,k_y}=\widehat g_{k_x}\widehat h_{k_y}\), and hence
\begin{equation}\label{eq:separable-fourier-state}
    \ket{\psi_f}_{xy}=\ket{\psi_g}_x\otimes\ket{\psi_h}_y.
\end{equation}
The preparation circuit factorizes as
\begin{equation}
    U_{\rm sep}=U_g^{(x)}\otimes U_h^{(y)},
\end{equation}
so the one-dimensional Fourier-state preparation circuits can be used
independently on the two registers. Note that for cosine initial condition with $k_x = k_y = 1$, we just have 4 fourier mode pairs, so that $x_0 = 1, x_1 = 0, y_0 = 1, y_1 = 0$ therefore in addition to separable initial state preparation, we get simplified hamiltonian circuit as well from \cref{eq:alpha-boolean} and \cref{eq:beta-boolean}.

\paragraph{Nonseparable profiles}
For a nonseparable initial profile, the Fourier state is generally entangled
across the two spatial registers.  Restricting to
\(\mathcal K_1=\{0,1,2,N_h-2,N_h-1\}\) in each direction, define the retained
spectral weight
\begin{equation}\label{eq:2d-retained-weight}
    w_{\mathcal K}
    =\frac{\sum_{k_x,k_y\in\mathcal K_1}|\widehat f_{k_x,k_y}|^2}
           {\sum_{k_x,k_y=0}^{N_h-1}|\widehat f_{k_x,k_y}|^2}.
\end{equation}
The normalized retained coefficient matrix is
\begin{equation}\label{eq:retained-C-matrix}
    C_{\mu\nu}
    =\frac{\widehat f_{k_\mu,k_\nu}}
           {\sqrt{\sum_{k_x,k_y\in\mathcal K_1}|\widehat f_{k_x,k_y}|^2}},
    \qquad
    k_\mu,k_\nu\in\mathcal K_1.
\end{equation}
Thus the retained normalized Fourier state is
\begin{equation}\label{eq:retained-state}
    \ket{\psi_{\mathcal K}}
    =\sum_{\mu,\nu=0}^{4}C_{\mu\nu}\ket{k_\mu}_x\ket{k_\nu}_y.
\end{equation}

\paragraph{Low Schmidt-rank approximation\cite{Araujo_2024}}
Let
\begin{equation}
    C=U\Sigma V^\dagger,
    \qquad
    \sigma_0\ge\sigma_1\ge\cdots,
\end{equation}
be a singular-value decomposition.  Then
\begin{equation}\label{eq:schmidt-state}
    \ket{\psi_{\mathcal K}}
    =\sum_r \sigma_r\ket{u_r}_x\ket{v_r}_y,
\end{equation}
where
\begin{equation}
    \ket{u_r}_x=\sum_\mu U_{\mu r}\ket{k_\mu}_x,
    \qquad
    \ket{v_r}_y=\sum_\nu V_{\nu r}^*\ket{k_\nu}_y.
\end{equation}
The normalized rank-\(R\) approximation is
\begin{equation}\label{eq:rankR}
    \ket{\psi^{(R)}}
    =\frac{1}{\sqrt{\sum_{r=0}^{R-1}\sigma_r^2}}
      \sum_{r=0}^{R-1}\sigma_r\ket{u_r}_x\ket{v_r}_y.
\end{equation}
Its fidelity with the retained state is
\begin{equation}\label{eq:fidelity-retained}
    F_R^{(\mathcal K)}=\sum_{r=0}^{R-1}\sigma_r^2,
\end{equation}
and its fidelity with the full Fourier state is
\begin{equation}\label{eq:fidelity-full}
    F_R^{(\rm full)}=w_{\mathcal K}F_R^{(\mathcal K)}.
\end{equation}
This separates the error due to Fourier truncation from the error due to the
Schmidt-rank approximation.

\paragraph{Rank-two preparation.}
For \(R=2\), define
\begin{equation}
    \alpha_r=\frac{\sigma_r}{\sqrt{\sigma_0^2+\sigma_1^2}},
    \qquad r=0,1.
\end{equation}
Then
\begin{equation}
    \ket{\psi^{(2)}}
    =\alpha_0\ket{u_0}_x\ket{v_0}_y
     +\alpha_1\ket{u_1}_x\ket{v_1}_y.
\end{equation}
On compact three-qubit registers \(R_\mu=(s_\mu,q_{\mu,1},q_{\mu,0})\), choose
seed states
\begin{equation}
    \ket{e_0}=\ket0_s\ket{00},
    \qquad
    \ket{e_1}=\ket1_s\ket{00}.
\end{equation}
Starting from \(\ket{e_0}_x\ket{e_0}_y\), let
\begin{equation}
    \vartheta=\operatorname{atan2}(\alpha_1,\alpha_0).
\end{equation}
Then
\begin{equation}\label{eq:label-state}
    \operatorname{CX}_{s_x\to s_y}R_Y^{(s_x)}(2\vartheta)
    \ket{e_0}_x\ket{e_0}_y
    =\alpha_0\ket{e_0}_x\ket{e_0}_y
     +\alpha_1\ket{e_1}_x\ket{e_1}_y
    \equiv \ket{\Phi_{\rm label}}.
\end{equation}
Local unitaries satisfying
\begin{equation}
    U_x\ket{e_r}=\ket{u_r},
    \qquad
    V_y\ket{e_r}=\ket{v_r}
\end{equation}
then give
\begin{equation}\label{eq:Uinit2}
    \boxed{
    U_{\rm init}^{(2)}
    =(U_x\otimes V_y)\operatorname{CX}_{s_x\to s_y}R_Y^{(s_x)}(2\vartheta).
    }
\end{equation}

\paragraph{Structured local Fourier-pair unitary.}
For parity-symmetric Schmidt vectors, the leading even and odd components often
have the form
\begin{equation}\label{eq:schmidt-parity}
    \ket{u_0}
    =a_0\ket0+a_1(\ket1+\ket{N_h-1})+a_2(\ket2+\ket{N_h-2}),
\end{equation}
\begin{equation}
    \ket{u_1}
    =b_1(\ket1-\ket{N_h-1})+b_2(\ket2-\ket{N_h-2}),
\end{equation}
with
\begin{equation}
    |a_0|^2+2|a_1|^2+2|a_2|^2=1,
    \qquad
    2|b_1|^2+2|b_2|^2=1.
\end{equation}
The corresponding local unitary can be factorized as
\begin{equation}\label{eq:Uloc}
    \boxed{U_{\rm loc}=S_{\rm ext}PM.}
\end{equation}
Here \(M\) prepares the magnitude distribution, \(P\) creates the symmetric or
antisymmetric signed pair, and \(S_{\rm ext}\) maps compact signed labels to the
full \(n\)-qubit Fourier labels.

For the magnitude preparation, let
\begin{equation}
    \ket{M_0}=\ket{00},
    \qquad
    \ket{M_1}=\ket{01},
    \qquad
    \ket{M_2}=\ket{10}.
\end{equation}
Define
\begin{equation}
    c_0=a_0,
    \qquad
    c_1=\sqrt2\,a_1,
    \qquad
    c_2=\sqrt2\,a_2,
    \qquad
    d_1=\sqrt2\,b_1,
    \qquad
    d_2=\sqrt2\,b_2.
\end{equation}
The target action is
\begin{equation}\label{eq:M-action}
    M\ket{e_0}
    =\ket0_s(c_0\ket{M_0}+c_1\ket{M_1}+c_2\ket{M_2}),
\end{equation}
\begin{equation}
    M\ket{e_1}
    =\ket1_s(d_1\ket{M_1}+d_2\ket{M_2}).
\end{equation}
This can be implemented by sign-multiplexed \(R_Y\) rotations.  First rotate
\(q_1\) by
\begin{equation}
    \theta_e=2\operatorname{atan2}\left(c_2,\sqrt{|c_0|^2+|c_1|^2}\right)
\end{equation}
on the even branch and by
\begin{equation}
    \theta_o=2\operatorname{atan2}(d_2,d_1)
\end{equation}
on the odd branch.  Then, conditioned on \(q_1=0\), rotate \(q_0\) by
\begin{equation}
    \phi_e=2\operatorname{atan2}(c_1,c_0)
\end{equation}
on the even branch and by \(\phi_o=\pi\) on the odd branch.

The signed-pair map is implemented on the valid magnitude subspace by applying
\begin{equation}
    \operatorname{CH}_{q_0\to s},
    \qquad
    \operatorname{CH}_{q_1\to s},
    \qquad
    X_{q_0}\operatorname{CCX}_{s,q_0\to q_1}X_{q_0}.
\end{equation}
With compact labels
\begin{equation}
    0\leftrightarrow 000,
    \quad
    +1\leftrightarrow 001,
    \quad
    +2\leftrightarrow 010,
    \quad
    -2\leftrightarrow 100,
    \quad
    -1\leftrightarrow 101,
\end{equation}
this yields
\begin{equation}\label{eq:PM-action}
\begin{aligned}
    PM\ket{e_0}
    &=a_0\ket{000}+a_1(\ket{001}+\ket{101})+a_2(\ket{010}+\ket{100}),\\
    PM\ket{e_1}
    &=b_1(\ket{001}-\ket{101})+b_2(\ket{010}-\ket{100}).
\end{aligned}
\end{equation}
Finally, the sign-extension circuit maps the compact negative labels into the
full Fourier register.  A nearest-neighbor CNOT chain
\begin{equation}
    S_{\rm ext}=\operatorname{CX}_{q_2\to q_1}\cdots
    \operatorname{CX}_{q_{n-1}\to q_{n-2}}
\end{equation}
gives
\begin{equation}\label{eq:signext}
    \ket{10\cdots00}\mapsto\ket{1\cdots1\,10}=\ket{N_h-2},
    \qquad
    \ket{10\cdots01}\mapsto\ket{1\cdots1\,11}=\ket{N_h-1}.
\end{equation}

\section{Circuit diagrams}
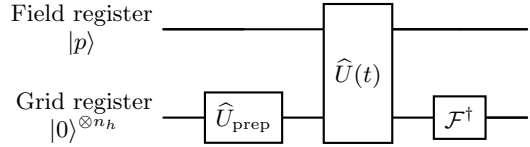
\begin{figure}[htbp]
    \centering
    \begin{quantikz}[
        row sep=0.5cm,
        column sep=0.55cm
    ]
        \lstick{Field register\\$\ket{p}$}
        & \qw
        & \gate[wires=2]{\widehat{U}(t)}
        & \qw
        & \qw
        \\
        \lstick{Grid register\\$\ket{0}^{\otimes n_h}$}
        & \gate{\widehat U_{\mathrm{prep}}}
        & {}
        & \gate{\mathcal{F}^{\dagger}}
        & \qw
    \end{quantikz}
\caption{Overall circuit for the one-dimensional acoustic-wave evolution in
\cref{eq:1d-overall-circuit}. The Fourier-space flux state is prepared
directly on the $n_h$-qubit grid register, evolved under
$\widehat U(t)$ together with the field qubit, and transformed back to the
spatial basis by the inverse quantum Fourier transform
$\mathcal F^\dagger$.}
    \label{fig:evolution-circuit}
\end{figure}

\begin{figure}[htbp]
    \centering
    \resizebox{\linewidth}{!}{%
    \begin{quantikz}[
        row sep=0.75cm,
        column sep=0.25cm
    ]
        \lstick{Field register\\$\ket{p}$}
        & \gate{\Rz(\frac{-\pi2^{0}}{2^{n_h}})}
        & \gate{\Rz(\frac{-\pi\cdot2^{1}}{2^{n_h}})}
        & \gate{\Rz(\frac{-\pi\cdot2^{n_h-1}}{2^{n_h}})}
        & \gate{\Had}
        & \gate{\Rz(-4\pi t2^{0})}
        & \gate{\Rz(-4\pi t2^{1})}
        & \gate{\Rz(-4\pi t2^{n_h-1})}
        & \gate{\Rz(-4\pi t2^{n_h})}
        & \gate{\Rz(8\pi t2^{0})}
        & \gate{\Rz(8\pi t2^{1})}
        & \gate{\Rz(8\pi t2^{n_h-1})}
        & \gate{\Had}
        & \gate{\Rz(\frac{\pi2^{0}}{2^{n_h}})}
        & \gate{\Rz(\frac{\pi2^{1}}{2^{n_h}})}
        & \gate{\Rz(\frac{\pi2^{n_h-1}}{2^{n_h}})}
        & \qw
        \\
        \lstick{$q_{0}$}
        & \ctrl{-1}
        & \qw
        & \qw
        & \qw
        & \ctrl{-1}
        & \qw
        & \qw
        & \qw
        & \ctrl{-1}
        & \qw
        & \qw
        & \qw
        & \ctrl{-1}
        & \qw
        & \qw
        & \qw
        \\
        \lstick{$q_{1}$}
        & \qw
        & \ctrl{-2}
        & \qw
        & \qw
        & \qw
        & \ctrl{-2}
        & \qw
        & \qw
        & \qw
        & \ctrl{-2}
        & \qw
        & \qw
        & \qw
        & \ctrl{-2}
        & \qw
        & \qw
        \\
        \lstick{$\vdots$}
        & {}
        & \gate[style={draw=none}]{\cdots}
        & {}
        & {}
        & {}
        & \gate[style={draw=none}]{\cdots}
        & {}
        & {}
        & {}
        & \gate[style={draw=none}]{\cdots}
        & {}
        & {}
        & {}
        & \gate[style={draw=none}]{\cdots}
        & {}
        & {}
        \\
        \lstick{$q_{n_h-1}$}
        & \qw
        & \qw
        & \ctrl{-4}
        & \qw
        & \qw
        & \qw
        & \ctrl{-4}
        & \ctrl{-4}
        & \ctrl{-4}
        & \ctrl{-4}
        & \ctrl{-4}
        & \qw
        & \qw
        & \qw
        & \ctrl{-4}
        & \qw
    \end{quantikz}%
    }
\caption{Gate-level decomposition of the approximate one-dimensional
Fourier-space propagator $\widehat U_{\mathrm{approx}}(t)$ in
\cref{eq:1d-fourier-app-ham}. Grid-controlled phase rotations implement the
mode-dependent basis transformation and the diagonal evolution generated by
$\widetilde H_Z$, while the most-significant grid qubit implements the folding
between positive and negative Fourier modes.}
    \label{fig:hz-tilde-circuit}
\end{figure}
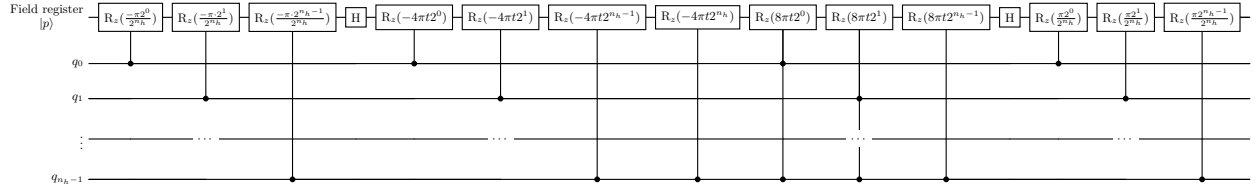

\begin{figure}[htbp]
    \centering
    \begin{quantikz}[
        row sep=0.6cm,
        column sep=0.4cm
    ]
        \lstick{Field register\\$\ket{p}$}
        & \gate{X}
        & \qw
        & \qw
        & \qw
        & \qw
        \\
        \lstick{$q_{0}$}
        & \gate{X}
        & \qw
        & \qw
        & \qw
        & \qw
        \\
        \lstick{$q_{1}$}
        & \gate{\Had}
        & \ctrl{1}
        & \qw
        & \qw
        & \qw
        \\
        \lstick{$\vdots$}
        & {}
        & \targ{}
        & \gate[style={draw=none}]{\cdots}
        & \ctrl{1}
        & {}
        \\
        \lstick{$q_{n_h-1}$}
        & \qw
        & \qw
        & \qw
        & \targ{}
        & \qw
    \end{quantikz}
\caption{Fourier-space state-preparation circuit for the one-dimensional
cosine pressure initial condition with $k_0=1$. The circuit prepares the two
conjugate Fourier modes
$\ket{\psi_p}=(\ket{1}+\ket{N_h-1})/\sqrt{2}$ on the grid register and sets
the field register to the pressure component. }
    \label{fig:cosine-circuit}
\end{figure}
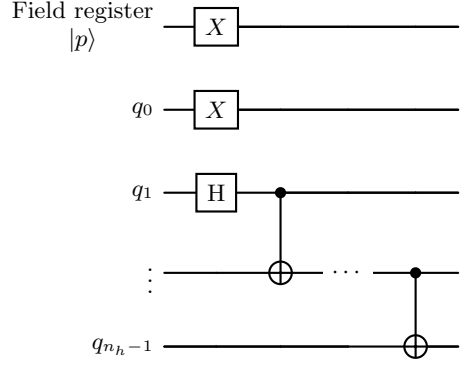

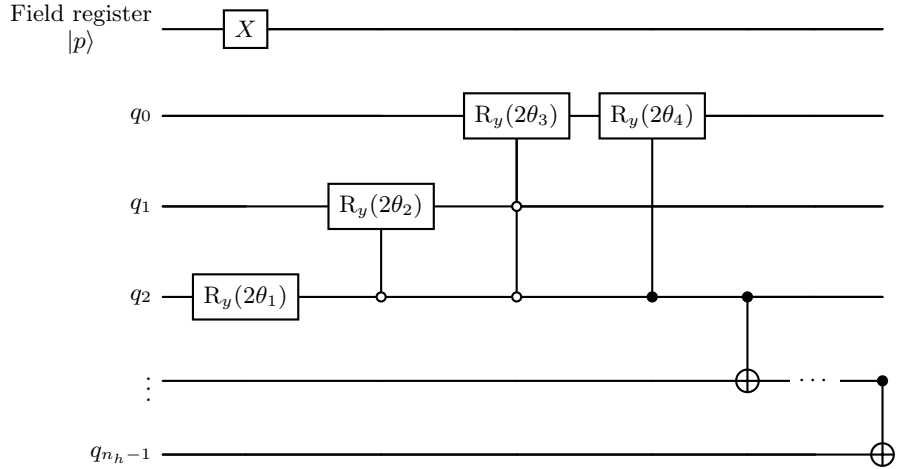
\begin{figure}[htbp]
    \centering
    \begin{quantikz}[
        row sep=0.6cm,
        column sep=0.4cm
    ]
        \lstick{Field register\\$\ket{p}$}
        & \gate{X}
        & \qw
        & \qw
        & \qw
        & \qw
        & \qw
        & \qw
        \\
        \lstick{$q_{0}$}
        & \qw
        & \qw
        & \gate{\Ry(2\theta_{3})}
        & \gate{\Ry(2\theta_{4})}
        & \qw
        & \qw
        & \qw
        \\
        \lstick{$q_{1}$}
        & \qw
        & \gate{\Ry(2\theta_{2})}
        & \octrl{-1}
        & \qw
        & \qw
        & \qw
        & \qw
        \\
        \lstick{$q_{2}$}
        & \gate{\Ry(2\theta_{1})}
        & \octrl{-1}
        & \octrl{-2}
        & \ctrl{-2}
        & \ctrl{1}
        & \qw
        & \qw
        \\
        \lstick{$\vdots$}
        & {}
        & {}
        & {}
        & {}
        & \targ{}
        & \gate[style={draw=none}]{\cdots}
        & \ctrl{1}
        \\
        \lstick{$q_{n_h-1}$}
        & \qw
        & \qw
        & \qw
        & \qw
        & \qw
        & \qw
        & \targ{}
    \end{quantikz}
\caption{Fourier-space state-preparation circuit for the one-dimensional
Gaussian pressure initial condition. The circuit prepares the five-mode
approximation supported on $k=0,\pm1,\pm2$; the rotation angles
$\theta_1,\ldots,\theta_4$ are determined by the normalized Fourier
coefficients defined in \cref{subsec:1d-fivemode}.}
    \label{fig:gaussian-circuit}
\end{figure}

\begin{figure}[htbp]
    \centering
    \begin{quantikz}[
        row sep=0.5cm,
        column sep=0.55cm
    ]
        \lstick{Field register\\$\ket{p}=\ket{10}$}
        & \qw
        & \gate[wires=3]{U_{\rm bd}^{p}(t)}
        & \qw
        & \qw
        \\
        \lstick{$X$ grid register\\$\ket{0}^{\otimes n_h}$}
        & \gate[wires=2]{\widehat U_{\mathrm{prep}}}
        & {}
        & \gate{\mathcal{F}^{\dagger}}
        & \qw
        \\
        \lstick{$Y$ grid register\\$\ket{0}^{\otimes n_h}$}
        & {}
        & {}
        & \gate{\mathcal{F}^{\dagger}}
        & \qw
    \end{quantikz}
\caption{Overall pressure-input circuit for the two-dimensional acoustic wave
equation. The two-qubit field register is initialized in
$\ket p=\ket{10}$, while $\widehat U_{\mathrm{prep}}$ prepares the initial
Fourier state jointly on the $X$ and $Y$ grid registers. The bright--dark
propagator $U_{\mathrm{bd}}^{p}(t)$ performs the Fourier-space evolution, after
which separate inverse quantum Fourier transforms return both spatial
registers to the position basis.}
    \label{fig:2d-evolution-circuit}
\end{figure}
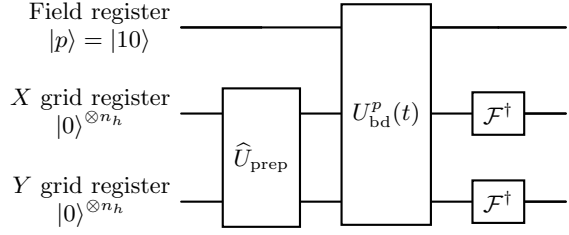

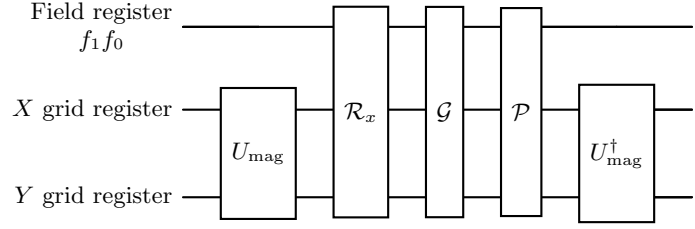
\begin{figure}[htbp]
    \centering
    \begin{quantikz}[
        row sep=0.5cm,
        column sep=0.5cm
    ]
        \lstick{Field register\\$f_{1}f_{0}$}
        & \qw
        & \gate[wires=3]{\mathcal R_{x}}
        & \gate[wires=3]{\mathcal G}
        & \gate[wires=3]{\mathcal P}
        & \qw
        & \qw
        \\
        \lstick{$X$ grid register}
        & \gate[wires=2]{U_{\mathrm{mag}}}
        & {}
        & {}
        & {}
        & \gate[wires=2]{U_{\mathrm{mag}}^{\dagger}}
        & \qw
        \\
        \lstick{$Y$ grid register}
        & {}
        & {}
        & {}
        & {}
        & {}
        & \qw
    \end{quantikz}
\caption{Pressure-specific bright--dark propagator in
\cref{eq:pressure-bd-circuit}. The reversible circuit $U_{\mathrm{mag}}$
computes the retained mode magnitudes, $\mathcal R_x$ performs the
bright--flux rotation, and $\mathcal G$ and $\mathcal P$ restore the
coordinate and phase conventions before the magnitude registers are
uncomputed.}
    \label{fig:Ubd-circuit}
\end{figure}

.

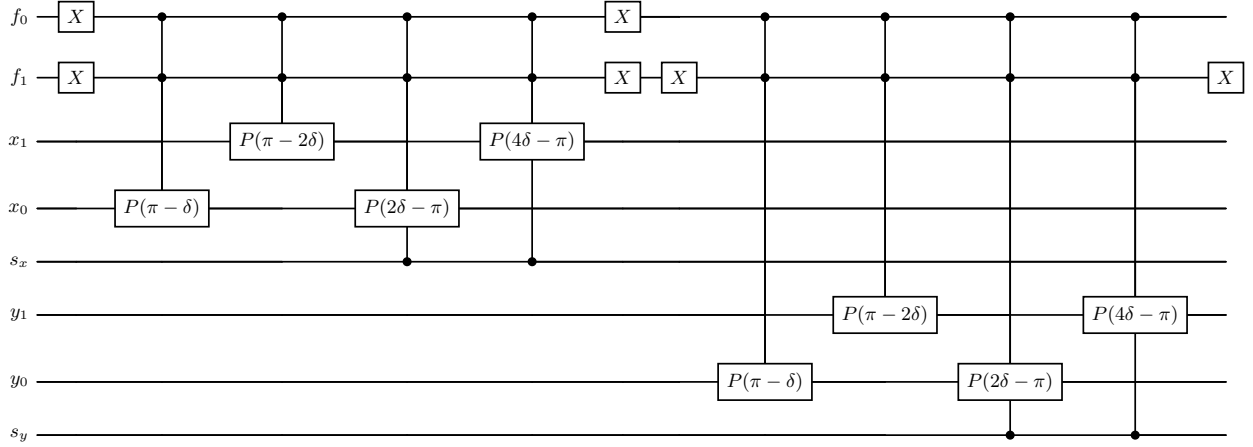
\begin{figure}[htbp]
    \centering
    \resizebox{\linewidth}{!}{%
    \begin{quantikz}[
        row sep=0.5cm,
        column sep=0.35cm
    ]
        \lstick{$f_{0}$}
        & \gate{X}
        & \ctrl{3}
        & \ctrl{2}
        & \ctrl{3}
        & \ctrl{2}
        & \gate{X}
        & \qw
        & \ctrl{6}
        & \ctrl{5}
        & \ctrl{6}
        & \ctrl{5}
        & \qw
        \\
        \lstick{$f_{1}$}
        & \gate{X}
        & \ctrl{2}
        & \ctrl{1}
        & \ctrl{2}
        & \ctrl{1}
        & \gate{X}
        & \gate{X}
        & \ctrl{5}
        & \ctrl{4}
        & \ctrl{5}
        & \ctrl{4}
        & \gate{X}
        \\
        \lstick{$x_{1}$}
        & \qw
        & \qw
        & \gate{P(\pi-2\delta)}
        & \qw
        & \gate{P(4\delta-\pi)}
        & \qw
        & \qw
        & \qw
        & \qw
        & \qw
        & \qw
        & \qw
        \\
        \lstick{$x_{0}$}
        & \qw
        & \gate{P(\pi-\delta)}
        & \qw
        & \gate{P(2\delta-\pi)}
        & \qw
        & \qw
        & \qw
        & \qw
        & \qw
        & \qw
        & \qw
        & \qw
        \\
        \lstick{$s_{x}$}
        & \qw
        & \qw
        & \qw
        & \ctrl{-1}
        & \ctrl{-2}
        & \qw
        & \qw
        & \qw
        & \qw
        & \qw
        & \qw
        & \qw
        \\
        \lstick{$y_{1}$}
        & \qw
        & \qw
        & \qw
        & \qw
        & \qw
        & \qw
        & \qw
        & \qw
        & \gate{P(\pi-2\delta)}
        & \qw
        & \gate{P(4\delta-\pi)}
        & \qw
        \\
        \lstick{$y_{0}$}
        & \qw
        & \qw
        & \qw
        & \qw
        & \qw
        & \qw
        & \qw
        & \gate{P(\pi-\delta)}
        & \qw
        & \gate{P(2\delta-\pi)}
        & \qw
        & \qw
        \\
        \lstick{$s_{y}$}
        & \qw
        & \qw
        & \qw
        & \qw
        & \qw
        & \qw
        & \qw
        & \qw
        & \qw
        & \ctrl{-1}
        & \ctrl{-2}
        & \qw
    \end{quantikz}%
    }
\caption{Gate-level decomposition of the two-dimensional phase transformation
$\mathcal P$ for the retained Fourier modes
$k_x,k_y\in\{0,\pm1,\pm2\}$. Controlled phase gates synthesize the separated
factors $P_{k_x}^{(x)}P_{k_y}^{(y)}$ in \cref{eq:P-factorization} from the
sign and magnitude bits, with $\delta=\pi/N_h$.}
    \label{fig:P-circuit}
\end{figure}

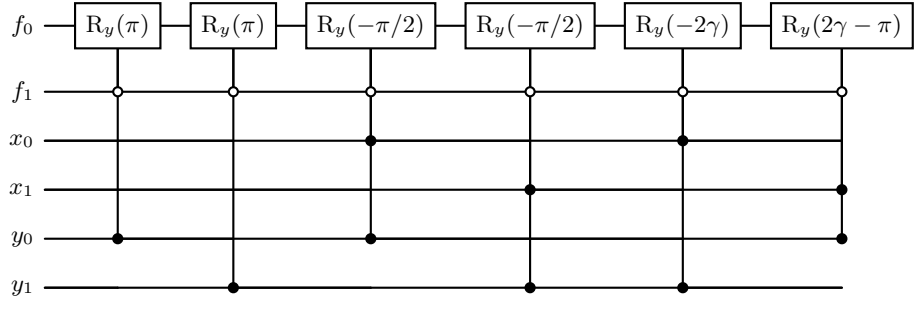
\begin{figure}[htbp]
    \centering
    \begin{quantikz}[
        row sep=0.5cm,
        column sep=0.4cm
    ]
        \lstick{$f_{0}$}
        & \gate{\Ry(\pi)}
        & \gate{\Ry(\pi)}
        & \gate{\Ry(-\pi/2)}
        & \gate{\Ry(-\pi/2)}
        & \gate{\Ry(-2\gamma)}
        & \gate{\Ry(2\gamma-\pi)}
        \\
        \lstick{$f_{1}$}
        & \octrl{-1}
        & \octrl{-1}
        & \octrl{-1}
        & \octrl{-1}
        & \octrl{-1}
        & \octrl{-1}
        \\
        \lstick{$x_{0}$}
        & \qw
        & \qw
        & \ctrl{-2}
        & \qw
        & \ctrl{-2}
        & \qw
        \\
        \lstick{$x_{1}$}
        & \qw
        & \qw
        & \qw
        & \ctrl{-3}
        & \qw
        & \ctrl{-3}
        \\
        \lstick{$y_{0}$}
        & \ctrl{-4}
        & \qw
        & \ctrl{-4}
        & \qw
        & \qw
        & \ctrl{-4}
        \\
        \lstick{$y_{1}$}
        & \qw
        & \ctrl{-5}
        & \qw
        & \ctrl{-5}
        & \ctrl{-5}
        & \qw
    \end{quantikz}
\caption{Gate-level decomposition of the Givens transformation
$\mathcal G$ for the retained two-dimensional Fourier modes. Open-controlled
$R_Y$ rotations on field qubit $f_0$, conditioned on $f_1=0$, implement the
mode-dependent angle $\alpha_{\bm k}$ in
\cref{eq:G-gate-boolean} and rotate the two velocity components into the
bright--dark basis.}
    \label{fig:G-circuit}
\end{figure}

\begin{figure}[htbp]
    \centering
    \resizebox{\linewidth}{!}{%
    \begin{quantikz}[
        row sep=0.5cm,
        column sep=0.4cm
    ]
        \lstick{$f_{0}$}
        & \qw
        & \qw
        & \qw
        & \qw
        & \qw
        & \qw
        & \qw
        & \qw
        \\
        \lstick{$f_{1}$}
        & \gate{\Rx(\beta_{1})}
        & \gate{\Rx(\beta_{1})}
        & \gate{\Rx(\beta_{2})}
        & \gate{\Rx(\beta_{2})}
        & \gate{\Rx(\beta_{11}-2\beta_{1})}
        & \gate{\Rx(\beta_{12}-\beta_{1}-\beta_{2})}
        & \gate{\Rx(\beta_{12}-\beta_{1}-\beta_{2})}
        & \gate{\Rx(\beta_{22}-2\beta_{2})}
        \\
        \lstick{$x_{0}$}
        & \ctrl{-1}
        & \qw
        & \qw
        & \qw
        & \ctrl{-1}
        & \ctrl{-1}
        & \qw
        & \qw
        \\
        \lstick{$x_{1}$}
        & \qw
        & \qw
        & \ctrl{-2}
        & \qw
        & \qw
        & \qw
        & \ctrl{-2}
        & \ctrl{-2}
        \\
        \lstick{$y_{0}$}
        & \qw
        & \ctrl{-3}
        & \qw
        & \qw
        & \ctrl{-3}
        & \qw
        & \ctrl{-3}
        & \qw
        \\
        \lstick{$y_{1}$}
        & \qw
        & \qw
        & \qw
        & \ctrl{-4}
        & \qw
        & \ctrl{-4}
        & \qw
        & \ctrl{-4}
    \end{quantikz}%
    }
\caption{Gate-level decomposition of the bright--flux rotation
$\mathcal R_x$ in \cref{eq:Rx-bright-pressure}. Singly and doubly controlled
$R_X$ rotations on field qubit $f_1$ synthesize the mode-dependent angle
$\beta_{\bm k}=2t\Omega_{\bm k}$ from the Boolean expansion in
\cref{eq:beta-boolean}.}
    \label{fig:Rx-circuit}
\end{figure}
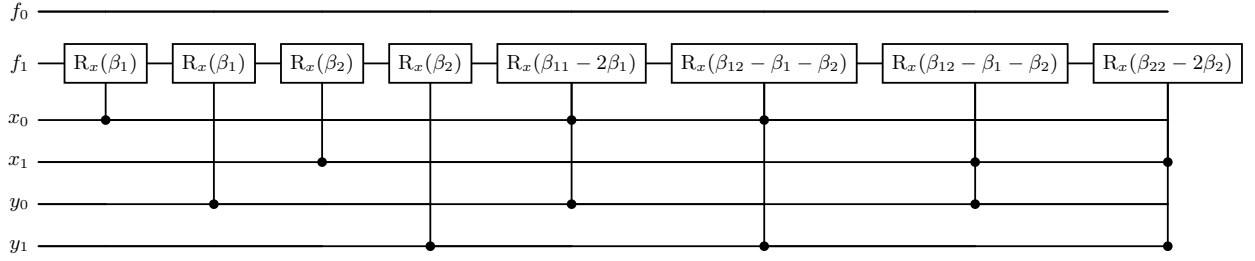

\begin{figure}[htbp]
    \centering
    \begin{quantikz}[
        row sep=0.5cm,
        column sep=0.5cm
    ]
        \lstick{$s_{x}$}
        & \gate{\Ry(2\vartheta)}
        & \ctrl{3}
        & \gate[wires=3]{U_{x}}
        & \qw
        \\
        \lstick{$q_{x_{1}}$}
        & \qw
        & \qw
        & {}
        & \qw
        \\
        \lstick{$q_{x_{0}}$}
        & \qw
        & \qw
        & {}
        & \qw
        \\
        \lstick{$s_{y}$}
        & \qw
        & \targ{}
        & \gate[wires=3]{V_{y}}
        & \qw
        \\
        \lstick{$q_{y_{1}}$}
        & \qw
        & \qw
        & {}
        & \qw
        \\
        \lstick{$q_{y_{0}}$}
        & \qw
        & \qw
        & {}
        & \qw
    \end{quantikz}
\caption{Rank-two preparation of a nonseparable two-dimensional Fourier
state. The rotation $R_Y(2\vartheta)$ encodes the two normalized Schmidt
weights, the controlled-NOT correlates the Schmidt labels of the $X$ and $Y$
registers, and the local unitaries $U_x$ and $V_y$ map these labels to the
corresponding Schmidt vectors, as in \cref{eq:Uinit2}.}
    \label{fig:entangled-init-circuit}
\end{figure}
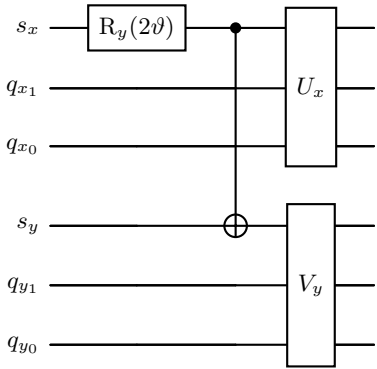

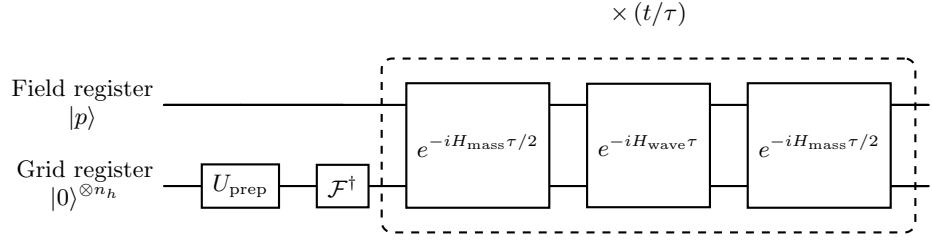
\begin{figure}[htbp]
    \centering
    \begin{quantikz}[
        row sep=0.5cm,
        column sep=0.5cm
    ]
        \lstick{Field register\\$\ket{p}$}
        & \qw
        & \qw
        & \gate[wires=2]{e^{-iH_{\rm mass}\tau/2}}
        \gategroup[wires=2,steps=3,style={dashed,rounded corners,inner sep=6pt},
        label style={label position=above,anchor=south,yshift=0.1cm}]{$\times\,(t/\tau)$}
        & \gate[wires=2]{e^{-iH_{\rm wave}\tau}}
        & \gate[wires=2]{e^{-iH_{\rm mass}\tau/2}}
        & \qw
        \\
        \lstick{Grid register\\$\ket{0}^{\otimes n_h}$}
        & \gate{U_{\mathrm{prep}}}
        & \gate{\mathcal F^{\dagger}}
        & {}
        & {}
        & {}
        & \qw
    \end{quantikz}
\caption{Circuit for the Strang-split evolution of the one-dimensional linear
Dirac dynamics equation. After preparation of the initial flux state, each
of the $r=t/\tau$ steps applies a half mass evolution, one acoustic-wave
evolution, and a second half mass evolution, corresponding to
\cref{eq:strang}. The mass operators act in the position basis, whereas the
wave block uses Fourier-space diagonalization.}
    \label{fig:kg-strang-circuit}
\end{figure}

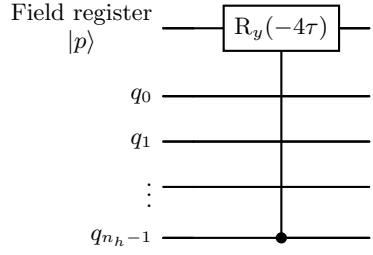
\begin{figure}[htbp]
    \centering
    \begin{quantikz}[
        row sep=0.6cm,
        column sep=0.4cm
    ]
        \lstick{Field register\\$\ket{p}$}
        & \qw
        & \gate{\Ry(-4\tau)}
        & \qw
        \\
        \lstick{$q_{0}$}
        & \qw
        & \qw
        & \qw
        \\
        \lstick{$q_{1}$}
        & \qw
        & \qw
        & \qw
        \\
        \lstick{$\vdots$}
        & {}
        & {}
        & {}
        \\
        \lstick{$q_{n_h-1}$}
        & \qw
        & \ctrl{-4}
        & \qw
    \end{quantikz}
\caption{Mass-evolution circuit for the stepwise profile
$m_-=0$, $m_+=2$. The most-significant grid qubit distinguishes the two
halves of the spatial domain and controls the rotation
$R_Y(-4\tau)$ on the field qubit; the left half, where $m_-=0$, is unchanged.}
    \label{fig:kg-mass-circuit}
\end{figure}

\FloatBarrier

\section{Glossary}
\label{sec:glossary}
\begin{table}[htbp]
    \centering
\caption{Notation used throughout the wave-equation formulations, circuit
constructions, and hardware experiments.}
    \label{tab:notation}
    \begin{tabular}{c l}
        \hline
        \textbf{Symbol}
        &
        \parbox{7cm}{\textbf{Description}}
        \\
        \hline

        \(v\)
        &
        \parbox[t]{7cm}{Velocity field}
        \\

        \(p\)
        &
        \parbox[t]{7cm}{flux field}
        \\

        \(D\)
        &
        \parbox[t]{7cm}{Forward finite-difference matrix}
        \\

        \(M\)
        &
        \parbox[t]{7cm}{Diagonal mass matrix}
        \\

        \(N\)
        &
        \parbox[t]{7cm}{System size, i.e., the total number of degrees of freedom}
        \\

        \(n\)
        &
        \parbox[t]{7cm}{Number of qubits used to represent the complete system}
        \\

        \(N_h\)
        &
        \parbox[t]{7cm}{Number of spatial grid points}
        \\

        \(n_h\)
        &
        \parbox[t]{7cm}{Number of grid qubits, where \(N_h = 2^{n_h}\)}
        \\

        \(f\)
        &
        \parbox[t]{7cm}{Field-qubit register distinguishing the velocity and flux components}
        \\

        \(H\)
        &
        \parbox[t]{7cm}{System Hamiltonian matrix}
        \\

        \(U(t)=e^{-iHt}\)
        &
        \parbox[t]{7cm}{Time-evolution operator generated by \(H\)}
        \\

        \(\widetilde{H}\)
        &
        \parbox[t]{7cm}{Approximate Hamiltonian matrix}
        \\

        \(\widehat{H}\)
        &
        \parbox[t]{7cm}{Hamiltonian matrix in the Fourier domain}
        \\

        \(\widehat{U}(t)=e^{-i\widehat{H}t}\)
        &
        \parbox[t]{7cm}{Time-evolution operator in the Fourier domain}
        \\

        \(U_{\mathrm{prep}}\)
        &
        \parbox[t]{7cm}{Initial-state preparation operator}
        \\

        \(\widehat{U}_{\mathrm{prep}}\)
        &
        \parbox[t]{7cm}{Initial-state preparation operator in the Fourier domain}
        \\

        \hline
    \end{tabular}
\end{table}

\bibliographystyle{plain}
\bibliography{refs}

\end{document}